\documentclass[a4paper,11pt]{article}

\usepackage[letterpaper,margin=1in]{geometry}

\usepackage{preamble-fv-simple-tmp}
\usepackage{macros}
\usepackage{multicol}
\usepackage{todonotes}
\usepackage{textpos}

\author[1]{Benjamin Bergougnoux} %%% add authors
\author[2]{Jan Dreier}
\author[1,3]{Lars Jaffke\thanks{Received funding from the Norwegian Research Council and from the European Research Council (ERC) under the European Union’s Horizon 2020 research and innovation programme (grant No. 714704).}}

\affil[1]{University of Warsaw, Poland}%%% add affiliations
\affil[ ]{\texttt{benjamin.bergougnoux@mimuw.edu.pl}} %%% add email addresses

\affil[2]{TU Wien, Austria}
\affil[ ]{\texttt{dreier@ac.tuwien.ac.at}}

\affil[3]{University of Bergen, Norway} 
\affil[ ]{\texttt{lars.jaffke@uib.no}} 

\title{A logic-based algorithmic meta-theorem for mim-width}

\newcommand\myparagraph[1]{\paragraph*{#1.}}

\newif\iflong
\longtrue

\begin{document}

\maketitle

\iffalse
\begin{textblock}{20}(-1.9, 6.6)
\includegraphics[width=40px]{logo-erc}%
\end{textblock}
\begin{textblock}{20}(-2.15, 6.8)
\includegraphics[width=60px]{logo-eu}%
\end{textblock}
\fi

%
\begin{abstract}
    We introduce a logic called \emph{distance neighborhood logic with acyclicity and connectivity constraints} (\ac \en for short)
    which extends existential $\mathsf{MSO_1}$
	with predicates for querying neighborhoods of vertex sets in various powers of a graph
    and for verifying connectivity and acyclicity of vertex sets.
    Building upon [Bergougnoux and Kant{\'e}, ESA 2019; SIDMA 2021],
    we show that the model checking problem for every fixed \ac \en formula is solvable in $n^{O(w)}$ time
	when the input graph is given together with a branch decomposition of mim-width $w$.
	Nearly all problems that are known to be solvable in polynomial time given a branch decomposition of constant mim-width
	can be expressed in this framework.
    We add several natural problems to this list, including problems asking for diverse sets of solutions.

	Our model checking algorithm 
	is efficient whenever the given branch decomposition of the input 
	graph has small index in terms of the \emph{$d$-neighborhood equivalence} [Bui-Xuan, Telle, and Vatshelle, TCS 2013].
    We therefore unify and extend known algorithms for tree-width, clique-width and rank-width.
    Our algorithm has a single-exponential dependence on these three width measures and
	asymptotically matches run times of the fastest known algorithms for several problems.
	This results in algorithms with tight run times under the Exponential Time Hypothesis (\ETH)
	for tree-width, clique-width and rank-width;
	the above mentioned run time for mim-width is nearly tight under the \ETH for several problems as well.
	Our results are also tight in terms of the expressive power of the logic:
	we show that already slight extensions of our logic make the model checking problem para-\NP-hard
	when parameterized by mim-width plus formula length.
\end{abstract}

\section{Introduction}
Bounded-width decompositions have been successful tools for dealing with the intractability of graph problems over the last decades.
Several landmark results
have shown that the power of such decompositions is closely linked to 
the expressibility of graph properties in some logic.
Perhaps the most famous examples are Courcelle's Theorem~\cite{CourcelleI},
stating that all problems expressible in \MSOtwo logic are linear time solvable
on graph classes of bounded tree-width
as well as a result by 
Courcelle, Makowsky, and Rotics~\cite{CourcelleMakowskyRotics2000},
showing that all problems expressible in \MSOone logic are linear time solvable
on graph classes of bounded clique-width (given a clique-width expression with a bounded number of labels).

If one considers the weaker first-order logic (\FO), then there are tractability results for a wider range of sparse and dense graph classes.
A long line of work~\cite{DawarEtAl2007,DvorakEtAl2013,FlumGrohe2001,FrickGrohe2001,Seese1996} 
culminated in the celebrated result by Grohe, Kreutzer, and Siebertz~\cite{GroheKreutzerSiebertz2017}
that problems expressible in first-order logic can be solved in almost linear time on nowhere dense graph classes.
For sparse graphs, this result is optimal under certain complexity theoretic assumptions~\cite{DvorakEtAl2013,Kreutzer2011}.
There have been efforts to transfer these results to dense graph classes~\cite{dreier2021lacon,gajarsky2016new,gajarsky2018first}.
A recent result by Bonnet, Kim, Thomass{\'{e}}, and Watrigant~\cite{BonnetEtAl2020}
shows that first-order problems can be solved in linear time on graph classes of bounded twin-width,
assuming a good twin-width decomposition is given.

For \emph{mim-width}, a relatively young width measure due to Vatshelle~\cite{VatshelleThesis},
such logic-based algorithmic meta-theorems have remained evasive,
even though algorithmic applications of mim-width have received considerable attention 
recently~\cite{BergougnouxKante2021,BergougnouxPapadopoulosTelle2020,BonomoEtAl2021,BrettelEtAl2020,BrettellHMP22,BrettelEtAl2020b,Bui-XuanTV2013,GalbyEtAl2020,JaffkeKwonStrommeTelle2019,JaffkeKwonTelle2020b,JaffkeKwonTelle2020}.
Informally speaking,
the mim-width of a graph bounds the size of any induced matching appearing in a cut in a recursive decomposition of 
the vertex set.
Such a recursive decomposition can be captured by a branch decomposition.
The strength of the mim-width parameter lies in its high expressive power.
While a bound on the tree-width or clique-width of a graph always implies a bound on its mim-width,
there are several well-studied graph classes such as interval and permutation graphs 
that have $n$-vertex graphs of clique-width $\Omega(\sqrt{n})$~\cite{GolumbicRotics2000} 
while their mim-width is bounded by a constant~\cite{BelmonteVatshelle2013}.
In fact, the \emph{linear} mim-width of interval and permutation graphs is at most~$1$.
Many problems are polynomial-time solvable when the input graph is given together with 
one of its branch decompositions of constant mim-width; 
we list them in \Cref{par:mim:probs}.
The existence of an \XP-algorithm for computing or approximating mim-width is wide open though.

In this work, we present a logic-based algorithmic meta-theorem for mim-width:
We show that every problem expressible in a logic we call \emph{distance neighborhood logic with acyclicity and connectivity constraints} (\ac \en logic for short)
can be solved in polynomial time on graph classes of bounded mim-width,
assuming an appropriate branch decomposition of the input graph is given.
The high expressive power of this logic, together with the generality of mim-width yields a powerful framework providing and unifying tractability results for a high number of problems for various graph classes.
A remarkable feature of our meta-theorem is that it has a
quite moderate run time dependence on the length of the input formula.
In stark contrast,
the run time of model checking algorithms for \FO and \MSO
often heavily depends on the formula length and the width parameter.
Even non-elementary tower functions $2^{2^{\cdot^{_{\cdot^{_\cdot}}}}}$, 
where the height depends on the length of the formula, 
are quite common and necessary~\cite{FrickG04}.
The efficiency of our \ac \en logic model checking is
shared with a modal logic introduced by Pilipczuk~\cite{Pilipczuk11},
whose model checking problem has a single-exponential algorithm for tree-width.
 
A wide range of problems is tractable on graphs of bounded mim-width (see \Cref{par:mim:probs}),
however, graph classes of bounded mim-width exhibit
some peculiar behavior when it comes to problems that are \emph{not} tractable.
For instance, while \textsc{Independent Set} is polynomial time solvable, 
\textsc{Clique} (i.e., \textsc{Independent Set} in the complement graph) 
is \NP-complete on graph classes of bounded mim-width.
Similarly, \textsc{Dominating Set} is polynomial time solvable, 
but \textsc{Co-Dominating Set} (\textsc{Dominating Set} in the complement graph) is again \NP-complete on graphs of bounded mim-width.
These hardness results follow from the fact that complements of planar graphs have linear mim-width at most~$6$, with corresponding decompositions being computable in polynomial time~\cite{BelmonteVatshelle2013,MatulaBeck1983}.

This behavior poses an extra challenge when it comes to finding the right logic for a meta-theorem for mim-width.
At first glance, it seems unreasonable to expect a clean logical formulation that separates 
the tractable from the intractable problems mentioned above.
We need a logic that can check whether \emph{some} vertices in a given set are connected (to capture \textsc{Independent set}) but cannot
check whether \emph{all} vertices in a given set are connected (to avoid capturing \textsc{Clique}).
We further want to be able to check whether \emph{all} vertices have a \emph{neighbor} in some set (to capture \textsc{Dominating Set})
but cannot be able to check whether \emph{all} vertices have a \emph{non-neighbor} in some set (to avoid capturing \textsc{Co-Dominating Set}).
Thus, we need a limited way to quantify over vertices that is not symmetric with respect to the edge relation, 
i.e., which has a different expressiveness when it comes to detecting edges and non-edges.

We solve this problem by basing our logic on \emph{neighborhood} operations.
For example by checking whether a set $U$ is disjoint from its neighborhood $N(U) = \bigcup_{u \in U} N(u)$, 
we can infer whether $U$ forms an independent set.
While a vertex is in the neighborhood of $U$ if it is adjacent to \emph{some} vertex in $U$, in order for $U$ to be a clique, vertices from $U$ need to be adjacent to \emph{all} other vertices of $U$.
Thus, we cannot infer anything about whether $U$ forms a clique by considering its neighborhood.
Similarly, $U$ is a dominating set if every vertex is contained in $U \cup N(U)$, while co-dominating sets again cannot be expressed in terms of neighborhoods.

Before we proceed with a description of our results,
we need to discuss one more feature of the mim-width parameter,
namely its surprising robustness against the operation of taking powers of graphs.
Jaffke, Kwon, Str{\o}mme, and Telle~\cite{JaffkeKwonStrommeTelle2019} showed that, 
independently of the value of~$r$,
the mim-width of the $r$-th power of any graph $G$ is at most twice the mim-width of $G$;
which is witnessed by the same branch decomposition.
Besides (distance-$1$) neighborhoods, 
we therefore also allow neighborhoods in the $r$-th power of a graph as basic building blocks of our logic, where $r$ may be as large as $n$.

\subsection{Distance neighborhood logic}
\emph{Distance neighborhood logic} is obtained by extending existential \MSOone with the additional ability to directly reason about neighborhoods of sets.
The central building block in our neighborhood logic is the neighborhood operator $N_d^r(\cdot)$.
For a subset of vertices $U$ of some graph $G$, 
we define $N_d^r(U)$ to be all vertices in $G$ with distance at most $r$ and at least $1$ 
to at least $d$ vertices in $U$.
We use this operator to define so-called \emph{neighborhood terms}.
These are built from set-variables, set-constants (i.e., unary relations) or other neighborhood terms
by applying a neighborhood operator $N_d^r(\cdot)$ or standard set operations
such as intersection
($\cap$), union ($\cup)$, subtraction ($\setminus$) or complementation (denoted by a bar on top of the term).
The \emph{neighborhood logic}, denoted by \en, is the extension of existential \MSOone by allowing the following.
\begin{itemize}
    \item \emph{Size measurement of terms}: We can write for example $|t| \le m$ 
		to express that a neighborhood term $t$ should have size at most $m$.
    \item \emph{Comparison between terms}: We can write for example $t_1 = t_2$ or $t_1 \subseteq t_2$ 
		to express that a neighborhood term $t_1$ should be equal to or contained in another term $t_2$.
\end{itemize}

We proceed by giving some examples of \en formulas and explaining what properties they express.
For the precise definition of neighborhood terms and \en we refer to \Cref{sec:logic}.
The very simple formula
$N_1^1(X) \cap X = \emptyset$
expresses that no vertex from $X$ intersects the neighborhood of $X$, i.e., that $X$ forms an independent set.
Thus, the \en formula
\[
\exists X \, |X| \ge m \land N_1^1(X) \cap X = \emptyset
\]
expresses the existence of an independent set of size at least $m$.
Similarly, the formula
\[
\exists X\, |X| \le m \land \mathbf{P} \subseteq X \cup N_3^r(X)
\]
expresses the existence of a set of size at most $m$ that 
distance-$r$ $3$-dominates all vertices labeled $\mathbf{P}$.
That is, a vertex set $X$ such that 
each vertex labeled $\mathbf{P}$ that is not in $X$ 
has at least three vertices at distance at most $r$ in $X$.
A \emph{semitotal dominating set} is a dominating set $S$ 
such that for each vertex $v \in S$ there is another vertex in $S$ that is at distance at most two from $v$.
The following \en formula expresses that a graph has a semitotal dominating set of size at most $m$:
\begin{align}
	\label{eq:semitotal:dnlog}
\exists X\, |X| \le m \land \overline{X \cup N^1_1(X)} = \emptyset \land X \subseteq N^2_1(X)
\end{align}
The following more advanced \en formula expresses that a graph has a dominating odd cycle transversal of size at most $m$ 
(where an odd cycle transversal of a graph $G$ is a vertex set whose removal results in a bipartite graph):
\begin{align*}
	\exists X \exists Y \exists Z &~\overline{X \cup Y \cup Z} = \emptyset 
		\land X \cap Y = \emptyset \land X \cap Z = \emptyset \land Y \cap Z = \emptyset \\
		\land &~\card{X} \le m \land \overline{X \cup N^1_1(X)} = \emptyset \\
		\land &~N^1_1(Y) \cap Y = \emptyset \land N^1_1(Z) \cap Z = \emptyset
\end{align*}
Note that the value of $m$ in the size measurements contribute exactly one unit to the length of a formula.
This means the length $|\phi|$ of a formula $\phi$ is independent of the actual values used in the size measurements of $\phi$.
Since each symbol $N^r_d$ also only contributes one unit to the length of a formula $\phi$, 
we associate the following quantities with each \en formula $\phi$, which will ultimately influence the run time 
of our model checking algorithms:
\begin{itemize}
    \item We let $d(\phi)$ be the largest value $d$ such that $N^{\cdot}_{d}(\cdot)$ occurs in $\phi$ (but at least two).
    \item We let $r(\phi)$ be the largest value $r$ such that $N^{r}_{\cdot}(\cdot)$ occurs in $\phi$ (but at least one).
\end{itemize}
We would like to remark that when expressing properties with a \en formula $\phi$, the quantities 
$d(\phi)$ and $r(\phi)$ are often bounded by small constants.
For a precise definition of our logics see \Cref{sec:logic}.

\myparagraph{Relation to modal logic}
Another way of motivating that \en logic is the ``right'' logic for mim-width is as follows.
Similar to tree-width, where almost all efficient algorithms perform some kind of dynamic programming on the tree decomposition,
many efficient algorithms for mim-width (e.g.,~\cite{BergougnouxKante2021,BergougnouxPapadopoulosTelle2020,Bui-XuanTV2013,GalbyEtAl2020,JaffkeKwonStrommeTelle2019})
rely on a bounded index\footnote{The number of equivalence classes.} 
for the $d$-neighborhood equivalence~\cite{Bui-XuanTV2013}.
We would like to find a logic, such that for every fixed formula of this logic
one can use the $d$-neighborhood equivalence to
decide in $n^{\calO(w)}$ time whether an $n$-vertex graph given with a decomposition of mim-width $w$
satisfies the formula.
Since the tree-width bounds the $d$-neighborhood equivalence,
such an algorithm has a run time of $2^{\calO(\tw)}n^{\calO(1)}$, where $\tw$ denotes the tree-width of the input graph.
We therefore have to restrict our search to logics that have a model-checking algorithm with single-exponential run time dependence on tree-width.
The modal logic \ecml presented by Pilipczuk~\cite{Pilipczuk11} has this property
and therefore is a good starting point.
 
However, \ecml is general enough to express \textsc{Maximum Cut} by a constant-length formula (\Cref{sec:hardness}).
By similar arguments as above, the logic we want to find needs to have an \FPT model-checking algorithm parameterized by clique-width (or rank-width) plus formula length
and therefore is not allowed to be able to express \textsc{Maximum Cut} by a constant length formula.
To avoid capturing \textsc{Maximum Cut}, we restrict our attention to the fragment $\mathsf{ECML_1}$ of \ecml that forbids quantification over edge sets.
Lastly, $\mathsf{ECML_1}$ still allows ultimately periodic counting.
We show that model-checking for $\mathsf{ECML_1}$ remains
para-\NP-hard parameterized by (linear) mim-width plus formula length, even if we only allow counting with period two, i.e., \emph{parity} counting (\Cref{sec:hardness}).
We therefore have to disregard periodic counting as well and arrive at a final fragment $\mathsf{ECML_1'}$.
We show in \Cref{sec:modal} that \en has the same expressiveness as a modal logic that extends $\mathsf{ECML_1'}$
with the ability to query neighborhoods in arbitrary powers of the graph.
It comes as a surprise that two seemingly unrelated mechanisms --- neighborhood operators and modal logic --- have the same expressive power
to describe tractable problems on graphs with bounded mim-width.

\subsection{Results}
Our first main result 
is that all problems that can be expressed in \en logic can be solved in \XP 
time parameterized by the mim-width of a given decomposition of the input graph.
Rather than the mim-width itself, the run time is bounded in terms of the index
of the \emph{$d$-neighborhood equivalence} relation~\cite{Bui-XuanTV2013} of the cuts appearing in the decomposition.
This quantity is bounded by an \XP function of the mim-width, 
and by \FPT functions of the tree-width, clique-width, or rank-width of corresponding decompositions,
and therefore our algorithm has interesting consequences for these parameterizations as well.\footnote{%
We would like to remark that in the following theorem, to obtain the bound for tree-width, we can drop the requirement of being given a 
decomposition, since tree-width $w$ can be $2$-approximated in time $2^{\calO(w)}n$~\cite{Korhonen2021}.
Similarly, we do not require the decomposition for rank-width $w$, 
since it can be $3$-approximated in time $2^{\calO(w)}n^{4}$~\cite{Oum2008}.
}
We would like to point out that even though the number of size measurements $s$ appears in the degree of the polynomial 
for tree-width and clique-width, this does not mean that we only have \XP-algorithms for these width measures.
In many problems, including the (optimization variants of the) \MSO-based meta-theorems~\cite{CourcelleI,CourcelleMakowskyRotics2000},
we have that $s \le 1$, as this already allows us to ask for instance for a solution of size at least or at most some given threshold.
See \Cref{tab:compare:logics}.
\begin{theorem}\label{thm:main:basic}
	There is an algorithm that determines for a given \en formula $\phi$, 
	and graph $G$ on $n$ vertices together with an \fwidth decomposition of $G$ of width $w \ge 1$,
	whether $\phi$ holds on $G$.
    When $r \defeq r(\phi)$, $d \defeq d(\phi)$ and $s$ is the number of size measurements in $\phi$,
    the run time of this algorithm is
	\begin{itemize}
		\item $n^{\calO(dw|\phi|^2)}$ if \fwidth is mim-width,
		\item $2^{\calO(d(wr|\phi|)^2)} n^{\calO(s)}$ (and $2^{\calO(dw|\phi|)} n^{\calO(s)}$ if $r = 1$), 
			if \fwidth is tree-width or clique-width,
		\item $2^{\calO(dw^4(r|\phi|)^2)} n^{\calO(s)}$ (and $2^{\calO(dw^2|\phi|)} n^{\calO(s)}$ if $r = 1$), 
			if \fwidth is rank-width.
	\end{itemize}
\end{theorem}

The previous theorem already reproves and generalizes several results from the literature in a unified framework.
In particular, all locally checkable vertex problems~\cite{Bui-XuanTV2013,TelleProskurowski1997} 
as well as their complements and distance-versions~\cite{JaffkeKwonStrommeTelle2019}
are expressible in \en logic.

\myparagraph{Connectivity and acyclicity, and optimization}
It is known that \textsc{Graph Connectivity} cannot be expressed in existential \MSOone logic, 
even when the logic is extended with several built-in relations~\cite{Fagin1975,FaginEtAl1995,Schwentick1994}.
However, it can be expressed by the following \en sentence for $n$-vertex graphs:
\[\exists X~\card{X} = 1 \wedge \overline{X \cup N_1^{n-1}(X)} = \emptyset\]
Nevertheless, we do not expect to be able to use the $N_d^r$-operators in a similar way to express that a \emph{vertex set} induces a connected subgraph.
We circumvent this problem by enhancing our framework with ``external'' predicates
for verifying connectivity and acyclicity,
using the techniques of Bergougnoux and Kant{\'e}~\cite{BergougnouxKante2021}.
We call the resulting logic the \emph{\ac distance neighborhood} logic,
and observe that problems such as \textsc{Longest Induced Path} and \textsc{Feedback Vertex Set}
can be expressed in \ac \en logic.
This way, nearly all problems that are known to be solvable in \XP time parameterized by the mim-width
of a given decomposition of the input graph can be expressed in our framework (see \Cref{par:mim:probs}).
We note that Pilipczuk~\cite{Pilipczuk11} also added separate connectivity predicates to his efficient modal logic for tree-width.
Recently, connectivity predicates were also analyzed in the context of \FO 
logic~\cite{PilipczukEtAl2022,schirrmacher_et_al:LIPIcs.CSL.2022.34}.

Rather than just verifying the existence of a tuple of sets with a specified \en-property, 
we can find sets that additionally achieve the best value of some linear optimization function.
While the precise definition is given in \Cref{sec:prelim},
we may for example have a graph $G$ with vertex-weights $w_1,\dots,w_k$
and measure the \emph{weight} of a solution $B_1,\dots,B_k \subseteq V(G)$ by
$\obj_G(B_1,\dots,B_k) = \sum_{i=1}^k \sum_{v \in B_i} w_i(v)$.
In our algorithms, all arithmetic operations depending on input weights are aggregations and comparisons.
We perform our computations under the commonly used random access model, and thus the run time of our algorithms is independent of the involved weights.
Nevertheless, the weight dependent terms remain small enough so that the overhead incurred by changing the computational model to Turing machines would be logarithmic in the highest occurring weight.

\Cref{thm:main:basic} overestimates the contribution of size measurements to the run time,
as size measurements with constant value only contribute with a constant factor to the run time.
The following more refined statement subsumes \Cref{thm:main:basic}, 
except for the parameterization by tree-width.
We compare the run times from our theorems with other logics for tree-width and clique-width in \Cref{tab:compare:logics}.
\begin{theorem}\label{thm:main:AC}
    There is an algorithm that, for a given \ac \en formula~$\phi(X_1,\dots,X_k)$, 
    weighted graph $G$ on $n$ vertices together with an \fwidth decomposition of $G$ of width $w \ge 1$,
    computes a tuple $B^*_1,\dots,B^*_k \subseteq V(G)$ such that
    \[
    \obj_G(B^*_1,\dots,B^*_k) = \max \bigl\{ \obj_G(B_1,\dots,B_k) \bigm| B_1,\dots,B_k \subseteq V(G), G \models \phi(B_1,\dots,B_k) \bigr\}
    \]
    or concludes that no such tuple exists.
    When $\phi$ has size measurements $m_1,\dots,m_s$
    and $r \defeq r(\phi)$, $d \defeq d(\phi)$ %, $b \defeq \log(\maxWeight(G))$,
    and $M \defeq \prod_{i \in [s]}(m_i+2)^{\calO(1)}$
    then the run time of this algorithm is
	\begin{itemize}
		\item $n^{\calO(dw|\phi|^2)}$
			(and $n^{\calO(dw|\phi|)}$ if $r = \calO(1)$) if \fwidth is mim-width,
		\item $2^{\calO(d(wr|\phi|)^2)} M n^3$ (and $2^{\calO(dw|\phi|)}Mn^3$ if $r = 1$)
			if \fwidth is clique-width, and
		\item $2^{\calO(dw^4(r|\phi|)^2)} M n^3$ (and $2^{\calO(dw^2|\phi|)}Mn^3$ if $r = 1$)
			if \fwidth is rank-width,
	\end{itemize}
\end{theorem}

\begin{table}
	\centering
	\renewcommand{\arraystretch}{1.25}
	\begin{tabular}{r|c|c|c|c}
		& \MSO & \ecml & \ac \en & \en \\
		\hline
		tree-width 		& $f_1(w, \card{\phi})n$~\cite{CourcelleI} & $2^{\calO_\phi(w)}n^{\calO(1)}$~\cite{Pilipczuk11} & $2^{\calO(dw^2\card{\phi})}n^{\calO(1)}$ [This] & $2^{\calO(dw\card{\phi})}n^{\calO(1)}$ [This] \\
		clique-width	& $f_2(w, \card{\phi})n$~\cite{CourcelleMakowskyRotics2000} & ? & $2^{\calO(dw\card{\phi})}n^{\calO(1)}$ [This] & $2^{\calO(dw\card{\phi})}n^{\calO(1)}$ [This] \\
	\end{tabular}
	\caption{Comparison of run times of our logic with those from other works. 
		For better comparability, we assume that we consider unweighted problems, 
		that we do not have access to higher powers of the graph, 
		and that we only have (at most) one size measurement, 
		which corresponds to optimization variants of the \MSO-based meta-theorems for tree-width and clique-width.
		Here, $\phi$ denotes the given formula and $w$ the width and
		in case of (\ac) \en logic, $d = d(\phi)$. 
	    The notation $\calO_\phi$ means that the hidden constant depends on $\phi$.
		The functions $f_1$ and $f_2$ are towers of exponentials whose height depends on $\card{\phi}$.
        The run time bound for \ac \en and tree-width is derived via
        a linear upper bound on the tree-width of a graph in terms of its rank-width~\cite{Kante2007,Oum08}.
        To obtain \FPT-algorithms for \ecml parameterized by clique-width, one would have to restrict it to a variant without 
        edge sets due to known hardness results. 
		Since one could keep ultimately periodic counting this would still not be comparable to \en.
    }
	\label{tab:compare:logics}
\end{table}

\subsection{Applications and relationship to other work}\label{par:mim:probs}
Many problems have been shown to be solvable in polynomial time if the input graph is given together with one 
of its branch decompositions of constant mim-width.
First and foremost, Bui-Xuan, Telle, and Vatshelle~\cite{Bui-XuanTV2013} showed that this holds for all 
vertex problems that are \emph{locally checkable},
which includes central subset problems such as
\textsc{Independent Set}, \textsc{Dominating Set}, \textsc{Perfect Code}, \textsc{Induced Matching},
as well as partitioning problems such as 
\textsc{$H$-Homomorphism} and \textsc{$H$-Covering} for fixed $H$, 
\textsc{Perfect Matching Cut}, and \textsc{Odd Cycle Transversal}.
We refer to~\cite{TelleProskurowski1997} for an overview.
The first non-local problems that were shown to be solvable in polynomial time given a branch decomposition of constant mim-width
were problems related to finding induced paths~\cite{JaffkeKwonTelle2020b}, \textsc{Feedback Vertex Set}~\cite{JaffkeKwonTelle2020},
and distance-versions of locally checkable problems~\cite{JaffkeKwonStrommeTelle2019}.
Bergougnoux and Kant{\'e}~\cite{BergougnouxKante2021} added all connected and acyclic versions of locally checkable problems 
to that list, which also generalizes the results of~\cite{JaffkeKwonTelle2020b,JaffkeKwonTelle2020}.
Other problems solvable in polynomial time given a branch decomposition of constant mim-width are
\textsc{Semitotal Dominating Set}~\cite{GalbyEtAl2020}, 
\textsc{Subset Feedback Vertex Set} and \textsc{Node Multiway Cut}~\cite{BergougnouxPapadopoulosTelle2020}.
All of the above mentioned algorithms run in $n^{\calO(w)}$ time except the ones from~\cite{BergougnouxPapadopoulosTelle2020} which run in $n^{\calO(w^2)}$ time,
where $w$ is the mim-width of the given branch decomposition.
Therefore, they are all \XP-algorithms when parameterized by $w$
and \W[1]-hardness has been shown for several locally checkable problems
and \textsc{Feedback Vertex Set}~\cite{BakkaneJaffke2022,FominGR2020,JaffkeKwonStrommeTelle2019,JaffkeKwonTelle2020}.

Except for \textsc{Subset Feedback Vertex Set} and \textsc{Node Multiway Cut}, 
all of the problems mentioned in this paragraph are expressible in \ac\en logic,
and therefore our meta-theorem unifies nearly all of the above algorithmic results.

In a recent preprint~\cite{GonzalezMann2022}, Gonzalez and Mann considered a new framework for locally 
checkable problems which extends that of~\cite{TelleProskurowski1997},
and showed that all problems expressible in that framework are polynomial-time solvable given a branch decomposition of constant mim-width.
This extended the list of polynomial-time solvable problems on graphs of bounded mim-width;
several examples are given in~\cite{GonzalezMann2022}.
Gonzalez and Mann~\cite{GonzalezMann2022} prove that all problems expressible in their framework can be expressed in \en logic as well.
%
%In a current work in progress, 
%we extend the framework of this paper with the ability 
%to capture properties of blocks that cross vertex cuts.
%This enables us to give an algorithm for \textsc{Subset Feedback Vertex Set}
%(and \textsc{Subset Odd Cycle Transversal}) 
%on graphs of bounded mim-width.
%At that point, \emph{every} problem known to be solvable in polynomial time on graphs of constant mim-width (given a decomposition)
%can be handled by our framework.

\myparagraph{New problems solvable on graphs of bounded mim-width}
We do not only unify nearly the entire algorithmic literature on problems solvable in polynomial time given a branch decomposition 
of bounded mim-width, we also extend it.
Several natural problems are expressible in (\ac) \en logic that have not been proved to be solvable in polynomial time on 
graphs of bounded mim-width before.
Several of them concern proper colorings with a bounded number of colors, 
where we have additional restrictions on the color classes or their interactions:
in \textsc{Acyclic $k$-Coloring}, we seek proper $k$-coloring such that each pair of color classes induces a forest,
and \textsc{Star $k$-Coloring} ask for a proper $k$-coloring  such that each pair of color classes induces a star forest.
In \textsc{$b$-Coloring with $k$ Colors}, we want a proper $k$-coloring such that each color class 
contains a vertex that has a neighbor in all the other color classes.
In \textsc{Conflict-free $k$-Coloring}, the goal is to find a (not necessarily proper) $k$-coloring of the input graph such that for every vertex there is a color that appears at most once in its neighborhood.
All of these problems are expressible in \ac \en logic, and it is worth noting that for \textsc{Star $k$-Coloring} 
(and, naturally, for \textsc{$b$-Coloring} and \textsc{Conflict-free $k$-Coloring})
\en logic already suffices.
The \textsc{$k$-$L(2,1)$-Labeling} problem, 
of importance in frequency assignment~\cite{Calamoneri2011,GriggsYeh1992,Hale1980},
asks for a vertex coloring of a graph with colors $\{1, \ldots, k\}$,
such that the colors of vertices at distance one differ by at least two,
and colors of vertices at distance two differ by at least one.
More generally, in the \textsc{$k$-$L(d_1, \ldots, d_s)$-Labeling} problem, 
we require that for all $i$, colors of vertices at distance $i$ differ by at least $d_i$.
Such problems can be expressed in \en logic as well.

\myparagraph{Application to solution diversity}
Baste et al.~\cite{BasteEtAl2020} initiated the study of the \emph{solution diversity} 
paradigm in the context of parameterized algorithms.
Here, we are given some combinatorial problem whose solutions are sets over some universe 
(a so-called \emph{subset problem}),
and instead of finding a single solution, we want to find a small set of solutions that is sufficiently diverse.
Natural ways of measuring the diversity of a set of solutions are 
to maximize the minimum pairwise Hamming distance of the indicator vectors of its members,
or to maximize the sum of the Hamming distances.
For a subset problem $\Pi$, we denote by \textsc{Min-Diverse $\Pi$} 
the problem asking for $p$ solutions to $\Pi$ whose minimum pairwise Hamming distance is at least some given value $t$,
and by \textsc{Sum-Diverse $\Pi$} 
the problem asking for $p$ solutions to $\Pi$ whose sum of pairwise Hamming distances is at least some given value $t$.
Both ways of measuring diversity can be expressed in our framework.
The diversity measure in \textsc{Min-Diverse $\Pi$} is expressed in \en logic directly, which requires the addition of 
$\calO(p^2)$ size measurements to the formula.
We get more efficient encodings for \textsc{Sum-Diverse $\Pi$}, without an increase in the number of size measurements.
This is due to the fact that the diversity measure in \textsc{Sum-Diverse $\Pi$} 
can be encoded in the weight function of the input graph.%
\iflong\footnote{We give the concrete run times resulting from the following observation in \Cref{tab:div:runtime}.}\fi
\begin{observation}\label{obs:diversity:formulas}
	Let $\Pi$ be a vertex subset problem expressible in \ac\en logic via formula $\phi(X)$ with $s$ size measurements,
    and suppose that each of the following diverse variants of $\Pi$ ask for $p$ solutions.
	\begin{enumerate}
		\item \textsc{Min-Diverse $\Pi$} is expressible in \ac\en logic via a formula $\phi'(X_1, \ldots, X_p)$ such that 
			$\card{\phi'} = \calO(p\card{\phi} + p^2)$, $d(\phi') = d(\phi)$, and $r(\phi') = r(\phi)$,
			and $\phi'$ has $s + \calO(p^2)$ size measurements. 
		\item \textsc{Sum-Diverse $\Pi$} can be expressed as an \ac\en optimization problem based on a formula $\phi'(X_1, \ldots, X_p)$ such that
			$\card{\phi'} = \calO(p\card{\phi})$, $d(\phi') = d(\phi)$, and $r(\phi') = r(\phi)$,
			and $\phi'$ has $s$ size measurements.
	\end{enumerate}
\end{observation}

We would like to point out that the algorithms implied by the previous observation generalize some \FPT-results for diverse problems parameterized by tree-width~\cite{BasteEtAl2020}
to the parameterization by clique-width, albeit only for problems expressible in \ac\en logic.
Very recently, and independently of this work,
Hanaka et al.~\cite{HanakaEtAl2022} showed that for fixed $p$, \textsc{Diverse Independent Set} can be solved 
in polynomial time on interval graphs.
The algorithm resulting from \Cref{obs:diversity:formulas} contains this result as a special case, 
since interval graphs have linear mim-width $1$, 
and a decomposition witnessing this can be computed in polynomial time~\cite{BelmonteVatshelle2013}.

\myparagraph{Computing mim-width}
The algorithms given by our meta-theorems require a decomposition of bounded mim-width 
to be given along with the input graph.
However, little is known about the complexity of computing or approximating mim-width in general.
Both the question whether computing mim-width exactly is para-\NP-hard, 
and the question whether there is an \XP-time approximation algorithm with \emph{any} approximation guarantee that does not depend 
on the size of the graph, are open at the moment.
The problem of computing mim-width has been shown to be \W[1]-hard, and there is no polynomial-time constant-factor approximation 
unless $\NP = \ZPP$~\cite{SaetherVatshelle2016}.
The only known positive special cases are an algorithm to compute the linear mim-width of trees~\cite{HogemoEtAl2019},
\FPT-algorithms parameterized by tree-width plus maximum degree or by treedepth, 
and a linear kernel parameterized by the feedback edge set number~\cite{EibenEtAl2022}.
However, in many graph classes of bounded mim-width, we can compute bounded mim-width decompositions in polynomial time~\cite{BelmonteVatshelle2013,JaffkeThesis,VatshelleThesis}.
These graph classes include circular-arc, circular permutation, convex, leaf power, $H$-graphs (for fixed $H$), 
as well as Dilworth-$k$ and complements of $k$-degenerate graphs (which includes complements of planar graphs) 
for fixed $k$.
We refer to~\cite{JaffkeThesis,VatshelleThesis} for an overview.

\subsection{Tightness of the results}
As mentioned above, 
several problems have been shown to be \W[1]-hard parameterized by the mim-width $w$ of a given decomposition~\cite{BakkaneJaffke2022,FominGR2020,JaffkeKwonStrommeTelle2019}.
Many of these reductions also rule out
$f(w)n^{o(w/\log w)}$ time algorithms under the Exponential Time Hypothesis~\cite{BakkaneJaffke2022}.
\Cref{thm:main:AC} gives $n^{\calO(w)}$ time algorithms for all these problems,
which matches the lower bound up to a $(\log w)$-factor in the degree of the polynomial.
Moreover, many fundamental graph problems (expressible in \ac\en logic) 
are known to not have algorithms whose run time is subexponential in the number of vertices of the input graph,
and therefore the tree-width or clique-width,
unless the \ETH fails~\cite{LokshtanovEtAl2011}.
In a recent preprint~\cite{BergougnouxKN22}, Bergougnoux, Korhonen, and Nederlof  proved that some problems including \textsc{Independent Set} cannot be solved in $2^{o(\rw^2)}\cdot n^{O(1)}$ unless \ETH fails.
These lower bounds implies that our model-checking algorithm is \ETH-tight for these parameterizations.

\iffalse
In some cases, 
algorithms based on our meta-theorem even depict known run time barriers very concisely.
If we consider the \textsc{$q$-Coloring} problem parameterized by the tree-width $w$,
\Cref{thm:main:basic} gives a bound of 
$2^{\calO(wq)}n^{\calO(1)}$.
%
On the other hand, the straightforward \en sentence expressing \textsc{Equitable $q$-Coloring} invokes 
$q$ size measurements, so \Cref{thm:main:basic} gives a run time of $2^{\calO(wq)}n^{\calO(q)}$ 
for this problem.
By known hardness results~\cite{FellowsEtAl2011,GomesEtAl2019},
the dependence on $q$ in the degree of the polynomial cannot be avoided for \textsc{Equitable $q$-Coloring} 
unless $\FPT = \W[1]$.
\fi

In several senses, \Cref{thm:main:AC} is also tight in terms of the expressive power of the logic.
We show that if we give the logic slightly more expressive power, then the corresponding model checking problem
encounters a hardness barrier.
In \Cref{sec:hardness} we define and analyze the following three extensions of \en.
	\begin{itemize}
		\item \textnormal{\textsf{DN+$\forall$}} allows
            the use of a single innermost universal quantification.
            This logic can express \textsc{Clique} by the sentence
            $\exists X \, \card{X} \ge k \wedge \forall X' \, \bigl( X' \cap X \neq \emptyset) \to  \bigl( X \subseteq X' \cup N_1^1(X') \bigr)$.
            \item \textnormal{\textsf{DN+EdgeSets}} allows, in the spirit of \ecml~\cite{Pilipczuk11}, quantification over edge sets $Y$
            and operators $N_Y(t)$ that evaluate to all vertices reachable from $t$ via an edge in $Y$.
            This logic can express \textsc{Max Cut} by the sentence
            $\exists_{\subseteq V} X \, \exists_{\subseteq E} Y \, |Y| \ge k \land N_Y(X)=\comp{X} \land N_Y(\comp{X})=X.$
        \item \textnormal{\textsf{DN+Parity}} allows operators $N^\text{even}(t)$
            that evaluate to all vertices with an even number of neighbors in $t$.
    \end{itemize}
Since \textsc{Clique} and \textsc{Max Cut} are \paraNPhard by mim-width, 
and since \textsf{DN+Parity} can express a problem that is \NP-hard on interval graphs,
we get the following.
\begin{theorem}\label{thm:hardness:extensions}
	The model checking problems for 
    \textnormal{\textsf{DN+$\forall$}},
    \textnormal{\textsf{DN+EdgeSets}} and
    \textnormal{\textsf{DN+Parity}} 
    are para-\NP-hard parameterized by the (linear) mim-width 
	of a given decomposition plus formula length. 
\end{theorem}

Note that the hardness results in the previous theorem are independent of the complexity of computing (linear) mim-width.
In one case the hardness follows from the \NP-hardness of \textsc{Clique} in co-planar graphs, 
a graph class for which there  is a linear-time algorithm to compute linear branch decompositions of mim-width at most~$6$~\cite{BelmonteVatshelle2013,MatulaBeck1983}.
In the other cases it follows from the \NP-hardness of problems in interval graphs,
where we can compute linear branch decompositions of mim-width at most $1$ in linear time~\cite{BelmonteVatshelle2013,BoothLueker1976,HabibEtAl2000}.

\subsection{Techniques}
	The general strategy of our algorithm is bottom-up dynamic programming 
	along the branch decomposition that is given together with the input graph $G$ and an \ac\en formula $\phi$ 
	with variables $(X_1, \ldots, X_k)$.
	For each cut $(A, \comp{A})$ induced by the given decomposition, we compute a set $\cB_A$ of partial solutions, 
	that is, assignments of tuples of subsets $(B_1, \ldots, B_k)$ of $A$ to $(X_1, \ldots, X_k)$,
	which, once completed with vertices from $\comp{A}$, might eventually lead to a proof that $G \models \phi$.
	The runtime of the algorithm is essentially determined by the sizes of such sets $\cB_A$,
	and we work towards keeping these sets small.

	The correctness of our algorithm is based on the property that each set $\cB_A$ \emph{represents} all partial solutions in $A$, in the following sense.
	If there exists a way of completing some partial solution of $A$ with vertices from $\comp{A}$ 
	that leads to a proof that $G \models \phi$, then $\cB_A$ contains at least one such partial solution.
	We compute such sets in a bottom-up manner along the given branch decomposition of the input graph.
	At leaf nodes, it is trivial to obtain a representative set of the desired size.
	At internal nodes (with two children), we start by taking all possible combinations of pairs of one 
	partial solution per child and then try to reduce the size of the resulting set while preserving the fact that 
	it represents all partial solutions.
	This happens via a \emph{reduce routine} which is intended to solve the following task: 
	given a slightly too large set of partial solutions $\cB_A'$ in $A$, 
	it produces a set of partial solutions $\cB_A \subseteq \cB_A'$ that represents $\cB_A'$
	within the desired run time, and such that $\cB_A$ is small enough.
	Using standard methods, we show in \cref{sec:genericAlgo} that each problem that admits such a reduce routine
	can be solved efficiently.
	
	One crucial tool in this reduce routine is the equivalence relation $\equiv^A_\phi$ over $2^A$ 
	based on the $d$-neighborhood equivalence relation 
	due to Bui-Xuan, Telle, and Vatshelle~\cite{Bui-XuanTV2013},
	adjusted according to properties of $\phi$.
	Every pair of tuples over $2^A$ 
	that is equivalent under $\equiv^A_\phi$ 
	has, roughly speaking,
	the same neighborhood structure across the cut $(A, \comp{A})$ 
	when observed through the lens of $\phi$.
	Therefore two such equivalent tuples have very similar properties when it comes to being completable to a solution 
	witnessing $G \models \phi$.
	Moreover, $\equiv^A_\phi$ 
	has the additional appeal that
	bounds on the width measures we consider in this work imply bounds on the index of these equivalence relations.
	
	Before we outline how to obtain the reduce routine for the model checking problem of \ac\en logic,
	we need to discuss one more concept.
	The choice of partial solutions in $A$ depends on how they ``expect'' to be completed 
	by vertices from $\comp{A}$.
	Consider for instance the case when $\phi$ encodes the \textsc{Dominating Set} problem; 
	here, vertices from $\comp{A}$ may dominate vertices from $A$, which in turn affects the choices for partial solutions in $A$.
	We consider each \emph{expectation $\btE$ from $\comp{A}$},
	formally a $k$-tuple of equivalence classes of $\equiv^{\comp{A}}_\phi$,
	which signals that 
	the partial assignments in question expect to be completed with a $k$-tuple of subsets from $\comp{A}$
	that is coordinate-wise contained in $\btE$.

We design a reduce routine for the model checking problem of \ac\en logic in a number of steps.
First, we consider only \en logic, without acyclicity and connectivity predicates.
We define a minimal fragment of \en logic that we call \emph{core \en logic} (\encore logic),
and show that it has the same expressive power as \en.
We observe that the conversion of any \en formula to a \encore formula 
can be done without prohibitively increasing the formula length.
We first prove the meta-theorem for quantifier-free \encore formulas.
In this case, given any expectation $\btE$,
we can restrict the search for the aforementioned ``best'' partial solutions compatible with $\btE$
to partial solutions that are best in an equivalence class of another efficiently computable equivalence relation 
depending on $\btE$ that has small enough index.

The generalization to \en logic uses the conversion from \encore logic,
the reduce routine from the meta-theorem for \encore logic,
and several bounds of the number of equivalence classes of the neighborhood equivalence 
in terms of the width measure in question.

The generalization of \Cref{thm:main:basic} to \ac \en logic 
happens in several steps.
As before, we first consider a simpler fragment of \ac\en logic;
in particular we consider \emph{\ac-clauses}, which are conjunctions of a \encore formula and 
predicates $\acy(X)$ or $\conn(X)$ for some set variable $X$.
The motivation for this restriction is as follows.
In the final theorem, we can convert the given \ac\en formula into a normal form 
containing \ac\encore literals, and acyclicity or connectivity literals of the form
$\conn(X)$, $\acy(X)$, $\neg \conn(X)$, or $\neg \acy(X)$.
We observe that we can deal with the negations of connectivity and acyclitiy contraints using \en (and therefore \encore) logic directly:
	\begin{itemize}
        \item $\neg\conn(X)$ is equivalent to the existence of a partition $Z_1,Z_2$ of $X$ such that no edges go between these two parts.
        That is, $\neg\conn(X) \equiv \exists Z_1 \exists Z_2 \, Z_1 \neq \emptyset \land Z_2 \neq \emptyset \land Z_1 \cup Z_2 = X \land N(Z_1) \cap Z_2 = \emptyset$.
        \item Also, $\neg\acy(X)$ is equivalent to the existence of a subset $Z$ inducing graph with minimal degree two.
        That is, $\neg\acy(X) \equiv \exists Z \, Z \neq \emptyset \land Z \subseteq X \land Z \subseteq N_2^1(Z)$.
    \end{itemize}

It therefore suffices to prove the meta-theorem for \ac-clauses.
Again, this takes several steps.
First, we prove a lemma that deals with the case when all \ac-constraints are connectivity constraints.
This can be done using representatives of the $1$-neighbor equivalence relation together with the rank-based approach 
due to Bodlaender et al.~\cite{BodlaenderCKN15}.
The proof follows a similar strategy as the one presented in~\cite{BergougnouxKante2021};
however it becomes more complex: 
While~\cite{BergougnouxKante2021} deals with the addition of a single connectivity constraint,
we have to be able to take care of any number $k \ge 1$ of connectivity constraints.

The next step is to allow acyclicity constraints, but only in combination with connectivity constraints. 
This means that each variable $X$ for which $\acy(X)$ is in the \ac-clause, $\conn(X)$ must also be in the
\ac-clause.
We improve\footnote{The general strategy again follows that of~\cite{BergougnouxKante2021} with non-trivial and crucial improvements to obtain \FPT-algorithms parameterized by clique-width/rank-width plus the number of acyclicity constraints.
	In the discussion after Theorem 8.1 in~\cite{BergougnouxKante2021} it is claimed that their results
	can be generalized to give these \FPT-algorithms.
	However, the equivalence relation designed in~\cite{BergougnouxKante2021} has too many equivalence classes 
	and therefore their approach does not work.} 
the equivalence relation from~\cite{BergougnouxKante2021} 
based on the $2$-neighbor equivalence relation and we use it to give 
a reduce-routine with the desired running time.
However, there is a substantial difference between this step and the previous one.
For connectivity constraints, we manage to give a reduce-routine whose size and run time bounds 
depend only on the number of 1-neighbor equivalence classes.
The equivalence relations for acyclicity constraints require several combinatorial analyses involving trees 
crossing cuts that differ from width measure to width measure.

The last step is to decouple the acyclicity constraints from the connectivity constraints which is 
achieved via a reduction that preserves the values of the width measures within a sufficient margin.
After all these efforts, we are able to derive \Cref{thm:main:AC} in a similar way as we derived \Cref{thm:main:basic}
by using the corresponding lemmas.
\section{Preliminaries}
\label{sec:prelim}

We denote by $\N$ the set of non-negative integers and by $\N^+$ the set $\N\setminus \{0\}$.
For $k\in \N$, we denote by $[k]$ the set of integers $\{1, 2, \ldots, k\}$.
We let $\max(\emptyset):=-\infty$.
For a set $A$, we denote its power set by $\cP(A)$.
We denote by $\Pset(A)^{k}$ the set of all tuples $(B_1,\dots,B_k)$ with $B_1,\dots,B_k\subseteq A$.
We use the notation $\tB$ to refer to a tuple $(B_1,B_2,\dots,B_{|\tB|})$.
Given $\tB,\tC\in \Pset(A)^{k}$, we denote by $\tB \cup \tC$ the tuple $\tD \in\Pset(A)^{k}$ with $D_i = B_i \cup C_i$ for all $i \in [k]$.
For a subset $S$ of $A$, we denote by $\tB \cap S$ the tuple $\tC \in \Pset(S)^{k}$ with $C_i = B_i \cap S$ for all $i \in [k]$.

\myparagraph{Graphs}
We consider simple, vertex-colored, vertex-weighted graphs.
We denote by $n$ the order of $V(G)$, where the corresponding graph is always clear from the context.
For every vertex set $A\subseteq V(G)$ we denote by $\comp{A}$ the set $V(G)\setminus A$.
For a graph $G$ and color class $\textbf{P}$, we denote the vertices in $G$ with color $\textbf{P}$ by $\textbf{P}(G) \subseteq V(G)$.
Note that color classes may be empty and that vertices may have zero, one or multiple colors.
A \emph{$k$-weighted graph} is a graph $G$ that associates with every $v \in V(G)$ and every $S \subseteq [k]$ a weight $w(v,S) \in \N$.
For every $k$-weighted graph $G$ and $k$-tuple $\tB\in \Pset(V(G))^{k}$ we say the \emph{weight} of $\tB$ in $G$ equals $\obj_G(\tB) = \sum_{v \in V(G)} w(v,\{ i \mid v \in B_i \})$.
%Let $\maxWeight(G) = 2+ \max \{ w(v,S) \mid v \in V(G), S \subseteq [k] \}$.
We use the well-known random access model where atomic operations on integers take constant time.
When taking a $k$-weighted graph as input, its weights are loaded into separate registers, so we may add and compare them in constant time.
We promise that at all times, the values of these weight computations do not exceed $\sum_{S \subseteq [k]} \sum_{v \in V(G)} w(v,S)$,
so that these computations could be executed on a Turing machine with an overhead logarithmic in the highest occuring weight.

Given two vertices $u$ and $v$, we denote by $\dist(u,v)$ the distance between $u$ and $v$ in $G$, i.e., the number of edges in a shortest path between $u$ and $v$.
The \emph{distance matrix} of $G$ is the $n \times n$-matrix $M$ such that $M[u,v]=\dist(u,v)$ for every $u,v\in V(G)$.

The subgraph of $G$ induced by a subset $A \subseteq V(G)$ of its vertices is denoted by $G[A]$.
For two disjoint subsets $A$ and $B$ of $V(G)$, we denote by $G[A,B]$ the bipartite graph with vertex set $A\cup B$ and edge set $\{ab \in E(G)\mid a\in A \text{ and } \ b\in B \}$. 
Moreover,  we denote by $M_{A,B}$ the adjacency matrix between $A$ and $B$, i.e., the $|A| \times |B|$-matrix such that $M_{A,B}[a,b]=1$ if $ab\in E(G)$ and 0 otherwise.

Given $r\in \N$ and $v\in V(G)$, we denote the open neighborhood of $v$ in $G^r$ (the $r$-th power of $G$) by $N^r(v)=\{u\in V(G)\setminus \{v\} \mid \dist(u,v)\leq r\}$.
\begin{definition}\label{def:neighborhood:operator}
	For $d,r\in \N$ and $U\subseteq V(G)$, we define the $d$-neighborhood of $U$ in $G^r$ as
	$$N_d^r(U)=\{v\in V(G)\mid \abs{N^r(v)\cap U}\geq d\}.$$
\end{definition}
Observe that $N_d^1(U)$ equals the set of vertices $v\in V(G)$ with at least $d$ neighbors in $U$.

\myparagraph{\bf $(d,R)$-neighbor equivalence} The following concepts were introduced in \cite{Bui-XuanTV2013}, but we extend them to higher distances. 
Let $G$ be a graph, $A\subseteq V(G)$ and $d,r\in \bN^+$.
Two subsets $B$ and $C$ of $A$ are \emph{$(d,r)$-neighbor equivalent over $A$}, denoted by $B\equi{A}{d,r} C$, if $\min(\abs{N^r(v)\cap B},d)=\min(\abs{N^r(v)\cap C},d)$ for every $v\in \comp{A}$. 
Given a set $R\subseteq \bN^+$, we say that $B\equi{A}{d,R} C$ if $B\equi{A}{d,r} C$ for every $r\in R$.
It is not hard to check that $\equi{A}{d,R}$ is an equivalence relation.
See Figure \ref{fig:nec} for an example of $(2,1)$-neighbor equivalent sets.
\begin{figure}[h!]
	\centering
	\includegraphics[scale=0.95]{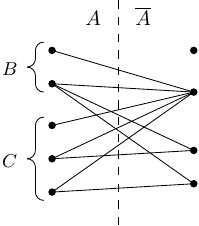}
	\caption{It holds that $B\equi{A}{2,1} C$, but not $B\equi{A}{3,1} C$.}
	\label{fig:nec}
\end{figure}

For a graph $G$, $d\in \bN^+$ and $R\subseteq \bN^+$, let $\nec_{d,R} : 2^{V(G)} \to \bN$ be such that for all $A\subseteq V(G)$, $\nec_{d,R}(A)$ equals the number of equivalence classes of $\equi{A}{d,R}$ in $G$.
To simplify notation, we will use the shorthand $\snec_{d,R}(A)$ to denote $\max(\nec_{d,R}(A),\nec_{d,R}(\comp{A}))$ (where $\mathsf{s}$ stands for symmetric).

To manipulate the equivalence classes of $\equi{A}{d,R}$, we compute a representative of each equivalence class in polynomial time.
This is achieved with the following notion of representatives.
We fix an arbitrary ordering of $V(G)$.
For each $B\subseteq A$, let us denote by $\rep{A}{d,R}(B)$ the lexicographically smallest set $C\subseteq A$ among all $C\equi{A}{d,R} B$ of minimum size.
Moreover, we denote by $\Rep{A}{d,R}$ the set $\{\rep{A}{d,R}(B)\mid B\subseteq A\}$.
It is worth noticing that the empty set always belongs to $\Rep{A}{d,R}$, for all $A\subseteq V(G)$.
Moreover, we have $\Rep{V(G)}{d,R}=\Rep{\emptyset}{d,R}=\{\emptyset\}$ for all $d$ and $R$.
The following lemma was proved for $R=\{1\}$ in \cite[Lemma 1]{Bui-XuanTV2013}. We simply generalize it to handle any $R$. 
The proof can be found in the appendix.
\begin{restatable}[$\star$]{lemma}{computerep}\label{lem:computerep}
        Let $d\in \N^+$ and $R\subseteq \bN^+$ be finite.
        Given a graph $G$, its distance matrix (if $R \neq \{1\}$), and a set $A\subseteq V(G)$, in
        time $O(\nec_{d,R}(A)\cdot \abs{R} \cdot n^2 \cdot \log(\nec_{d,R}(A)))$ one can compute $\Rep{A}{d,R}$ together with a data
        structure that, given a set $B\subseteq A$, returns a pointer to $\rep{A}{d,R}(B)$ in time $O(n^2\cdot \abs{R} \cdot \log(\nec_{d,R}(A)))$.
\end{restatable}

\myparagraph{Graph width measures}
A \emph{rooted binary tree} is a binary tree with a distinguished vertex called the \emph{root}. 
Let $G$ be a graph.
A \emph{rooted layout} of $G$ is a pair $\cL=(T,\delta)$, where $T$ is a rooted binary tree and $\delta$ is a bijective function from $V(G)$ to the leaves of $T$. 
We call the elements of $V(T)$ \emph{nodes}, to avoid confusion with the vertices of $G$.
For each node $x$ of $T$, let $L_x$ be the set of all leaves $l$ of $T$ such that the path from the root of $T$ to $l$ contains $x$. 
We denote by $V_x^{\cL}$ the set of vertices that are in bijection with $L_x$, i.e., 
$V_x^{\cL}=\{v \in V(G)\mid \delta(v)\in L_x\}$. We omit $\cL$ from the superscript when it is clear from the context.

A function $\f: 2^{V(G)} \to\bN$ is called a \emph{set function} on $G$.
The width measures considered in this paper are instantiations of the following one.
Given a set function $\f: 2^{V(G)} \to\bN$ and a rooted layout $\cL=(T,\delta)$,
the $\f$-width of $G$ on $\cL$, denoted by $\f(G,\cL)$, is $\max\{\f(V_x ) \mid x \in V(T)\}$.  
The $\f$-width of $G$, denoted by $\f(G)$, is the minimum $\f$-width over all rooted layouts of $G$.

In particular, this yields the width measure $\snec_{d,R}(G)$
using the set function $\snec_{d,R}(A)$ defined previously.
We can define various additional width measures by using different set functions.
Recall that a \emph{matching} $M \subseteq E(G)$ is a set of pairwise disjoint edges, 
and that $M$ is an \emph{induced matching} if there are no additional edges between the endpoints of $M$ in $G$, i.e.,
$M = E(G[V(M)])$, where $V(M)$ is the set of endpoints of the edges in $M$.
\begin{itemize}
    \item We define $\mmw(A)$ as the size of a maximum matching in the graph $G[A,\comp{A}]$ and 
        $\mmw(G)$ is called the \emph{maximum-matching-width} of $G$.
    \item We define $\mim(A)$ as the size of a maximum induced matching in the graph $G[A,\comp{A}]$ and
        $\mim(G)$ is called the \emph{mim-width} of $G$.
    \item We define $\mw(A)$ as the cardinality of $\{ N(v)\cap \comp{A} \mid v\in A \}$ and 
        $\mw(G)$ is called the \emph{module-width} of $G$.
    \item We define $\rw(A)$ as the rank over $GF(2)$ of the matrix $M_{A,\comp{A}}$ and
        $\rw(G)$ is called the \emph{rank-width} of $G$.
\end{itemize}

For every graph $G$, $\mmw(G)\leq \tw(G) +1 \leq 3\cdot\mmw(G)$ and $\mw(G)\leq \cw(G)\leq 2\cdot\mw(G)$ where $\tw(G)$ and $\cw(G)$ are the tree-width and the clique-width of $G$ \cite{Rao2006,VatshelleThesis}. One can moreover translate, in time at most $O(n^2)$, a given decomposition into the other one with width at most the given bounds.
The following lemma shows how $\snec_{d,1}(A)$ and $\snec_{d,r}(A)$ is upper bounded by the other parameters. 
\newcommand\booldim{\text{\textsf{bool-dim}}}
\begin{lemma}[\cite{BelmonteVatshelle2013,VatshelleThesis}]\label{lem:necd1}
	Let $G$ be a graph and $d,r\in \bN^+$. For every $A\subseteq V(G)$, we have the following upper bounds on $\snec_{d,1}(A)$:
	\begin{multicols}{2}
        \begin{enumerate}[(a)]
			\item\label{lem:necd1:mmw} $(2d+2)^{\mmw(A)}$
			\item\label{lem:necd1:mw} $(d+1)^{\mw(A)}$
			\item\label{lem:necd1:rw} $2^{d \cdot \rw(A)^2}$
			\item\label{lem:necd1:mimw} $n^{d \cdot \mim(A)}$
		\end{enumerate}
	\end{multicols}
    \noindent Moreover, we have $\snec_{1,r}(A)\leq \snec_{1,1}(A)^r$ and the following upper bounds on $\snec_{d,r}(A)$:
	\begin{multicols}{2}
        \begin{enumerate}[(a)]
			\setcounter{enumi}{4}
			\item\label{lem:necd1:r} $\snec_{1,1}(A)^{dr^2\cdot \log(\snec_{1,1}(A))}$
			\item\label{lem:necd1:r:mimw} $n^{2d\cdot\mim(A)}$
		\end{enumerate}
	\end{multicols}
\end{lemma}
\begin{proof}
    The bound (\ref{lem:necd1:mw}) has been shown in \cite[Lemma 5.2.2]{VatshelleThesis},
    (\ref{lem:necd1:rw}) has been derived in \cite{BergougnouxKante2021},
    (\ref{lem:necd1:mimw}) is from \cite[Lemma 2]{BelmonteVatshelle2013}.
	We show the remaining bounds.
	
    For (\ref{lem:necd1:mmw}),
    let $M$ be a maximal matching in $G[A, \comp{A}]$.
    We define an equivalence relation $\sim_M$ over $2^A$ as follows.
    For $B, C \subseteq A$, we let $B \sim_M C$ if $B \cap V(M) = C \cap V(M)$ 
    and for all $v \in V(M) \cap \comp{A}$, $\min(d, \card{N(v) \cap B}) = \min(d, \card{N(v) \cap C})$.
    Suppose that $B \sim_M C$, and let $v \in \comp{A}$.
    If $v \notin V(M)$, then we observe that $N(v) \cap A \subseteq V(M)$, 
    otherwise $M$ is not a maximal matching.
    Since $B \cap V(M) = C \cap V(M)$, we conclude that $N(v) \cap B = N(v) \cap C$.
    If $v \in V(M)$, then $\min(d, \card{N(v) \cap B}) = \min(d, \card{N(v) \cap C})$ by the definition of $\sim_M$.
    We conclude that $B \equiv^A_{d, 1} C$,
    and therefore $\snec_{d,1}(A) \le \card{2^A/{\sim_M}} \le 2^{\card{M}}\cdot (d+1)^{\card{M}} = (2d+2)^{\mmw(A)}$.
	
	We show that $\snec_{1,r}(A)\leq \snec_{1,1}(A)^r$ by induction on $r$. 
	For $r = 1$ it is trivial. 
	Suppose $r > 1$.
	We observe that for any $X \subseteq A$, $N^r(X) \cap \comp{A} = N^{r-1}(X) \cap \comp{A} \cup N(N^{r-1}(X) \cap A) \cap \comp{A}$.
	By the induction hypothesis, $\card{\{N^{r-1}(X) \cap \comp{A} \mid X \subseteq A\}} \le \snec_{1, 1}(A)^{r-1}$,
	and since $N^{r-1}(X) \cap A \subseteq A$, 
	$\card{\{N(N^{r-1}(X) \cap A) \cap \comp{A} \mid X \subseteq A\}} \le \snec_{1, 1}(A)$,
	so the claimed bound follows.
	
    Consider (\ref{lem:necd1:r}). 
	In \cite[Lemma 5]{Bui-XuanTV2013}, it is shown%
	\footnote{In~\cite[Lemma 5]{Bui-XuanTV2013}, this bound is stated using the quantity $\booldim(A)$ instead of $\snec_{1,1}(A)$; however $\booldim(A)$ is defined as  $\log \snec_{1,1}(A)$.}
    that $\snec_{d, 1}(A) \le \snec_{1,1}(A)^{d\cdot \log \snec_{1,1}(A)}$.
	Since $\snec_{d, r}(A) = \snec^{G^r}_{d,1}(A)$ we can apply this lemma, and using the previous bound,
	\begin{align*}
		\snec_{d, r}(A) &\le \snec_{1,r}(A)^{d\cdot \log \snec_{1,r}(A)} \\
				&\le \left(\snec_{1,1}(A)^r\right)^{dr\cdot \log\snec_{1,1}(A)} \\
				&\le \snec_{1,1}(A)^{dr^2\log\snec_{1,1}(A)}.
	\end{align*}
	
    Lastly, we show (\ref{lem:necd1:r:mimw}).
	By~\cite[Theorem 5]{JaffkeKwonStrommeTelle2019}, we have that $\mim^{G^r}(A) \le 2\cdot \mim^G(A)$.
    Therefore, by (\ref{lem:necd1:mimw}), we have that $\snec^G_{d, r}(A) = \snec^{G^r}_{d, 1} \le n^{d\cdot\mim^{G^r}(A)} \le n^{2d\cdot\mim(A)}$.
\end{proof}

\section{Distance neighborhood logic}\label{sec:logic}

We define our distance neighborhood logic (\en for short) and its extension with acyclicity and connectivity constraints (\ac \en for short) by extending existential monadic second-order (\MSOone) logic.
Remember that \MSOone allows quantification over vertices and sets of vertices
together with an adjacency relation $E(\cdot,\cdot)$ and equality relation $=$ between vertices, unary vertex relations (i.e., colors)
as well as the containment relation $\in$ of vertices in sets.
Existential \MSOone is the restriction of \MSOone to existential quantifiers, while furthermore requiring that only quantifier-free formulas may be negated.
Vertex variables are denoted by lower-case letters ($x,y,z,\dots$), 
while set variables are denoted by upper-case letters ($X,Y,Z,\dots$).
Furthermore, unary relations (or colors) are denoted by bold letters ($\mathbf{P},\mathbf{Q},\mathbf{R},\dots$).

\myparagraph{Syntax}
We first define so-called \emph{neighborhood terms} using the following rules:

\begin{enumerate}
	\item Every set variable $X$ is a neighborhood term.
	\item $N^r_d(t)$ is a neighborhood term for every $d,r \in \N^+$ and neighborhood term $t$.
	\item Every unary relational symbol $\mathbf{P}$ is a neighborhood term.
	\item $\emptyset$ is a neighborhood term.
    \item If $t_1$ and $t_2$ are neighborhood terms then $\comp{t_1}$, $t_1 \cap t_2$, $t_1 \cup t_2$ and $t_1 \setminus t_2$ are also neighborhood terms.
\end{enumerate}

\noindent Then \en is the extension of existential \MSOone by the following two rules:
\begin{enumerate}
	\setcounter{enumi}{5}
	\item If $t$ is a neighborhood term and $m \in \N$ then $|t| = m$, $|t| \le m$ and $|t| \ge m$ are formulas called \emph{size measurements}.
	\item If $t_1$ and $t_2$ are neighborhood terms then $t_1 = t_2$, $t_1 \subseteq t_2$ and $t_1 \supseteq t_2$ are formulas.
\end{enumerate}
 
\noindent On top of this, \ac \en extends \en with the following additional rule.
\begin{enumerate}
	\setcounter{enumi}{7}
	\item If $t$ is a neighborhood term, then $\conn(t)$ and $\acy(t)$ are formulas.
\end{enumerate}

Let $\xi$ be a formula or a term of \en or \ac \en.
We denote the variables that occur in $\xi$ by the ordered tuple $\var(\xi)$.
We define $d(\xi) \defeq \max\{2,d\}$, where $d$ is the largest value such that a term of the form $N^\cdot_d(\cdot)$ appears in $\xi$.
We denote by $R(\xi)$ the set of all numbers $r \in N$ such that either $r=1$ or a term of the form $N^r_\cdot(\cdot)$ appears in $\xi$.
We further set $r(\xi) = \max R(\xi)$.

We define the \emph{length} $|\phi|$ of a (\ac) \en formula $\phi$ to be the number of symbols of $\phi$.
Note that every number (as occurring for example in a size measurement $t \le m$ or in a super or subscript of a neighborhood term $N^r_d(\cdot)$) is \emph{one} symbol.

\myparagraph{Semantics}
Next, we define the semantics of our logic. 
We consider \emph{vertex-colored} graphs.
This means, each vertex of a graph may be in zero, one or more unary relations ($\mathbf{P},\mathbf{Q},\mathbf{R},\dots$).
An \emph{interpretation} of a formula $\phi$ is a tuple $(G,\beta)$ consisting of a graph $G$
and a function $\beta \colon \dom(\beta) \to 2^{V(G)}$ with $\var(\phi) \subseteq \dom(\beta)$.
Given an interpretation $(G,\beta)$, we define the semantics of neighborhood terms.

\begin{enumerate}
	\item $\ip{X}^{(G,\beta)} = \beta(X)$,
    \item $\ip{N^r_d(t)}^{(G,\beta)} = N^r_d(\ip{t}^{(G,\beta))})$ where the second $N^r_d$ is evaluated in $G$ (see \Cref{def:neighborhood:operator}),
	\item $\ip{\mathbf{P}}^{(G,\beta)} = \mathbf{P}(G)$,
	\item $\ip{\emptyset}^{(G,\beta)} = \emptyset$,
	\item $\ip{\comp t}^{(G,\beta)} = V(G) \setminus \ip{t}^{(G,\beta)}$,
	$\ip{t_1 \star t_2}^{(G,\beta)} = \ip{t_1}^{(G,\beta)} \star \ip{t_2}^{(G,\beta)}$ for $\star \in \{\cap,\cup, \setminus\}$,
\end{enumerate}
\en inherits the semantics from \MSOone, with the following semantics of the additional rules.
\begin{enumerate}
	\setcounter{enumi}{5}
	\item $\ip{|t| \prec m}^{(G,\beta)} = 1$ if $|\ip{t}^{(G,\beta)}| \prec m$ and $\ip{|t| \prec m}^{(G,\beta)} = 0$ otherwise, for $\prec \in \{=,\le,\ge\}$,
\item $\ip{t_1 \prec t_2}^{(G,\beta)} = 1$ if $\ip{t_1}^{(G,\beta)} \prec \ip{t_1}^{(G,\beta)}$ and $\ip{t_1 \prec t_2}^{(G,\beta)} = 0$ otherwise, for $\prec \in {\{=,\subseteq,\supseteq\}}$.
\end{enumerate}
For \ac \en, the semantics of the additional rule are as follows. 
\begin{enumerate}
	\setcounter{enumi}{7}
	\item $\ip{ \conn(t)}^{(G,\beta)}=1$ if $G[\ip{t}^{(G,\beta)}]$ is connected and $\ip{ \conn(t)}^{(G,\beta)}=0$ otherwise, \\
	      $\ip{ \acy(t)}^{(G,\beta)}=1$ if $G[\ip{t}^{(G,\beta)}]$ is acyclic and $\ip{ \acy(t)}^{(G,\beta)}=0$ otherwise.
\end{enumerate}

For a formula or term $\xi$, we write $\xi(X_1,\dots,X_k)$ to indicate that its free variables are exactly the set variables $X_1,\dots,X_k$.
For a graph $G$, formula or term $\xi$ and tuple $\tB \in \Pset(V(G))^{k}$ we write $\ip{\xi(\tB)}^{G}$ as a synonym for 
$\ip{\xi}^{G,\beta}$, where $\beta$ assigns $X_i$ to $B_i$ for all $1 \le i \le k$.
We write $G \models \phi(\tB)$ as a shorthand for $\ip{\phi(\tB)}^G = 1$.

\subsection{Core logic}
Some of the operations in our logic can be understood as ``syntactic sugar'', that does not increase the expressiveness, but merely reduces some friction when expressing problems.
To facilitate our proofs, we consider a smaller fragment \encore (core distance neighborhood logic) of \en that has the same expressive power as \en,
and a similar equivalent fragment \ac \encore of \ac \en.
To this end, we first describe a procedure that simplifies \ac \en formulas.
Assume we start with a \ac \en formula $\phi$.
First, we get rid of all vertex quantifiers.
\begin{itemize}
    \item A vertex variable $x$ can be replaced by a set variable $X$ with $|X| = 1$, where
    $x \in Y$ translates to $X \subseteq Y$,
    $x=y$ translates to $X=Y$, and
    $E(x,y)$ translates to $X\neq Y \land X \subseteq N_1^1(Y)$.
\end{itemize}

Let $Y_\emptyset$ be a new set variable.
We can construct a formula $\phi'$ from $\phi$ by replacing every occurrence of $\emptyset$ with $Y_\emptyset$.
Then $\phi$ is equivalent to $\exists Y \, |Y_\emptyset| \le 0 \land \phi'$.
Next, we exhaustively apply the following simplifications:
\begin{itemize}
	\item $|t| = m$ can be replaced with $|t| \le m \land |t| \ge m$,
	\item $|t| \ge m$ can be replaced with $\neg(\abs{t}\leq m-1)$,
    \item $t_1 \supseteq t_2$ can be replaced with $\comp{t_1} \cap t_2 = Y_\emptyset$.
    \item $t_1 \subseteq t_2$ can be replaced with $t_1 \cap \comp{t_2} = Y_\emptyset$.
    \item $t_1 \cup t_2$ can be replaced with $\overline{\comp{t_1} \cap \comp{t_2}}$.
    \item $t_1 \setminus t_2$ can be replaced with $t_1 \cap \comp{t_2}$.
\end{itemize}

Thus, we also get rid of the $\cup$, $\setminus$, $\subseteq$, $\supseteq$, $\ge m$ and $= m$ operators.
Next, we convert $\phi$ into prenex-normal form. 
Then we can exhaustively apply the following simplification:
\begin{itemize}
    \item
        The formula $\phi$ is in prenex-normal form, i.e., of the form $\exists X_1 \, \dots \exists X_k \, \psi$, where $\psi$ is quantifier-free.
        If $\psi$ contains a term $N^r_d(t)$, $t \cap t'$, $\comp t$, $\conn(t)$, $\acy(t)$ where $t$ is not a variable
        then we can replace $\psi$ by 
        $\exists Y \, t=Y \land \psi'$, where $Y$ is some unused variable and $\psi'$ is obtained from $\psi$ by replacing $t$ with $Y$.
\end{itemize}

This means, we removed nested terms
and can therefore assume that all terms in $\phi$ are of the form $\mathbf P$, $N^r_d(X)$, $X \cap Y$, $\comp X$, $\conn(X)$ and $\acy(X)$.
This invites the following definition of our core neighborhood logics \encore and \ac \encore.

\myparagraph{Definition of core logics}
We first define \emph{primitive formulas}.
\begin{enumerate}
    \item If $\mathbf{P}$ is a unary relational symbol and $X$ is a variable then $\mathbf{P}=X$ is a \emph{primitive formula}.
    \item If $X$, $Y$ and $Z$ are variables then $X = Y$, $X = \comp Y$, 
        and $X \cap Y = Z$ are \emph{primitive formulas}.
    \item If $X$ and $Y$ are variables and $m,d \in \N$ then $N^r_d(X) = Y$ is a \emph{primitive formula}.
    \item If $X$ is a variable and $m \in \N$ then $|X| \le m$ is a \emph{primitive formula}.
    \item If $X$ is variable then $\acy(X)$ and $\conn(X)$ are \emph{primitive \ac formulas}.
\end{enumerate}

Let \encore be the fragment of \en containing all formulas of the form
$\exists X_1 \, \dots \exists X_k \, \psi$, where $\psi$ is a Boolean combination
of primitive formulas as described by the items 1.\ to 4.
We define \ac \encore to be the fragment of \ac \en containing all formulas of the form
$\exists X_1 \, \dots \exists X_k \, \psi$, where $\psi$ is a Boolean combination
of primitive formulas as described by the items 1.\ to 5.
The following observation follows from applying the aforementioned simplifications.

\begin{observation}\label{obs:core}
    For every formula $\phi \in \text{\ac \en}$ one can compute in time $O(|\phi|^2)$ an equivalent formula $\phi' \in \text{\ac \encore}$ with
    \begin{itemize}
        \item $|\phi'| \le 10|\phi|$,
        \item $\var(\phi') \le |\phi|$,
        \item $R(\phi') = R(\phi)$,
        \item $d(\phi') = d(\phi)$,
        \item $\prod_{i \in [\ell']} (m'_i+2) \le 2\prod_{i\in[\ell]} (m_i+2)^2$, 
            where $m'_1,\dots,m'_{\ell'}$ are the values of size measurements in $\phi'$ and
            $m_1,\dots,m_\ell$ are the values of size measurements in $\phi$.
    \end{itemize}
    Moreover, if $\phi \in \text{\en}$ then $\phi' \in \text{\encore}$.
\end{observation}
\section{Toolkit for dynamic programming algorithms}
\label{sec:genericAlgo}
This section provides tools to bound the run time of algorithms relying on recursive graph decompositions.
In this paper, we focus on optimization problems whose solutions are elements of $\Pset(V(G))^k$, i.e., $k$-tuples of vertex sets of the input graph $G$ and where the recursive decomposition is a rooted layout, but these tools can be adapted to other settings as well.

Recall that for a $k$-weighted graph $G$ and a $k$-tuple $\tB\in \Pset(V(G))^{k}$, 
the weight of $\tB$ in $G$ is
$\obj(G,\tB) = \sum_{v \in V(G)} w(v,\{ i \mid v \in B_i \})$.
\begin{definition}\label{def:kproblem}
    For a constant $k$, a \emph{$k$-problem} $\Pi$ associates every $k$-weighted graph $G$ with a set of \emph{solutions} $\Pi(G) \in \Pset(V(G))^{k}$. 
    For a $k$-weighted graph $G$, an \emph{optimal solution} to $\Pi$ is 
    a tuple $\tB^* \in \Pi(G)$ that
    maximizes the weight, i.e., such that $\obj_G(\tB^*) = \max \{\obj_G(\tB) \mid \tB \in \Pi(G) \}$.
\end{definition}
Our tools are based on the following notion of representativity between tuples of vertices.
\begin{definition}\label{def:represent}
    Let $\Pi$ be a $k$-problem, $G$ be a $k$-weighted graph and $A\subseteq V(G)$.
	Given $\cB\subseteq \Pset(A)^{k}$ and $\tD\in \Pset(\comp{A})^{k}$, we define 
    $\best_{\Pi,G}(\cB,\tD)=\max\{\obj_G(\tB)\mid \tB\in \cB \text{ and } \tB\cup \tD\in \Pi(G)\}.$
    For $\cB_1,\cB_2\subseteq \Pset(A)^{k}$, we say that \emph{$\cB_1$ $\Pi$-represents $\cB_2$ over $A$} if for every 
    $\tD\in \Pset(\comp{A})^{k}$, we have $\best_{\Pi,G}(\cB_1,\tD)=\best_{\Pi,G}(\cB_2,\tD)$.
    We drop $\Pi$ and $G$ from these notations if they are clear from the context.
\end{definition}
Observe that if there is no $\tB\in \cB$ such that $\tB\cup \tD\in\Pi(G)$, then $\best(\cB,\tD)=\max(\emptyset)=-\infty$.
It is easy to see that the relation ``represents over $A$'' is an equivalence relation.

To solve $\Pi$ on a $k$-weighted graph $G$ with a rooted layout $(T,\delta)$, one can use a
standard dynamic programming algorithm that does a bottom-up traversal of $T$
and computes, for each node $t$ of $T$, a set $\cB_t$ that represents $\Pset(V_t)^{k}$ over $V_t$.
For a leaf $t$, we can choose $\cB_t=\Pset(V_t)^{k}$ since $\abs{V_t}=1$.
	For an internal node $t$ with children $a$ and $b$, 
	we aim to inductively compute $\cB_t$ from $\cB_a$ and $\cB_b$,
	where $\cB_a$ represents $\Pset(V_a)^k$ over $V_a$
	and $\cB_b$ represents $\Pset(V_b)^k$ over $V_b$.
	However, in this process we have to make sure that the size of $\cB_t$ remains small.
	If we succeed, then at the root node $r$ of $T$, we have a small set $\cB_r$ that represents $\Pset(V(G))^{k}$ over $V(G)$
	and it suffices to search for the optimal solution to $\Pi$ in $G$ among the elements of $\cB_r$ instead of 
	all tuples in $\Pset(V(G))^{k}$.

For two sets $\cB_1,\cB_2\subseteq \Pset(V(G))^{k}$, we denote the product between these sets by $\cB_1\otimes \cB_2= \{\tB_1 \cup \tB_2 \mid \tB_1\in \cB_1 \text{ and } \tB_2\in \cB_2 \}$.
For an internal node $t$ of $T$ with children $a$ and $b$, 
the following lemma shows that we could obtain a representative set $\cB_t$ as a product $\cB_a \otimes \cB_b$
(with the obvious downside that 
$\abs{\cB_1\otimes \cB_2} = \abs{\cB_1}\cdot\abs{\cB_2}$ and therefore one cannot rely only on $\otimes$ to compute small representative sets).

\begin{lemma}\label{lem:merge}
	Let $\Pi$ be a $k$-problem, $G$ be a $k$-weighted graph, $A_1,A_2$ be two disjoint subsets of $V(G)$, and $\cB_1\subseteq \Pset(A_1)^{k}$, $\cB_2\subseteq \Pset(A_2)^{k}$.
    If $\cB_1$ represents $\Pset(A_1)^{k}$ over $A_1$ and $\cB_2$
    represents $\Pset(A_2)^{k}$ over $A_2$ then $\cB_1\otimes \cB_2$
    represents $\Pset(A_1\cup A_2)^{k}$ over $A_1\cup A_2$.
\end{lemma}
\begin{proof}
	Assume $\cB_i\subseteq \Pset(A_i)^{k}$ represents $\Pset(A_i)^{k}$ over $A_i$ for $i\in\{1,2\}$.
    To prove that $\cB_1\otimes \cB_2$ represents $\Pset(A_1\cup A_2)^{k}$ over
    $A_1\cup A_2$, we need to prove that for every $\tD\in \Pset(\comp{A_1\cup A_2})^{k}$ we have $\best(\cB_1\otimes\cB_2,\tD)=\best(\Pset(A_1\cup A_2)^{k},\tD)$.
	Let $\tD\in \Pset(\comp{A_1\cup A_2})^{k}$.
	From the definition of $\best$ and the linearity of the weight term, we deduce that
	\begin{align*}
        \hspace{-0.3cm}\best(\cB_1\otimes\cB_2,\tD) &= \max \{\obj(\tB_1\cup \tB_2)\mid \tB_1\in\cB_1, \tB_2\in\cB_2,  (\tB_1\cup\tB_2\cup\tD)\in \Pi(G)\}\\
                                         &= \max\{\best(\cB_1,\tB_2\cup \tD) + \obj(\tB_2) \mid \tB_2 \in \cB_2\}.
	\end{align*}
    Since $\cB_1$ represents $\Pset(A_1)^{k}$ over $A_1$, can substitute
	\begin{align*}
        \hspace{-0.3cm}\best(\cB_1\otimes\cB_2,\tD) &= \max\{\best(\Pset(A_1)^{k},\tB_2\cup \tD) + \obj(\tB_2) \mid \tB_2 \in \cB_2\}\\
        {} &= \max \{\obj(\tB_1\cup \tB_2)\mid \tB_1\in\Pset(A_1)^{k}, \tB_2\in\cB_2,  (\tB_1\cup\tB_2\cup\tD)\in \Pi(G)\}\\
        {} &= \best(\Pset(A_1)^{k}\otimes \cB_2, \tD).
	\end{align*}
    Symmetrically, since $\cB_2$ represents $\Pset(A_2)^{k}$ over $A_2$, $\best(\cB_1\otimes\cB_2,\tD)=\best(\Pset(A_1)^{k}\otimes\Pset(A_2)^{k},\tD)$.
	Observe that further $\Pset(A_1\cup A_2)^{k}=\Pset(A_1)^{k}\otimes \Pset(A_2)^{k}$.
	Consequently, we have $\best(\cB_1\otimes\cB_2,\tD)=\best(\Pset(A_1\cup A_2)^{k},\tD)$.
    As this holds for every $\tD\in \Pset(\comp{A_1\cup A_2})^{k}$, this last equality proves that $\cB_1\otimes\cB_2$ represents $\Pset(A_1\cup A_2)^{k}$ over $A_1\cup A_2$.
\end{proof}

	As mentioned above, we require a mechanism that computes a small representative set $\cB_t$ 
	at each internal node $t \in V(T)$ with children $a$ and $b$.
	Using the previous lemma, the starting point of this computation is $\cB_a \otimes \cB_b$
	(where $\cB_a$ and $\cB_b$ are the corresponding representative sets),
	which might become prohibitively large.
	The way we compute a small enough subset that still represents $\cB_a \otimes \cB_b$
	typically depends on the concrete application setting.
	We therefore define the following notion of a \emph{reduce routine} as a placeholder which,
	if implemented accordingly, will yield desired bounds on the size of $\cB_t$ as well as on the time 
	needed to compute $\cB_t$ from $\cB_a \otimes \cB_b$.
\begin{definition}\label{def:reduce}
    Let $\Pi$ be a $k$-problem.
    A \emph{reduce routine for $\Pi$} is a subroutine $\reduce_{G,A}(\cB)$ that takes as input 
    \begin{itemize}
        \item a $k$-weighted graph $G$ together with its distance matrix, and
        \item sets $A\subseteq V(G)$, $\cB\subseteq \Pset(A)^{k}$,
    \end{itemize}
    and returns a subset $\reduce_{G,A}(\cB) = \reduce(\cB) \subseteq \cB$ such that $\reduce(\cB)$ represents $\cB$ over~$A$.
    We say a function $\s(G,A)$ is a \emph{size bound} to the reduce routine if $\abs{\reduce(\cB)}\leq \s(G,A)$ for all $k$-weighted graphs $G$ and $A \subseteq V(G)$.
    We further say a function $\f(G,A)$ is a \emph{run time bound} to the reduce routine if
    it runs on every input in time at most $\f(G,A)\cdot \abs{\cB}$.
    Given a rooted layout $\cL=(T,\delta)$ of a graph $G$,
    we define $\s(G,\cL) = \max\{\s(G,V_x) \mid x \in V(T)\}$
    and $\f(G,\cL) = \max\{\f(G,V_x) \mid x \in V(T)\}$.
\end{definition}
Note that the distance matrix could be computed in polynomial time by the reduce routine itself,
but is given as input to speed up the computation by a polynomial factor.
If we have access to a reduce routine, we 
can compute $\cB_t$ from $\cB_a$ and $\cB_b$ by setting $\cB_t=\reduce_{G,A}(\cB_a\otimes\cB_b)$ (\Cref{algo:generic}).
We now show that problems that admit reduce routines with small size and run time bounds
can be solved efficiently.
Note that after applying \Cref{thm:genericAlgo}, we merely have to evaluate
$\Pi$ for at most $s(G,\cL)$ tuples and take one with the largest weight.

\begin{theorem}\label{thm:genericAlgo}
    If a $k$-problem $\Pi$ admits reduce routine with size bound $\s$ and run
    time bound $\f$ then there exists an algorithm that, given a $k$-weighted $n$-vertex graph $G$ and a rooted
    layout $\cL$ of $G$, computes in time 
    $O\left(2^{2k}k \cdot s(G,\cL)^2  \cdot (\f(G,\cL)\cdot n + n^2) +n^3 \right)$
    a set that has size at most $s(G,\cL)$ and contains an optimal solution to $\Pi$ if $\Pi$ has a solution.
\end{theorem}
\begin{proof}
    Let $\Pi$ be a $k$-problem admitting a reduce routine $\reduce$ with size bound $\s$ and run time bound $\f$,
    and $G$ be a $k$-weighted graph with a rooted layout $\cL=(T,\delta)$.
	We claim that \Cref{algo:generic} solves the problem.

\SetAlCapSkip{1em} %adds spacing before the caption.
\begin{algorithm}
	\SetAlgoLined
	Compute the distance matrix of $G$ using the Floyd-Warshall algorithm\label{line:distmatrix}\;
	\For{every leaf $t$ of $T$}
	{
		$\cB_t=\Pset(V_t)^{k}$\;
	}
    \For{every internal node $t$ of $T$ with children $a,b$ in a bottom-up traversal of $T$}
	{	
        Compute $\cB_t=\reduce_{G,V_t}(\cB_a\otimes \cB_b)$\label{line:reduce}\;
	}
    \Return $\cB_r$, where $r$ is the root of $T$\;
    \caption{Solving a $k$-problem $\Pi$ using a reduce routine $\reduce$, a $k$-weighted graph $G$, and a rooted layout $(T,\delta)$ of $G$.}
    \label{algo:generic}
\end{algorithm}

\myparagraph{Correctness} We claim that for every node $t$ of $T$, $\cB_t$ represents $\Pset(V_t)^{k}$ over $V_t$.
    This is obviously true for the leaves of $T$. Assume that $t$ is an internal node with children $a,b$ and for $i\in\{a,b\}$ suppose that $\cB_i$ represents $\Pset(V_i)^{k}$ over $V_i$.
    Observe that when $\reduce_{G,V_t}(\cB_a\otimes \cB_b)$ is called at
    line~\ref{line:reduce}, we give as input the distance matrix of $G$
    computed at line~\ref{line:distmatrix}.
    By \Cref{lem:merge}, we know that $\cB_a\otimes \cB_b$ represents $\Pset(V_t)^{k}$ over $V_t=V_a\cup V_b$.
    Moreover, $\reduce_{G,V_t}(\cB_a\otimes \cB_b)$ represents $\cB_a\otimes\cB_b$ over $V_t$.
    Since ``represents over $V_t$'' is a transitive relation, we deduce that $\cB_t=\reduce_{G,V_t}(\cB_a\otimes \cB_b)$ represents $\Pset(V_t)^{k}$ over $V_t$.
    By induction, we conclude that $\cB_t$ represents $\Pset(V_t)^{k}$ over $V_t$ for all nodes $t$ of $T$.
	
    Let $\emptyset^k$ be the $k$-tuple $(\emptyset,\dots,\emptyset)$.
    Since for the root $r$ of $T$ holds $V_r=V(G)$, we conclude $\max\{\obj(\tB)
    \mid \tB\in \cB_r \cap \Pi(G) \} =
    \best(\cB_r,\emptyset^k)=\best(\Pset(V(G))^{k},\emptyset^k) =
    \max\{\obj(\tB) \mid \tB\in \Pi(G)\}$.
    This proves the correctness of \Cref{algo:generic}.

	\myparagraph{Run time}
    Computing the distance matrix of $G$ using the Floyd-Warshall algorithm takes $O(n^3)$ time~\cite{cormen01introduction}.
    For a leaf $t$ of $T$, since $|V_t|=1$, we can compute $\cB_t = \Pset(V_t)^{k}$ in time $O(2^k)$. 
    Let $t$ be an internal node of $T$ with children $a,b$.
    We have $\abs{\cB_a \otimes \cB_b} = \abs{\cB_a}\cdot \abs{\cB_b}\leq \max(2^{2k}, \s(G,V_a)^2)\leq 2^{2k}\s(G,\cL)^2$.
    Consequently, computing the set $\cB_a\otimes \cB_b$  takes time $O(2^{2k} \cdot \s(G,\cL)^2\cdot k \cdot n)$.
    Computing $\reduce_{G,V_t}(\cB_a\otimes\cB_b)$ with run time bound $\f$ takes time $\f(G,V_t)\cdot\abs{\cB_a\otimes\cB_b} = O(\s(G,\cL)^2 \cdot \f(G,\cL))$.
	Since every rooted layout has $2n-1$ nodes, we conclude that \Cref{algo:generic} runs in the claimed run time. 
\end{proof}
\section{Algorithmic meta-theorems}
\label{sec:DN:metaAlgo}
For a (\ac) \en formula $\phi$, we simplify the notations for the $(d(\phi),R(\phi))$-neighbor equivalence by using $\phi$ as a synonym for $d(\phi),R(\phi)$. 
For example, we denote $\nec_{d(\phi),R(\phi)}(A)$ by $\nec_{\phi}(A)$ and $\tB\equi{A}{d(\phi),R(\phi)} \tC$ by $\tB\equi{A}{\phi} \tC$.
We associate with every \ac \en formula $\phi(X_1,\dots,X_k)$
a $k$-problem $\Pi_\phi$ such that $\Pi_\phi(G) = \{ \tB \in \Pset(V(G))^{k} \mid G \models \phi(\tB)\}$.

\subsection{Distance neighborhood logic}
We start with a base version of our meta theorem without acyclicity or connectivity constraints.
We restrict our attention first to \encore formulas and then lift the result to \en.
We will need the following family of equivalence relations based on $\equi{A}{\phi}$.
\begin{lemma}\label{lem:equivalenceformula}
    Let $\phi(X_1,\dots,X_k)$ be a quantifier-free \encore formula with $\ell$ size measurements $\abs{t_1}\leq m_1,\dots,\abs{t_\ell}\leq m_\ell$.
	Let $A\subseteq V(G)$. For every $\btE \in (\Rep{\comp{A}}{\phi})^{k}$, there exists an equivalence relation $\bowtie_{\btE}$ over $\Pset(A)^{k}$ with at most
	$\nec_{\phi}(A)^{k}\cdot 2^{\abs{\phi}} \cdot \prod_{1\leq i \leq \ell} (m_i+2)$ equivalence classes such that for every
	$\tD \in \Pset(\comp{A})^{k}$ with $\tD\equi{\comp{A}}{\phi} \btE$
	and $\tB, \tC \in \Pset(A)^{k}$ with $\tB \bowtie_{\btE} \tC$, we have 
	$G \models \phi(\tB \cup \tD)$ if and only if 
	$G \models \phi(\tC \cup \tD)$.
\end{lemma}
\begin{proof}
	We define the equivalence relation $\bowtie_{\btE}$ over $\Pset(A)^{k}$ such that $\tB\bowtie_{\btE} \tC$ if the following conditions are satisfied:
	\begin{enumerate}[(A)]
		\item\label{item:equiphi} $\tB\equi{A}{\phi} \tC$.
		\item\label{item:size:measurement} For every $i\in [\ell]$, we have $\max(m_i+1,\abs{ \ip{t_i(\tB \cup \btE)}^G \cap A})=\max(m_i+1,\abs{\ip{t_i(\tC \cup \btE)}^G \cap A})$.
		\item\label{item:primitive:formula} For every primitive formula of the form $t= t'$ in $\phi$, 
        ${\ip{t(\tB\cup \btE)}^G \cap A = \ip{t'(\tB\cup \btE)}^G \cap A}$ iff $\ip{t(\tC\cup \btE)}^G \cap A = \ip{t'(\tC\cup \btE)}^G \cap A$
	\end{enumerate}

	Since $\phi$ has at most $\abs{\phi}$ primitive formulas,
	the number of equivalence classes of $\bowtie_{\btE}$ is at most $\nec_{\phi}(A)^{k}\cdot \prod_{1\leq i \leq \ell} (m_i+2) \cdot 2^{\abs{\phi}}$.
	Let $\tD \in \Pset(\comp{A})^{k}$ and $\tB, \tC\in \Pset(A)^{k}$ such that $\tD\equi{\comp{A}}{\phi} \btE$ and $\tB\bowtie_{\btE} \tC$.
	To prove this lemma, it remains to prove that $G \models \phi(\tB \cup \tD)$ if and only if $G \models \phi(\tC \cup \tD)$.
	We start by proving the following claim.
	\begin{claim}\label{claim:terms}
		For every term $t$ of $\phi$ the following three equalities hold.
		\begin{align}
			\ip{t(\tB \cup \tD)}^{G} \cap \comp{A} &= \ip{t(\tC \cup \tD)}^{G} \cap \comp{A} \label{eq:overcompA} \\
			\ip{t(\tB \cup \tD)}^{G} \cap A & = \ip{t(\tB \cup \btE)}^{G} \cap A   \label{eq:overA1} \\
			\ip{t(\tC\cup \tD)}^{G} \cap A &= \ip{t(\tC \cup \btE)}^{G} \cap A \label{eq:overA2}
		\end{align} 
	\end{claim}
	\begin{claimproof}
		Let $t$ be a term of $\phi$. We do a case distinction on the structure of $t$. 
		
		\myparagraph{When $t \equiv \bfP$}
		The equations hold since $\ip{t}^G = \ip{\textbf{P}}^G=\bfP(G)$ is independent of $\tB,\tC,\tD,\btE$.
		
		\myparagraph{When $t \equiv X$ or $t \equiv \comp{X}$ or $t \equiv X \cap Y$}
		The equations hold because $B_i,C_i\subseteq A$ and $D_i,\bE_i\subseteq \comp{A}$ for every $1 \leq i \leq k$.
		Then in all three cases we have (\ref{eq:overcompA})~$t(\tB\cup \tD)\cap \comp{A} = t(\tD) = t(\tC\cup \tD)\cap \comp{A}$,
		(\ref{eq:overA1})~$t(\tB \cup \tD) \cap A=t(\tB)=t(\tB \cup \btE) \cap A$ and (\ref{eq:overA2})~$t(\tC \cup \tD) \cap A = t(\tC) =t(\tC \cup \btE) \cap A$.
		
		\myparagraph{When $t \equiv N_d^r(X_i)$}
		Since $B_i$ and $D_i$ are disjoint and $N^r_d(B_i\cup D_i)$ contains all vertices with at least $d$ neighbors in $B_i\cup D_i$ in $G^r$, we deduce that
		\begin{align*}
			t(\tB\cup \tD)\cap \comp{A}= N^r_d(B_i\cup D_i)\cap \comp{A}= &\bigcup_{\substack{d_1,d_2\in \N \\ d_1+d_2=d}} 
			N^r_{d_1}(B_i) \cap N^r_{d_2}(D_i) \cap \comp{A},
		\end{align*}
		where $N_0^r(W)=V(G)$ for every $W\subseteq V(G)$.
		Since $\tB\equi{A}{\phi} \tC$, we have $B_i\equi{A}{d,r} C_i$.
		By definition of $\equi{A}{d,r}$, for every $v\in \comp{A}$ we have $\min{d,\abs{N^r(v)\cap B_i}}=\min{d,\abs{N^r(v)\cap C_i}}$.
		Thus, for every $d_1\leq d$ we have $N^r_{d_1}(B_i)\cap \comp{A} = N^r_{d_1}(C_i)\cap \comp{A}.$
		Therefore,
		\begin{align*}
			t(\tB\cup \tD)\cap \comp{A} &=\bigcup_{\substack{d_1,d_2\in \N \\ d_1+d_2=d}} 
			N^r_{d_1}(C_i) \cap N^r_{d_2}(D_i) \cap \comp{A}
			=N^r_d(C_i\cup D_i)\cap \comp{A} = t(\tC\cup \tD)\cap \comp{A}.
		\end{align*}
		This proves (\ref{eq:overcompA}). The proof of (\ref{eq:overA1}) and (\ref{eq:overA2}) is symmetrical since $\tD \equi{\comp{A}}{\phi} \btE$.
	\end{claimproof}
	
	From condition~(\ref{item:size:measurement}) and \Cref{claim:terms}, we deduce that for every $1\leq i\leq \ell$, we have 
	$\ip{(\abs{t_i}\leq m_i)(\tB\cup \tD)}^G=\ip{(\abs{t_i}\leq m_i)(\tC\cup \tD)}^G$.
	Moreover, from condition~(\ref{item:primitive:formula}) and \Cref{claim:terms}, we deduce that for every primitive subformula $\psi$ in $\phi$ of the form $\psi\equiv t= t'$, we have $\ip{\psi(\tB\cup\tD)}^G = \ip{\psi(\tB \cup \tD)}^G$.
	Since $\phi$ is a Boolean combination of its primitive formulas, for every $\tD\in \{\tB,\tC\}$, the value of $\ip{\phi(\tD\cup \tD)}^G$ is entirely determined by the values of $\ip{\psi(\tD\cup \tD)}^G$ for every primitive formula $\psi$ of $\phi$.
	We conclude that $G \models \phi(\tB \cup \tD)$ if and only if $G \models \phi(\tC \cup \tD)$.
\end{proof}

In the following, we prove \Cref{thm:main:basic} by designing a reduce routine for $\Pi_\phi$. 

\begin{lemma}\label{lem:reducelogic}
    Let $\phi(X_1,\dots,X_k)$ be a quantifier-free \encore formula whose size measurements have values $m_1,\dots,m_\ell$.
	The $k$-problem $\Pi_\phi$ admits a reduce routine with size bound
	\[
	\s(G,A) = \snec^G_{\phi}(A)^{2k}\cdot 2^{\abs{\phi}} \cdot  \prod_{1\leq i\leq \ell} (m_i+2)
	\]
	and run time bound 
	\[
    \f(G,A) =
    O(|\phi|^2 \cdot \snec^G_{\phi}(A)^{k+1} \cdot n^2).
	\]
\end{lemma}
\begin{proof}
	Let $G$ be a $\abs{\phi}$-weighted graph, $A\subseteq V(G)$ and $\cB\subseteq \Pset(A)^{k}$.
	Let us explain how to compute a subset $\reduce_{G,A}(\cB)= \reduce(\cB)$ of $\cB$ in time $\abs{\cB}\cdot \f(G,A)$ such that $\reduce(\cB)$ represents $\cB$ over $A$ and $\abs{\reduce(\cB)}\leq \s(G,A)$.
	
	For every $\btE\in (\Rep{\comp{A}}{\phi})^{k}$, let $\bowtie_\btE$ be the equivalence relation given by \Cref{lem:equivalenceformula}.
	Let further $\reduce(\cB,\btE)$ be a set that, for every equivalence class $\cC$ of $\bowtie_\btE$ over $\cB$, contains a tuple $\tB \in \cC$ such that $\obj(\tB)$ is maximal among all elements in $\cC$.
    We will first compute $\reduce(\cB,\btE)$ for all $\btE\in (\Rep{\comp{A}}{\phi})^{k}$
    and then output the set $\reduce(\cB)$ as the union of all $\reduce(\cB,\btE)$ over all $\btE\in (\Rep{\comp{A}}{\phi})^{k}$.
	
	\myparagraph{Representativity}
	To prove that $\reduce(\cB)$ represents $\cB$ over $A$, we need to show that for every $\tD\in \Pset(\comp{A})^{k}$, we have $\best(\reduce(\cB),\tD) \ge \best(\cB,\tD)$.
	Let $\tD\in \Pset(\comp{A})^{k}$.
	If $\best(\cB,\tD)=-\infty$ then the statement holds, thus suppose that $\best(\cB,\tD)\neq-\infty$. 
    Then, there exists $\tB\in \cB$ such that $G \models \phi(\tB\cup \tD)$ and $\obj(\tB)=\best(\cB,\tD)$.
	Let $\btE\in (\Rep{\comp{A}}{k})$ such that $\tD\equi{\comp{A}}{\phi} \btE$.
	By construction, $\reduce(\cB,\btE)$ contains a solution $\tC$ such that $\tB \bowtie_\btE \tC$ and $\obj(\tB)\leq \obj(\tC)$.
	From \Cref{lem:equivalenceformula}, we know that $G\models \phi(\tB\cup \tD)$ iff $G\models \phi(\tC\cup \tD)$.
	Thus, $G\models \phi(\tC\cup \tD)$ and we have
    $\obj(\tC)\ge\best(\reduce(\cB),\tD)$.

	\myparagraph{Size bound}
    We claim that $\abs{\reduce(\cB)}\leq \s(G,A)$. By \Cref{lem:equivalenceformula}, each relation $\bowtie_\btE$ has 
    $\nec_{\phi}(A)^{k}\cdot 2^{|\phi|} \cdot \prod_{1\leq i \leq \ell} (m_i+2)$ equivalence classes.
	Since $\abs{\Rep{\comp{A}}{\phi}}=\nec_{\phi}(\comp{A})$ and $\nec_{\phi}(A)\cdot \nec_{\phi}(\comp{A})\leq \snec_{\phi}(A)^2$, we conclude that $\abs{\reduce(\cB)}\leq \s(G,A)$.
	
	\myparagraph{Run time bound} 
	According to \Cref{def:reduce}, the reduce routine we design gets the distance matrix of $G$ as part of the input.
	We claim that computing $\reduce(\cB)$ can be done in time $O(\abs{\cB}\cdot \f(G,A))$.
	Observe that for every term $t$ of $\phi$ and interpretation $\beta$ of $\phi$, we have $\ip{\phi}^{G,\beta}\subseteq V(G)$. 
    Consequently, $\abs{ \ip{\phi}^{(G,\beta)} } \leq n$ and we can assume w.l.o.g. that for every $1\leq i \leq \ell$, we have $m_i\leq n$. 
	To compute $\reduce(\cB)$ in time $O(\abs{\cB}\cdot \f(G,A))$, we do the following computations.
	\begin{itemize}
		\item Using \Cref{lem:computerep} and $\abs{R(\phi)}\leq \abs{\phi}$, we compute
        in time $O(\snec_{\phi}(A)\cdot \abs{\phi} \cdot n^2 \cdot \log(\snec_{\phi}(A)))$ 
        the sets $\Rep{A}{\phi}$, $\Rep{\comp{A}}{\phi}$ and the data structures that, given a set $B\subseteq A$, computes a pointer to $\rep{A}{\phi}(B)$ in time $O(n^2\cdot \log(\nec_{\phi}(A)))$.
		%As we are given the distance matrix of $G$ and $\abs{R(\phi)}\leq \abs{\phi}$, 
		%computing these sets and data structures can be done in $O(\abs{\phi}\cdot \snec_{\phi}(A)\cdot n^2 \cdot \log(\snec_{\phi}(A)))$ time.
		
    \item For every $1\leq i\leq \ell$, $\tB\in \cB$ and $\btE\in (\Rep{\comp{A}}{\phi})^{k}$, we compute a list $L_{\tB}$ containing the pointers to $\rep{A}{\phi}(B_1),\dots,\rep{A}{\phi}(B_{k})$
        and we compute $g(\btE,\tB,i)=\max(m_i+1, \abs{\ip{t_i(\tB\cup \btE)}^G \cap A})$.
        Let $t_1=t_1',\dots,t_p=t_p'$ be all primitive formulas of this form in $\phi$.
        Let $h(\btE,\tB,i) = 1$ if $\ip{t_i(\tB\cup \btE)}^G \cap A = \ip{t_i'(\tB\cup \btE)}^G \cap A$ and $h(\btE,\tB,i) = 0$ otherwise. 

        Observe that each $L_{\tB}$ can be computed in time $O(n^2\cdot\log(\snec_{\phi}(A)) \cdot |\phi|)$
		 thanks to the data structures of \Cref{lem:computerep}.
		For given $\tB,\btE$, we can compute $g(\btE,\tB,1),\dots,g(\btE,\tB,\ell)$ and $h(\btE,\tB,1),\dots,h(\btE,\tB,p)$ in time $O(|\phi|\cdot n^2)$ through basic algorithmic techniques.
        Thus, computing all sets $L_{\tB}$ takes $O(\abs{\cB}\cdot n^2 \cdot \log(\snec_{\phi}(A)) \cdot |\phi|)$ time 
        and computing all values $g(\btE,\tB,i), h(\btE,\tB,i)$ takes $O(\abs{\cB}\cdot \snec_{\phi}(\comp{A})^{k} \cdot n^2 \cdot |\phi|)$ time.
        The aggregated run time up until now is bounded by 
        \begin{equation}\label{eq:asdf3}
        O(\abs{\cB}\cdot \snec_{\phi}(\comp{A})^{k+1} \cdot n^2 \cdot |\phi|).
        \end{equation}

        \item 
            Remember that in the computational random access model we can perform additions and comparisons of weights in constant time.
            %Remember that $\obj(\tB) = \sum_{v \in V(G)} w(v,\{ i \mid v \in B_i \})$.
            This means we can compute $\obj(\tB)$ for all $\tB \in \cB$ in time $O(|\cB| \cdot n \cdot |\phi|)$.
            The aggregated run time up until now is still bounded by (\ref{eq:asdf3}).
            After computing these values, we can decide in constant time whether $\obj(\tC)<\obj(\tB)$
            for arbitrary $\tB,\tC \in \cB$.

		\item For every $\btE\in (\Rep{\comp{A}}{\phi})^{k}$, we define $\leq_\btE$ to be the total preorder\footnote{A total preorder is a binary relation that is reflexive, 
		connected and transitive.} on $\cB$ such that $\tB\leq_\btE \tC$ if the concatenation of
        \begin{itemize}
            \item $L_{\tB}$ 
            \item $g(\btE,\tB,1),\dots,g(\btE,\tB,\ell)$ 
            \item $h(\btE,\tB,1),\dots,h(\btE,\tB,p)$ 
        \end{itemize} is lexicographically smaller than the corresponding term for $\tC$.
		Observe that for every $\tB,\tC\in \cB$, we have $\tB\bowtie_\btE\tC$ iff $\tB$ and $\tC$ are equivalent for $\leq_\btE$.
        Moreover, we can decide whether $\tB\leq_\btE \tC$ in time $O(|\phi|)$
		as we just need to compare $k$ pointers and $\ell+p$ integers.
		
		For every $\btE\in (\Rep{\comp{A}}{\phi})^{k}$, we compute $\reduce(\cB,\btE)$ as follows.
		First, we initialize $\reduce(\cB,\btE)$ as an empty self-balanced binary search tree using $\leq_\btE$ as order.
		Then, for each $\tB\in \cB$, we search whether $\reduce(\cB,\btE)$ contains a tuple $\tC$ equivalent to $\tB$ for $\leq_\btE$.
		If such tuple $\tC$ exists and $\obj(\tC)<\obj(\tB)$, we replace $\tC$ by $\tB$ in $\reduce(\cB,\btE)$.
		Otherwise, if no such tuple $\tC$ exists, we insert $\tB$ into $\reduce(\cB,\btE)$.
        Computing $\reduce(\cB,\btE)$ this way takes 
        \begin{equation}\label{eq:asdf1}
            O(\abs{\cB}\cdot |\phi| \cdot \log(\abs{\cB/\bowtie_\btE})
        \end{equation} time, where 
        \[\abs{\cB/\bowtie_\btE}=\snec_{\phi}(A)^{k}\cdot 2^{\card{\phi}} \cdot \prod_{1\leq  i \leq \ell}(m_i+2)\]
		is the number of equivalence classes of $\bowtie_\btE$ over $\cB$, i.e., the final size of $\reduce(\cB,\btE)$.
        Since $\log(\prod_{1\leq  i \leq \ell}(m_i+2))=O(|\phi|+\log(n))$ and $\log(\snec_{\phi}(A)^{k})=O(k+\log(\snec_{\phi}(A)))$, 
        we have 
        \begin{equation}\label{eq:asdf2}
        \log(\abs{\cB/\bowtie_\btE}) = O(|\phi| + \log(\snec_\phi(A)) + \log(n)).
        \end{equation}
        By plugging (\ref{eq:asdf2}) into (\ref{eq:asdf1}), we conclude that computing $\reduce(\cB) = \bigcup_{\btE}\reduce(\cB,\btE)$ 
		can be done in \[
        O(\snec_{\phi}(\comp{A})^{k} \cdot \abs{\cB} \cdot |\phi| \cdot (|\phi| + \log(\snec_\phi(A)) + \log(n)))
        = O(\snec_{\phi}(\comp{A})^{k+1} \cdot \abs{\cB} \cdot |\phi|^2 \cdot \log(n)).
        \]
	\end{itemize}
    Summing the previous bound with (\ref{eq:asdf3}), we bound the total run time by
    \[O(|\cB| \cdot |\phi|^2 \cdot \snec^G_{\phi}(A)^{k+1} \cdot n^2) = \abs{\cB}\cdot \f(G,A).\]    
\end{proof}

We are ready to prove our meta-theorem concerning quantifier-free \encore logic.
\begin{lemma}\label{lem:metathmCore}
    There is an algorithm that computes for a given
    $k$-weighted graph $G$, 
    quantifier-free \encore formula~$\phi(X_1,\dots,X_k)$, 
    and rooted layout $\cL$ of $G$ a tuple $\tB^* \in \Pset(V(G))^{k}$ such that
    \[
    \obj_G(\tB^*) = \max \bigl\{ \obj_G(\tB) \bigm| \tB \in V(G)^{k}, G \models \phi(\tB) \bigr\}
    \]
    or concludes that no such tuple exists.
    If $\phi$ has size measurements $m_1,\dots,m_\ell$ then the run time of this algorithm is
    \[ 2^{O(|\phi|)} \cdot \snec_{\phi}(G,\cL)^{6k} \cdot  n^3  \cdot
    \prod_{1\leq i \leq \ell} (m_i+2)^2.\]
\end{lemma}
\begin{proof}
    Remember that an optimal solution to $\Pi_\phi$ is 
    a tuple $\tB^* \in \Pi(G)$ with $\obj(G,\tB^*) = \max \{\obj(G,\tB) \mid \tB \in \Pi_\phi(G) \}$.
    Thus, it is sufficient to find an optimal solution to $\Pi_\phi$ or conclude that no solution exists.
	We plug the reduce routine for $\Pi_\phi$ given by \Cref{lem:reducelogic} into \Cref{thm:genericAlgo} to compute
    in time 
    $O(2^{2|\phi|}|\phi| \cdot \s(G,\cL)^2 \cdot (\f(G,\cL)\cdot n + n^2) + n^3 )$
    a set that has size at most $s(G,\cL)$ and contains an optimal solution if $\Pi_\phi$ has a solution.
    In time $O(s(G,\cL) \cdot n \cdot |\phi|)$ we can find and return the entry from this set with the maximal objective value.

    Technically, this gives for every formula $\phi$ a \emph{separate} algorithm $A_\phi$.
    However, $A_\phi$ can be easily computed from $\phi$ in single exponential time.
    We obtain a single algorithm by first computing $A_\phi$ from $\phi$ and then executing it.
\end{proof}

Next, we lift it to \en logic.
\begin{theorem}\label{thm:metathm}
    There is an algorithm that computes for a given \en formula~$\phi(X_1,\dots,X_k)$, 
    \mbox{$k$-weighted} graph $G$, and rooted layout $\cL$ of $G$ a tuple $\tB^* \in \Pset(V(G))^{k}$ such that
    \[
    \obj_G(\tB^*) = \max \bigl\{ \obj_G(\tB) \bigm| \tB \in V(G)^{k}, G \models \phi(\tB) \bigr\}
    \]
    or concludes that no such tuple exists.
    If $\phi$ has size measurements $m_1,\dots,m_\ell$ then the run time of this algorithm is
    \[ 2^{O(|\phi|)} \cdot \snec_{\phi}(G,\cL)^{6|\phi|} \cdot  n^3 \cdot
	\prod_{1\leq i \leq \ell} (m_i+2)^4. \]
\end{theorem}
\begin{proof}
    We use \Cref{obs:core} to construct a \encore formula $\phi'$ that is equivalent to $\phi$.
    Thus, we can optimize this formula instead.
    Since $|\var(\phi')| \le |\phi|$, 
    we can write $\phi'$ as $\exists X_{k+1} \dots \exists X_j \hat\phi$ for some $j \le |\phi|$, and quantifier-free $\hat\phi$.
    We construct a $j$-weighted graph $\hat G$ from $G$ such that variables $X_{k+1},\dots,X_j$ do not contribute to the weight
    (meaning $w(v,S) = w(v,S \cap [k])$ for each $v$ and $S$)
    and therefore
    \[
    \max \bigl\{ \obj_G(\tB) \bigm| \tB \in V(G)^{k}, G \models \phi(\tB) \bigr\}
    = \max \bigl\{ \obj_{\hat G}(\tB) \bigm| \tB \in V(\hat G)^{j}, \hat G \models \hat\phi(\tB) \bigr\}.
    \]
    We now apply \Cref{lem:metathmCore} to $\hat G$, $\hat \phi(X_1,\dots,X_j)$ and $\cL$.
    The run time follows from \Cref{obs:core},
    since $j \le |\phi|$, $|\hat\phi| = O(|\phi|)$, $R(\hat\phi)=R(\phi)$ and $d(\hat\phi)=d(\phi)$.
\end{proof}

\begin{proof}[Proof of \Cref{thm:main:basic}]
    Assume we want to decide whether $G \models \phi$.
    Let $r = r(\phi)$, $d = d(\phi)$ and $s$ be the number of size measurements in $\phi$.
    Since we are not dealing with an optimization problem, we can assume the input graph to be a weighted graph
    with $w(v,S)=0$ for all $v$ and $S$.
    Using \Cref{thm:metathm}, we can decide whether $G \models \phi$ in time 
    \[ 2^{O(|\phi|)} \cdot \snec_{\phi}(G,\cL)^{6|\phi|} \cdot  n^{O(s)}.
    \]
    Assume we are given a rooted layout $\cL$ with $\mim(G,\cL) = w$.
    Then by \Cref{lem:necd1}, 
    \[
        \snec_{\phi}(G,\cL) \le \prod_{r' \in R(\phi)} \snec_{d(\phi),r'}(G,\cL) \le n^{2d \cdot |\phi| \cdot \mim(G,\cL)} = n^{2d \cdot |\phi| \cdot w}.
    \]
    If we are given a tree-width decomposition of width $w$, we convert it, as stated in the preliminaries, in time $O(n^2)$ into a rooted layout $\cL$ with $\mmw(G,\cL) \le w+1$.
    Similarly, we turn a given clique-width decomposition of width $w$ into a rooted layout $\cL$ with $\mw(G,\cL) \le w$.
    By \Cref{lem:necd1}, $\snec_{1,1}(G,\cL) \le 2^{O(w)}$ and
    \[
        \snec_{\phi}(G,\cL) \le \prod_{r' \in R(\phi)} \snec_{d(\phi),r'}(G,\cL) \le \snec_{1,1}(G,\cL)^{dr^2 \cdot \log(\snec_{1,1}(G,\cL)) \cdot |\phi|} \le 2^{O(dr^2w^2|\phi|)}.
    \]
    For the special case that $r=1$, \Cref{lem:necd1} yields
    \[
        \snec_{\phi}(G,\cL) = \snec_{d,1}(G,\cL) \le (2d+2)^{w} \le 2^{O(d w)}.
    \]
    If we are given a rank-width decomposition then by \Cref{lem:necd1}, $\snec_{1,1}(G,\cL) \le 2^{O(w^2)}$.
    Therefore
    \[
        \snec_{\phi}(G,\cL) \le 2^{O(dr^2w^4|\phi|)}
    \]
    and for the special case $r=1$
    \[
        \snec_{\phi}(G,\cL) \le 2^{O(d w^2)}.
    \]
    The result then follows by substituting $\snec_{\phi}(G,\cL)$ accordingly in the run time.
\end{proof}

\subsection{Adding connectivity and acyclicity}
\label{ssec:A&C}
In this section, we provide a model checking algorithm for \ac \en based on the one for \en.
For doing so, we incorporate and generalize the framework for connectivity and acyclicity developed in \cite{BergougnouxKante2021}.

As before, to facilitate our proof, we first consider a simple subset of \ac \en and then lift the result to all of \ac \en.
An \emph{\ac-clause} is a quantifier-free formula 
of the form $\phi \wedge \omega_1 \wedge \dots \wedge \omega_t$,
where $\phi$ is a \encore formula and each $\omega_i$ is of the form $\conn(X)$ or $\acy(X)$ for some set variable $X$.
For an \ac-clause $\phi$, we denote by $\conn(\phi)$ (resp. $\acy(\phi)$) the set containing all set variables $X\in \var(\phi)$ such that $\conn(X)$ (resp. $\acy(X)$) is a subformula of $\phi$.
Moreover, we denote by $\phi_\conn$ the formula $\bigwedge_{X\in \conn(\phi)} \conn(X)$ 

Given an \ac-clause $\phi(X_1,\dots,X_k)$, a $k$-weighted graph $G$, $A\subseteq V(G)$, $\cB\subseteq \Pset(A)^{k}$ and $\tD\in \Pset(\comp{A})^{k}$, 
we define $\best_{\Pi_\phi}^{\conn}(\cB,\tD)= \max\{ \obj_G(\tB) \mid \tB\in \cB$ and for every $X_i\in \conn(\phi), G[B_i \cup D_i]$ is connected$\}$.

To deal with the connectivity constraints, we use the following lemma based on the rank-based approach introduced by Bodlaender et al. in \cite{BodlaenderCKN15}.
We generalize the ideas used in \cite{BergougnouxKante2021} to adapt the rank-based approach to the $(1,1)$-neighbor equivalence and multiple connectivity constraints (in \cite{BergougnouxKante2021}, it was only proved how to deal with two connectivity constraints for a specific problem).

\begin{restatable}[$\star$]{lemma}{lemRankbasedApproach}
	\label{lem:rankbased:approach}
	Let $\phi$ be an \ac-clause, $G$ be a $\abs{\var(\phi)}$-weighted graph and $A\subseteq V(G)$.
	Given $\btE\in (\Rep{\comp{A}}{\phi})^{\abs{\var(\phi)}}$ and $\cB\subseteq \Pset(A)^{\abs{\var(\phi)}}$ such that the elements of $\cB$ are pairwise equivalent for $\equi{A}{\phi}$, we can compute in time 
	\[ 	O(\abs{\cB}\cdot \snec_{1,1}(A)^{O(\abs{\var(\phi)})} \cdot  2^{\abs{\phi}} \cdot n^2) \]
	a subset $\reduce^\conn_{\btE}(\cB)$
    of $\cB$ such that  $\abs{\reduce^\conn_{\btE}(\cB)}\leq \snec_{1,1}(A)^{2\abs{\var(\phi)}}\cdot 2^{\abs{\phi}}$ and for every $\tD\in \Pset(\comp{A})^{\abs{\var(\phi)}}$ such that $\tD\equi{\comp{A}}{\phi} \btE$, we have $\best_{\Pi_\phi}^{\conn}(\cB,\tD)=\best_{\Pi_\phi}^{\conn}(\reduce^\conn_{\btE}(\cB),\tD)$.
\end{restatable}

\newcommand{\forest}{\mathsf{{forest}}}

To handle the acyclicity constraints, we provide a reduce routine for the \ac-clauses $\phi$ such that $\acy(\phi)\subseteq \conn(\phi)$, that is for every $X\in \acy(\phi)$, we want that $X$ is interpreted as a subset of vertices that induces a tree.

This reduce routine for such \ac-clauses $\phi$ used \Cref{lem:rankbased:approach} and the following concepts.
We start by defining $\forest(A,\bE)$ which is a collection of subsets of $A$  satisfying some properties (some of these properties are similar to the ones in \cite[Lemma 6.5]{BergougnouxKante2021}). 
We prove in the next lemma, that for every $B\notin \forest(A,\bE)$ and $D\subseteq \comp{A}$ with $\equi{\comp{A}}{1,2} \bE$ that $G[B\cup D]$ is not a tree.

\begin{definition}\label{def:forest:acy}
	Given a graph $G$, $A\subseteq V(G)$, $ B\subseteq A$ and $\bE\subseteq \comp{A}$, we denote by $ B^{1}_{\bE}$ the set of vertices in $ B$ with exactly one neighbor in $\bE$ and by $ B^{2+}_{\bE}$ those with at least two neighbors in $\bE$.
	We define $H_{\bE, B}$ to be the bipartite graph between $\cc( B)$ and $\{N_1^1(v)\cap \comp{A}\mid v\in  B^{1}_{\bE}\}$ such that $C\in \cc( B)$
	is adjacent to $F\in \{N_1^1(v)\cap \comp{A}\mid v\in  B^{1}_{\bE}\}$ iff there exists a vertex $v\in C \cap B^{1}_{\bE}$ such that $N_1^1(v)\cap \comp{A}=F$.
	Finally, we define $\forest_G(A,\bE)$ as the set containing all sets $ B\subseteq A$ such that:
	\begin{enumerate}
		\item The graphs $G[ B]$ and $H_{\bE, B}$ are forests.
		\item The size of $ B^{2+}_{\bE}$ is at most $2\mim(A)$.
		\item \label{item:no:twin} For every pair of distinct vertices $(u,v)$ in $ B^{1}_{\bE}\cup B^{2+}_{\bE}$, if $u$ and $v$ are connected or $u,v\in  B^{2+}_{\bE}$, then we have $N_1^1(u)\cap \comp{A} \neq N_1^1(v)\cap \comp{A}$.
		\item Either $ B=\emptyset$ or every connected component $C$ of $G[ B]$ intersects $ B^1_{\bE}\cup B^{2+}_{\bE}$.
	\end{enumerate}
	We omit $G$ from the subscript when it is clear from the context.
\end{definition}

\begin{restatable}[$\star$]{lemma}{ACproperties}\label{lem:A&C:properties}
	Let $G$ be a graph, $A\subseteq V(G)$ and $\bE\subseteq \comp{A}$.
	For every $ B\subseteq A$, if there exists $ D\subseteq \comp{A}$ such that $ D\equi{\comp{A}}{2,1} \bE$ and $G[ B\cup D]$ is a tree, then $ B\in \forest(A,\bE)$.
\end{restatable}

The following equivalence relation over $\forest(A,\bE)$ is a key ingredient to handle the acyclicity of the \ac-clauses $\phi$ such that $\acy(\phi)\subseteq \conn(\phi)$.

\begin{restatable}[$\star$]{lemma}{acyeq}
	\label{lem:acy:eq}
	Let $G$ be a graph, $A\subseteq V(G)$ and $\bE\subseteq \comp{A}$.
	There exists an equivalence relation $\sim^\acy_{\bE}$ over $\forest(A,\bE)$ satisfying the following properties:
	\begin{enumerate}
		\item\label{lem:acy:eq:1} 
		For every $ B, C\in \forest(A,\bE)$ and $ D\subseteq \comp{A}$ such that  $ B\sim^\acy_\bE  C$ and $ D\equi{\comp{A}}{2,1} \bE$, if $G[ B\cup  D]$ is a tree and $G[ C\cup D]$ is connected, then $G[ C\cup D]$ is a tree.
		
		\item\label{lem:acy:eq:2} 
		We can decide whether $ B\sim^\acy_\bE  C$ in time $O(n^2)$.
		
		\item\label{lem:acy:eq:3} 
		The number of equivalence classes of $\sim^\acy_\bE$ is upper-bounded by $\sfN_\acy(G,A)$ that is the minimum between 
		$2^{\mw(A)}\cdot 4\mw(A)$, $2^{3\rw(A)^2+1}$, and $2n^{2\mim(A)+1}$.
	\end{enumerate}
\end{restatable}

We provide our reduce routine for the model checking of \ac-clauses $\phi$ with $\acy(\phi)\subseteq \conn(\phi)$ by mixing the equivalence relations $\bowtie_\btE$ given by \Cref{lem:equivalenceformula}, $\sim^\acy_\bE$ from the previous lemma and the function $\reduce^\conn_{\btE}(\cB)$ of \Cref{lem:rankbased:approach}.

\begin{lemma}\label{lem:reduce:A&C}
	Let $\phi$ be a  \ac-clause such that $\acy(\phi)\subseteq \conn(\phi)$ with $p$ size measurements $\abs{t_1}\leq m_1,\dots,\abs{t_p}\leq m_p$.
	The $\abs{\var(\phi)}$-problem $\Pi_{\phi}$ admits a reduce routine with size bound
	\[
	\s(G,A) = \snec^G_{\phi}(A)^{3\abs{\var(\phi)}}\cdot \sfN_\acy(G,A)^{\abs{\var(\phi)}} \cdot 2^{2\abs{\phi}} \cdot  \prod_{1\leq i\leq p} (m_i+2)
	\]
	and run time bound
	\[
	\f(G,A) =
	O(2^{\abs{\phi}}  \cdot \snec_{\phi}(A)^{O(\abs{\var(\phi)})} \cdot \sfN_\acy(G,A)^{\abs{\var(\phi)}} \cdot n^2).
	\]  
\end{lemma}
\begin{proof}
	We assume w.l.o.g. that $\conn(\phi)$ the first $k$-variables $X_1,\dots,X_k$ of $\phi$ and $\acy(\phi)$ contains the first $\ell \leq k$ variables $X_1,\dots,X_\ell$.
	Let $\psi$ be the quantifier-free \encore formula such that $\phi\equiv \psi \wedge \bigwedge_{i\in[k]} \conn(X_i) \wedge \bigwedge_{j\in [\ell]} \acy(X_j)$.
	We assume w.l.o.g. that $\var(\psi)=\var(\phi)$ and $d(\phi)=d(\psi)$.
	This assumption implies that for every graph $G$ and $A\subseteq V(G)$, the equivalence relations $\equi{A}{\phi}$ and $\equi{A}{\psi}$ are the same.
	
	Let $G$ be a $\abs{\var(\phi)}$-weighted graph, $A\subseteq V(G)$ and $\cB\subseteq \Pset(A)^{\abs{\var(\phi)}}$.
	Let us explain how we compute a subset $\reduce(\cB)$ of $\cB$ of size at most $\s(G,A)$ such that $\reduce(\cB)$ $\Pi_\phi$-represents $\cB$ over $A$.
	We compute $\reduce(\cB)$ from the following equivalence relations and \Cref{lem:rankbased:approach}.
	
	For every $\btE\in \Rep{\comp{A}}{\phi}$, by \Cref{lem:equivalenceformula},
	there exists an equivalence relation $\bowtie_\btE$ over $\Pset(A)^{\abs{\var(\phi)}}$ such that:
	\begin{itemize}
		\item For every $\tB,\tC\in \cB$ and $\tD\in \Pset(\comp A)^{\abs{\var(\phi)}}$ such that $\tB\bowtie_{\btE} \tC$ and $\tD \equi{\comp{A}}{\phi} \btE$, we have $G\models \psi(\tB\cup \tD)$ iff $G\models \psi(\tC\cup\tD)$.
	\end{itemize}
	For every $\btE\in \Rep{\comp{A}}{\phi}$, we define $\forest(\cB,\btE)$ as the set of all $\tB\in \cB$ such that for every $i\in [\ell]$, we have $B_i\in \forest(A,\bE_i)$.
	We define the equivalence relation $\bowtie^\acy_\btE$ over $\forest(\cB,\btE)$ such that $\tB\bowtie^\acy_\btE \tC$ if $\tB\bowtie_{\btE} \tC$ and for every $i\in [\ell]$, we have $B_i\sim^\acy_{\bE_i} C_i$ where  $\sim^\acy_{\bE_i}$ is the equivalence relation given by \Cref{lem:acy:eq}.
	
	\noindent
	For every $\btE\in \Rep{\comp{A}}{\phi}$, we compute a subset $\cB_\btE \subseteq \cB$ as follows:
	\begin{itemize}
		\item We compute $\forest(\cB,\btE)$ and the equivalence classes $\cC_1,\dots,\cC_q$ of $\sim^\acy_{\btE}$ over $\forest(\cB,\btE)$.
		\item We compute $\cB_\btE=\bigcup_{i\in [q]} \reduce^\conn_{\btE}(\cC_i)$ where $\reduce^\conn_{\btE}$ is the algorithm given by \Cref{lem:rankbased:approach}.
	\end{itemize}
	Finally, we compute $\reduce(\cB)=\bigcup_{\btE\in (\Rep{\comp A}{\phi})^{\abs{\var{\phi}}}} \cB_\btE$.
	
	\myparagraph{Representativity} 
	We need to show that for every $\tD\in \Pset(\comp A)^{\abs{\var(\phi)}}$, we have $\best_{\Pi_{\phi}}(\cB,\tD)=\best_{\Pi_{\phi}}(\reduce(\cB),\tD)$.
	Let $\tD \in \Pset(\comp A)^{\abs{\var(\phi)}}$ and $\btE\in (\Rep{\comp A}{\phi})^{\abs{\var{\phi}}}$ with $\tD\equi{\comp A}{\phi} \btE$.
	If $\best_{\Pi_{\phi}}(\cB,\tD)=-\infty$, then $\best_{\Pi_{\phi}}(\reduce(\cB),\tD)=-\infty$ as $\reduce(\cB)$ is by construction a subset of $\cB$.
	
	Suppose that $\best_{\Pi_{\phi}}(\cB,\tD)\neq\infty$ and let $\tB\in \cB$ such that $G\models \phi(\tB\cup \tD)$ and $\obj(\tB)=\best_{\Pi_{\phi}}(\cB,\tD)$.
	Let $\cC$ be the equivalence class of $\bowtie^\acy_\btE$ such that $\tB\in \cC$.
	By \Cref{lem:rankbased:approach} we have $\best_{\Pi_\phi}^{\conn}(\cC,\tD)=\best_{\Pi_\phi}^{\conn}(\reduce^\conn_{\btE}(\cC),\tD)$.
	Since  $G\models \phi(\tB\cup \tD)$, we have also $G\models \phi_\conn(\tB\cup \tD)$ and thus $\obj(\tB)\geq \best_{\Pi_\phi}^{\conn}(\cC,\tD)$. 
	Thus, there exists $\tB^\star\in \reduce^\conn_{\btE}(\cC)$ such that $\obj(\tB)\leq \obj(\tB^\star)$ and $G\models \phi_\conn(\tB^\star\cup \tD)$.
	
	As $\cC$ is an equivalence class of $\bowtie^\acy_\btE$, we have $\tB \bowtie_{\btE} \tB^\star$ and for every $i\in[\ell]$, we have $B_i \sim^\acy_{\bE_i} B_i^\star$.
	We deduce that $G\models \phi(\tB^\star\cup \tD)$ from the following observations:
	\begin{itemize}
		\item We have $G\models \psi(\tB^\star\cup \tD)$ because $\tB \bowtie_{\btE} \tB^\star$ and $G\models \psi(\tB\cup \tD)$.
		\item For every $i\in [\ell]$, the graph $G[B_i^\star\cup D_i]$ is a tree.
		Indeed, since $G\models \phi_\conn(\tB^\star\cup \tD)$ and $\acy(\phi)\subseteq \conn(\phi)$, we know that $G[B_i\cup D_i]$ is connected. Thus, by \Cref{lem:acy:eq} since $B_i \sim^\acy_{\bE_i} B_i^\star$ and $G[B_i\cup D_i]$ is a tree, we deduce that $G[B_i^\star \cup D_i]$ is a tree.
	\end{itemize}
	By construction, we have $\tB^\star\in \reduce(\cB)$ and thus $\best_{\Pi_{\phi}}(\reduce(\cB),\tD)\geq \obj(\tB^\star)$. Since $\reduce(\cB) \subseteq \cB$ and $\obj(\tB)\leq \obj(\tB^\star)$, we deduce that $\best_{\Pi_{\phi}}(\cB,\tD)=\best_{\Pi_{\phi}}(\reduce(\cB),\tD)$.
	
	As this holds for every $\tD\in \Pset(\comp A)^{\abs{\var(\phi)}}$, $\reduce(\cB)$ $\Pi_{\phi}$-represents $\cB$.
	
	\myparagraph{Size bound}
	By \Cref{lem:acy:eq}, for every $\btE\in (\Rep{\comp A}{\phi})^{\abs{\var{\phi}}}$ and $i\in [\ell]$, the number of equivalence classes of $\sim^\acy_{\bE_i}$ over $\forest(A,\bE_i)$ is upper bounded by $\sfN_{\acy}(G,A)$.
	Moreover, by \Cref{lem:equivalenceformula}, $\bowtie_\btE$ has at most $\nec_{\phi}(A)^{k}\cdot 2^{\abs{\phi}} \cdot \prod_{1\leq i \leq p} (m_i+2)$ equivalence classes.
	We deduce that the number of equivalence classes of  $\bowtie^\acy_\btE$ over $\forest(\cB,\btE)$ is upper bounded by:
	\begin{equation}\label{eq:bowtie:acy}
		2^{\abs{\phi}} \cdot \nec_{\phi}(A)^{\abs{\var(\phi)}}\cdot \sfN_\acy(G,A)^{\ell} \cdot \prod_{1\leq i \leq p} (m_i+2).
	\end{equation}
	By \Cref{lem:rankbased:approach}, for every equivalence class $\cC$ of $\bowtie^\acy_\btE$ over $\forest(\cB,\btE)$ the size of $\reduce^\conn_{\btE}(\cC)$ is at most $2^{k} \cdot \snec_{1, 1}(A)^{2k}$.
	As $d(\phi)\geq 1$ and $1\in R(\phi)$, we have $\snec_{1, 1}(A)\leq \snec_{\phi}(A)$.
	Thus, for every $\btE\in (\Rep{\comp A}{\phi})^{\abs{\var{\phi}}}$, the size of $\cB_\btE$ is at most 
	\[  2^{\abs{\phi}+k} \cdot  \nec_{\phi}(A)^{\abs{\var(\phi)}+k}\cdot \sfN_\acy(A)^{\ell} \cdot 
	\prod_{1\leq i \leq p} (m_i+2). \]
	Since $k,\ell \leq \abs{\phi}$, $\abs{\Rep{\comp{A}}{\phi}}=\nec_{\phi}(\comp{A})$ and $\nec_{\phi}(A)\cdot \nec_{\phi}(\comp{A})\leq \snec_{\phi}(A)^2$, we can conclude that $\abs{\reduce(\cB)}\leq \s(G,A)$. 
	
	\myparagraph{Run time bound}
	According to \Cref{def:reduce}, the reduce routine we design gets the distance matrix of $G$ as part of the input.
	We claim that computing $\reduce(\cB)$ can be done in time $O(\abs{\cB}\cdot \f(G,A))$.
	We do the following computations:
	\begin{itemize}
		\item Using \Cref{lem:computerep}, we compute
		in time $O(\snec_{\phi}(A)\cdot \abs{R(\phi)} \cdot n^2 \cdot \log(\snec_{\phi}(A)))$ 
		the sets $\Rep{A}{\phi}$, $\Rep{\comp{A}}{\phi}$ and the data structures that, given a set $B\subseteq A$, computes a pointer to $\rep{A}{\phi}(B)$ in time $O(n^2\cdot\card{R(\phi)} \log(\nec_{\phi}(A)))$.
		As we are given the distance matrix of $G$ and $\abs{R(\phi)}\leq \abs{\phi}$, 
		computing these sets and data structures can be done in $O(\abs{\phi}\cdot\snec_{\phi}(A)\cdot n^2 \cdot \log(\snec_{\phi}(A)))$ time.
		
		\item To compute the sets $\forest(\cB,\btE)$, we need to compute $\mim(A)$ as we will use it in the following claim. 
		It is well known that $\mim(A)$ is the maximum size of a set in $\Rep{A}{1,1}$.
		By \Cref{lem:computerep}, we deduce that $\mim(A)$ can be computed in time $O(\nec_{1, 1}(A) \cdot n^2 \cdot \log(\nec_{1, 1}(A)))$. 
		
		\begin{claim}\label{claim:acy:forest}
			For every $B\subseteq A$ and  $\bE\subseteq \comp{A}$, we can decide whether $B\in \forest(A,\bE)$ in time $O(n^2)$.
		\end{claim}
		\begin{claimproof}
			Let $B\subseteq A$.
			We compute $B_{\bE}^{1}$ and $B^{2+}_{\bE}$ in time $O(n^2)$ by computing $\abs{N(v)\cap \bE}$ for every $v\in B$.
			We compute the set of connected components $\cc(B)$ in time $O(n^2)$.
			For every $v\in B^1_{\bE}\cup B^{2+}_{\bE}$, we consider the set $N(v)\cap \comp{A}$ as the restriction to $\comp{A}$ of the row of adjacency matrix associated with $v$. 
			Since $\abs{V(H_B)} = \abs{\cc(B)} + \abs{\{N(v)\cap \comp{A} \mid v\in B^{1}_{\bE}\}} \leq 2n$, it is easy to see that $H_B$ can be computed in time $O(n^2)$.
			
			It is trivial from these computations and given the value of $\mim(A)$, that we can check whether $B$ satisfies the properties of \Cref{def:forest:acy} in time $O(n^2)$.			
		\end{claimproof}
		For every $\btE\in (\Rep{\comp A}{\phi})^{\abs{\var{\phi}}}$, we compute $\forest(\cB,\btE)$ by checking for every $\tB\in \cB$ and $i\in [\ell]$ whether $B_i\in\forest(A,\bE_i)$.
		By \Cref{claim:acy:forest}, this can be done in time $O(\abs{\cB}\cdot \ell \cdot n^{2})$.
		
		\item For every $\btE\in (\Rep{\comp A}{\phi})^{\abs{\var(\phi)}}$, we compute the equivalence classes of $\bowtie_\btE^\acy$ over $\forest(\cB,\btE)$.
		
		We start by computing the equivalence classes of $\bowtie_\btE$ over $\forest(\cB,\btE)$ with the technique used in \Cref{lem:reducelogic} to compute, for each of equivalence class $\cC$ of $\bowtie_\btE$ over $\cB$, a tuple $\tB\in \cC$ such that $\obj(\tB)$ is maximum.
		Thanks to \Cref{lem:reducelogic}, we know that computing the equivalence classes of $\bowtie_\btE$ over $\forest(\cB,\btE)$  for all $\btE$, takes at most $\abs{\cB}\cdot \f(G,A)$ time.
		
		For each $\btE\in (\Rep{\comp A}{\phi})^{\abs{\var(\phi)}}$, we compute the equivalence classes of $\bowtie_{\btE}^\acy$ over $\forest(\cB, \btE)$ as follows.
		By definition, $\bowtie_{\btE}$ is a coarsening of $\bowtie_{\btE}^\acy$.
		We therefore take each equivalence class $\cC$ of $\bowtie_{\btE}$ over $\forest(\cB, \btE)$, 
		and further break it down according to the second condition of $\bowtie_{\btE}^\acy$.
		For every $\tB,\tC\in \cC$, we have $\tB\bowtie_\btE \tC$, and thus $\tB\bowtie_\btE^\acy \tC$ if and only if 
		for every $i\in [\ell]$, we have $B_i \sim^\acy_{\bE_i} C_i$.
		The latter can be decided in $O(n^2\cdot \ell)$ time for any such pair $\tB, \tC$ by \Cref{lem:acy:eq}.
		For every $\tB,\tC\in \cC$, we have $\tB\bowtie_\btE \tC$, and thus $\tB\bowtie_\btE^\acy \tC$ iff for every $i\in [\ell]$, we have $B_i \sim^\acy_{\bE_i} C_i$.
		Since by \Cref{lem:acy:eq}, each $\sim^\acy_{\bE_i}$ has $\sfN_\acy(G, A)$ equivalence classes and $\ell\leq \abs{\var(\phi)}$,
		we have that the equivalence classes of $\bowtie_{\btE}^\acy$ over $\cC$ can be computed in  
		$O(\abs{\cC}\cdot n^2 \cdot \sfN_\acy(G,A)^{\abs{\var(\phi)}})$ time.
		
		Thus, we can compute the equivalence classes of $\bowtie_\btE^\acy$ over $\forest(\cB,\btE)$ for all 
		$\btE\in (\Rep{\comp A}{\phi})^{\abs{\var(\phi)}}$  in time at most $\abs{\cB}\cdot \f(G,A)$.
		
		\item Let $\btE\in (\Rep{\comp A}{\phi})^{\abs{\var{\phi}}}$ and $\cC_1,\dots,\cC_q$ be the  the equivalence classes of $\bowtie_\btE^\acy$ over $\forest(\cB,\btE)$.
		By \Cref{lem:rankbased:approach}, for each $i\in [q]$, we can compute $\reduce^{\conn}_{\btE}(\cC_i)$ in time
		\[ 	O(\abs{\cC_i} \cdot  2^{\abs{\phi}} \cdot \snec_{\phi}(A)^{O(\abs{\var(\phi)})} \cdot  n^2). \]
		Since $\bigcup_{i\in[q]}\abs{\cC_i}=\forest(\cB,\btE)$ and $\abs{\forest(\cB,\btE)}\leq \abs{\cB}$, we deduce that computing  $\cB_\btE= \bigcup_{i\in[q]} \reduce^{\conn}_{\btE}(\cC_i)$ can be done in time
		\[ 	O(\abs{\cB}\cdot 2^{\abs{\phi}} \cdot \snec_{\phi}(A)^{O(\abs{\var(\phi)})} \cdot   n^2). \]
		
		\item 
		Since $\abs{(\Rep{\comp A}{\phi})^{\abs{\var(\phi)}}}\leq \snec_{\phi}(A)^{\abs{\var(\phi)}}$, we conclude that computing all the sets $\cB_\btE$ and $\reduce(\cB)$ takes at most $\abs{\cB}\cdot \f(G,A)$ time.
	\end{itemize}	
	The time necessary for each of these computations is bounded by  $\abs{\cB}\cdot \f(G,A)$.
\end{proof}

Through a reduction based on \cite{BergougnouxKante2021} and some modification of the reduce routine of \Cref{lem:reduce:A&C}, we are able to lift the statement to any \ac-clauses.

\begin{restatable}[$\star$]{lemma}{lemACclause}
	\label{lem:ACmetathmCore}
	There is an algorithm that, for a given quantifier-free \ac-clause~$\phi(X_1,\dots,X_k)$, 
	$k$-weighted graph $G$, and rooted layout $\cL$ of $G$, computes a tuple $\tB^* \in \Pset(V(G))^{k}$ such that
	\[
	\obj_G(\tB^*) = \max \bigl\{ \obj_G(\tB) \bigm| \tB \in V(G)^{k}, G \models \phi(\tB) \bigr\}
	\]
	or concludes that no such tuple exists.
	If $\phi$ has size measurements $m_1,\dots,m_p$ then the run time of this algorithm is
	\[ 2^{O(\abs{\phi})}  \cdot \snec_{\phi}(\cL)^{O(k)}\cdot \sfN_\acy(\cL)^{O(k)} \cdot n^3\cdot 
	\prod_{1\leq i \leq p} (m_i+2)^2\]
	Where $\sfN_\acy(\cL)$ is the maximum among $\sfN_\acy(G,V_x)$ over the node $x$ of $\cL$.
\end{restatable}

We lift the statement to \ac \en logic via a turing reduction.

\begin{theorem}\label{thm:ACmetathm}
    There is an algorithm that, for a given \ac \en formula~$\phi(X_1,\dots,X_k)$, 
    $k$-weighted graph $G$, and rooted layout $\cL$ of $G$, computes a tuple $\tB^* \in \Pset(V(G))^{k}$ such that
    \[
    \obj_G(\tB^*) = \max \bigl\{ \obj_G(\tB) \bigm| \tB \in V(G)^{k}, G \models \phi(\tB) \bigr\}
    \]
    or concludes that no such tuple exists.
    If $\phi$ has size measurements $m_1,\dots,m_\ell$ then the run time of this algorithm is
    \[ 
    2^{O(\abs{\phi})}  \cdot \snec_{\phi}(\cL)^{O(\abs{\phi})}\cdot \sfN_\acy(\cL)^{O(\abs{\phi})} \cdot n^3\cdot 
    \prod_{1\leq i \leq p} (m_i+2)^4.\]
\end{theorem}
\begin{proof}
    We use \Cref{obs:core} to construct a \ac \encore formula that is equivalent to $\phi$.
    We can write it as $\exists X_{k+1} \dots \exists X_j \hat\phi$, where $\hat\phi$ is quantifier-free.
    We construct a $j$-weighted graph $\hat G$ from $G$ such that variables $X_{k+1},\dots,X_j$ do not contribute to the weight
    (that is $w(v,S) = w(v,S \cap [k])$ for each $v$ and $S$)
    and therefore
    \[
    \max \bigl\{ \obj_G(\tB) \bigm| \tB \in V(G)^{k}, G \models \phi(\tB) \bigr\}
    = \max \bigl\{ \obj_{\hat G}(\tB) \bigm| \tB \in V(\hat G)^{j}, \hat G \models \hat\phi(\tB) \bigr\}.
    \]
    In time $2^{O(|\phi|)}$, we convert $\hat\phi$ into conjunctive normal form. 
    That is, we obtain a formula $\psi_1 \lor \dots \lor \psi_q$ that is equivalent to $\hat\phi$,
    where each $\psi_i$ has length $|\psi_i| \le |\phi|$ and is a conjunction of \ac \encore literals.
    This means all connectivity or acyclicity literals in $\psi_i$ are of the form $\conn(X)$, $\acy(X)$, $\neg\conn(X)$ or $\neg\acy(X)$.
    Note that the latter two can be expressed in \encore.
    \begin{itemize}
        \item $\neg\conn(X)$ is equivalent to the existence of a partition $Z_1,Z_2$ of $X$ such that no edges go between these two parts.
        That is, $\neg\conn(X) \equiv \exists Z_1 \exists Z_2 \, Z_1 \neq \emptyset \land Z_2 \neq \emptyset \land Z_1 \cup Z_2 = X \land N(Z_1) \cap Z_2 = \emptyset$.
        \item Also, $\neg\acy(X)$ is equivalent to the existence of a subset $Z$ inducing graph with minimal degree two.
        That is, $\neg\acy(X) \equiv \exists Z \, Z \neq \emptyset \land Z \subseteq X \land  Z \subseteq N_2^1(Z)$.
    \end{itemize}
    Using these substitutions, we obtain \ac-clauses $\psi'_i$ equivalent to $\psi_i$
    satisfying the following conditions.
    \begin{itemize}
        \item $|\psi'_i| = O(|\phi|)$
        \item $R(\psi'_i) \subseteq R(\phi)$,
        \item $d(\psi_i') \le d(\phi)$,
        \item $\prod_{i \in [\ell']} (m'_i+2) \le 2^{O(|\phi|)} \prod_{i\in[\ell]} (m_i+2)^2$, 
            where $m'_1,\dots,m'_{\ell'}$ are the values of size measurements in $\psi'_i$ and
            $m_1,\dots,m_\ell$ are the values of size measurements in $\phi$.
    \end{itemize}
    Since $\hat\phi = \psi_1' \lor \dots \lor \psi_q'$, we have
    \begin{multline*}
    \max \bigl\{ \obj_G(\tB) \bigm| \tB \in V(G)^{k}, G \models \phi(\tB) \bigr\} = \\
    \max_{1 \le i \le q} \max \bigl\{ \obj_{\hat G}(\tB) \bigm| \tB \in V(\hat G)^{j}, \hat G \models \psi_i'(\tB) \bigr\}.
    \end{multline*}
    We compute all values $\max \bigl\{ \obj_{\hat G}(\tB) \bigm| \tB \in V(\hat G)^{j}, \hat G \models \psi_i'(\tB) \bigr\}$
    using \Cref{lem:ACmetathmCore} and return the maximum.
\end{proof}

At last, \Cref{thm:main:AC} follows from \Cref{thm:ACmetathm}
in the same way as \Cref{thm:main:basic} follows from \Cref{thm:metathm},
so we omit the proof.
\section{Expressing properties in (\ac) \en logic}\label{sec:exp}
In this section we show how to express properties in \ac\en logic. 
In particular, we show how to express almost all problems that are known to be solvable in \XP time 
parameterized by the mim-width of a given rooted layout of the input graph,
some problems that were not known to be solvable in such a run time,
and diverse variants of problems expressible in \ac\en logic.

\subsection{Locally checkable vertex subset and partitioning problems}\label{sec:exp:lc}
We give the definitions of the \emph{locally checkable vertex subset and partitioning} problems 
as introduced by Telle and Proskurowski~\cite{TelleProskurowski1997},
and show that all of these problems are expressible in \en logic.
In fact, we work with their distance-$r$ variants~\cite{JaffkeKwonStrommeTelle2019};
which consider the $r$-th power of the input graph.
Setting $r$ to $1$, we recover the original definition of the locally checkable problems.

We start with subset problems.
For a non-negative integer $r$ and sets of non-negative natural numbers $\sigma, \rho \subseteq \bN$, 
a \emph{distance-$r$ $(\sigma, \rho)$-set} of a graph $G$ is a vertex set $S \subseteq V(G)$ such that
\begin{itemize}
	\item for all $v \in S$: $\card{N^r(v) \cap S} \in \sigma$, and
	\item for all $v \in V(G) \setminus S$: $\card{N^r(v) \cap S} \in \rho$.
\end{itemize}
To describe the complexity of a corresponding $(\sigma, \rho)$-problem,
we require the following notion of the \emph{$d$-value} of a set of numbers.
In particular, the $d$-value will determine the largest $d'$ such that $N^{\cdot}_{d'}(\cdot)$ 
appears in a \en logic formula describing a $(\sigma, \rho)$-set.
\begin{definition}[$d$-value]
	Let $d(\bN) = 0$.
	For every non-empty finite or co-finite set $\mu \subseteq \bN$,
	let $d(\mu) = 1 + \min\{\max\{x \mid x \in \mu\}, \max\{x \mid x \in \bN \setminus \mu\}\}$.
\end{definition}

\begin{lemma}\label{lem:lcv}
	Let $r$ be a positive integer and let $\mu \subseteq \bN$ be a nonempty finite or co-finite set.
    There is a \en term $t_{r, \mu}(X)$
    such that for all graphs $G$ and all $U \subseteq V(G)$,
    $\ip{t_{r, \mu}(U)}^{G} = \{v \in V(G) \mid \card{N^r(v) \cap U} \in \mu\}$.
\end{lemma}
\begin{proof}
	Let $X$ be a variable, let $a = \min \mu$ and $t_{r, \mu}^a = N^r_a(X)$.
	(If $a = 0$, then $t_{r, \mu}^a = \overline{\emptyset}$.)
	Then, for $i \in \{a + 1, \ldots, d(\mu)\}$, let:
	\begin{align*}
		t_{r, \mu}^i = t_{r, \mu}^{i-1} \left\lbrace
			\begin{array}{ll}
				\cup\, N^r_i(X), &\mbox{ if } i \in \mu \\
				\setminus\, N^r_i(X), &\mbox{ if } i \notin \mu \\
			\end{array}
			\right.
	\end{align*}
	Finally, let $t_{r, \mu} = t_{r, \mu}^{d(\mu)}$.
	It is then straightforward to verify that the term $t_{r, \mu}$ has the claimed properties.
\end{proof}

\begin{proposition}\label{prop:sigmarho}
	Let $r$ be a positive integer, and
	let $\sigma, \rho \subseteq \bN$ be finite or co-finite.
	There is a \en term $\phi_{r, \sigma, \rho}(X)$ 
	which evaluates to true if and only if a vertex set $X$ of a graph $G$ is a distance-$r$ $(\sigma, \rho)$-set in $G$.
\end{proposition}
\begin{proof}
	Using \cref{lem:lcv}, we construct the term as
	\(\phi_{r, \sigma, \rho}(X) = X \subseteq t_{r, \sigma}(X) \land \overline{X} \subseteq t_{r, \rho}(X).\)
\end{proof}

We now turn to the locally checkable \emph{partitioning} problems.
Let $q$ be a positive integer; a \emph{neighborhood constraint matrix} is a 
$q \times q$ matrix $D$ such that for all $i, j \in \{1, \ldots, q\}$,
$D[i, j]$ is a non-empty subset of $\bN$.
Fix $r$ to be a positive integer.
For a graph $G$, a vertex partition $(X_1, \ldots, X_q)$ is said to be a \emph{distance-$r$ $D$-partition} if 
for all $i, j \in \{1, \ldots, q\}$, we have that for all $v \in X_i$, $\card{N^r(v) \cap X_j} \in D[i, j]$.
If in a neighborhood constraint matrix $D$ all entries are either finite or co-finite sets,
we call $D$ a \emph{neighborhood constraint matrix over finite and co-finite sets}.

\begin{proposition}\label{prop:lcvp}
	Let $q$ and $r$ be positive integers, 
	and let $D$ be a $q \times q$ neighborhood constraint matrix
	over finite and co-finite sets.
	There is a \en formula $\phi_{r, D}(X_1, \ldots, X_q)$ which evaluates to true 
	if and only if a tuple of vertex sets $(X_1, \ldots, X_q)$ 
	in a graph $G$ is a distance-$r$ $D$-partition in $G$.
\end{proposition}
\begin{proof}
	To verify that $(X_1, \ldots, X_q)$ is a partition of the vertex set, we construct the formula
	\begin{align}
		\label{eq:exp:part}
		\partition(X_1, \ldots, X_q) \equiv \overline{X_1 \cup \cdots \cup X_q} = \emptyset \land \bigwedge\nolimits_{1 \le i < j \le q} X_i \cap X_j = \emptyset.
	\end{align}
	Using \cref{lem:lcv} we then construct the desired formula
    \[
		\phi_{r, D}(X_1, \ldots, X_q) \equiv \partition(X_1, \ldots, X_q) \land \bigwedge\nolimits_{i, j \in \{1, \ldots, q\}} X_j \subseteq t_{r, D[i, j]}(X_i).
    \]
\end{proof}

For fixed $r$ and finite or co-finite nonempty $\sigma, \rho \subseteq \bN$, 
we express the problem asking for a size-$k$ distance-$r$ $(\sigma, \rho)$-set using \cref{prop:sigmarho} 
as $\exists X\, \card{X} = k \land \phi_{r, \sigma, \rho}(X)$.
For fixed $r$ and a $q \times q$ neighborhood constraint matrix over finite and co-finite sets $D$, 
we express the problem asking for a distance-$r$ $D$-partition of a graph using \cref{prop:lcvp}
as $\exists X_1 \cdots \exists X_q\, \phi_{r, D}(X_1, \ldots, X_q)$.

\subsection{\labprob}\label{sec:exp:l21}
One strength of \en logic is that we can check neighborhoods of sets in various powers of the graph.
This way we were able to express the property that a graph has a semitotal dominating set of size at most $k$, 
see~equation~(\ref{eq:semitotal:dnlog}) on page~\pageref{eq:semitotal:dnlog}.
Here, we show that we can express the \labprob problem in \en logic.
First, we recall the problem definition.
For a graph $G$, a $k$-$L(d_1, \ldots, d_s)$-labeling is a coloring $\coloring$ 
of the vertices of $G$ with colors $[k]$ such that 
for all $i \in [s]$ and all $u, v \in V(G)$, if $\dist(u, v) \le i$,
then $\abs{\coloring(u) - \coloring(v)} \ge d_i$.
Note that we use here that the distance between $u$ and $v$ is \emph{at most $i$}.
Other definitions of the \labprob problem impose a constraint on $u$ and $v$ being at distance \emph{exactly $i$}.
We would like to remark that for this more general setting, we can design a \en formula expressing the problem as well.
\begin{proposition}
	For fixed $k$ and $d_1, \ldots, d_s$, 
	\labprob is expressible in \en logic.
\end{proposition}
\begin{proof}
	For all $h \in \{1,\ldots,s\}$, we create a $k \times k$ neighborhood constraint matrix $D_h$ as follows:
	\[
		\forall i, j \in \{1, \ldots, k\}\colon 
			D_h[i, j] = \left\lbrace
				\begin{array}{ll}
					\bN, &\mbox{ if } \abs{i - j} \ge d_h \\
					\{0\}, &\mbox{ otherwise}
				\end{array}
				\right.
	\]
	We can then express \labprob using \cref{prop:lcvp} by the formula
	\[
		\exists X_1 \cdots \exists X_k\, \bigwedge\nolimits_{h \in \{1, \ldots, s\}} \phi_{h, D_h}(X_1, \ldots, X_k).
	\]
\end{proof}

\subsection{Not locally checkable problems}
Here we show that several problems that are not LCVS or LCVP, some of them
studied in~\cite{BergougnouxKante2021,JaffkeKwonTelle2020b,JaffkeKwonTelle2020},
on graphs of bounded mim-width, 
and some of them not considered on graphs of bounded mim-width before,
are expressible in \ac \en logic.
We only show a few examples such problems and would like to point out that 
for any (locally checkable (distance-$r$)) partitioning problem, we can add connectivity or acyclicity constraints, 
similar to the framework of~\cite{BergougnouxKante2021}.
The \textsc{Acyclic $k$-Coloring} problem studied below is a nice example of how 
an independence and acyclicity condition can be combined in this framework.

The \textsc{Longest Induced Path} problem asks, 
given a graph $G$ and an integer $k$, whether $G$ contains an induced path on $k$ vertices.
The \textsc{Induced Disjoint Paths} problem asks, given a graph $G$ and terminal pairs $(s_1, t_1)$, $\ldots$, $(s_k, t_k)$,
whether $G$ contains a set of paths $P_1, \ldots, P_k$ 
such that for all $i \in \{1, \ldots, k\}$, $P_i$ is an induced $(s_i, t_i)$-path in $G$ 
and such that for all $i \neq j$,
each vertex in $V(P_i)$ has no neighbor in $V(P_j)$, except possibly when terminals coincide.
In the \textsc{Induced Disjoint Connected Subgraphs} problem, we 
are given a graph $G$ and sets of terminal vertices $T_1, \ldots, T_k \subseteq V(G)$ 
and the question is whether $G$ contains induced subgraphs $D_1, \ldots, D_k$ such that for all $i \in [k]$, 
$T_i \subseteq V(D_i)$, and such that for all $i \neq j$, $V(D_i)$ has no neighbor in $V(D_j)$, 
except possibly when terminals coincide.
For a fixed graph $H$, the \textsc{$H$-Induced Topological Minor} problem asks 
whether a graph $G$ has a subdivision of $H$ as an induced subgraph.
An \emph{acyclic $k$-coloring} of a graph is a proper coloring with $k$ colors such that each pair of color classes induces a forest.
A \emph{star $k$-coloring} of a graph is an acyclic coloring such that each pair of color classes induces a star forest.
A \emph{$b$-coloring} of a graph is a proper coloring such that each color class contains a vertex that has a neighbor in each of the other color classes.
A \emph{conflict free $k$-coloring} of a graph is a vertex-coloring (not necessarily proper) such that for each vertex, there is one color that appears at most once in its neighborhood.
Both the open and closed neighborhoods can be considered and we 
call the resulting problems \textsc{Open Conflict-free $k$-Coloring} and \textsc{Closed Conflict-free $k$-Coloring}, respectively.
\begin{proposition}
	The following problems are expressible in \ac \en logic:
	\begin{enumerate}
		\item\label{exp:fvs} \textsc{Feedback Vertex Set}
		\item\label{exp:lip} \textsc{Longest Induced Path}
		\item\label{exp:idcs} \textsc{Induced Disjoint Connected Subgraphs}
		\item\label{exp:idp} \textsc{Induced Disjoint Paths}
		\item\label{exp:Hitm} \textsc{$H$-Induced Topological Minor}
		\item\label{exp:acycol} \textsc{Acyclic $k$-Coloring}
	\end{enumerate}
	Moreover, the following problems are expressible in \en logic:
	\begin{enumerate}
		\setcounter{enumi}{6}
		\item\label{exp:starcol} \textsc{Star $k$-Coloring}
		\item\label{exp:bcol} \textsc{$b$-Coloring} with fixed number of colors $k$.
		\item\label{exp:ocfcol} \textsc{Open Conflict-free $k$-Coloring}
		\item\label{exp:ccfcol} \textsc{Closed Conflict-free $k$-Coloring}
	\end{enumerate}
\end{proposition}
\begin{proof}
	For a constant $r$ and a non-empty finite or co-finite set $\mu \subseteq \bN$, 
	we may use the term $t_{r, \mu}$ given in \cref{lem:lcv}. 
	For \textsc{Induced Disjoint Paths}, we create $k$ colors $\mathbf{P}_1 = \{s_1, t_1\}$, $\ldots$, $\mathbf{P}_k = \{s_k, t_k\}$,
	and for \textsc{Induced Disjoint Connected Subgraphs}, 
	we create $k$ colors $\mathbf{P}_1 = T_1$, $\ldots$, $\mathbf{P}_k = T_k$.
	For \textsc{$H$-Induced Topological Minor}, we assume that $V(H) = \{1, \ldots, k\}$, and $E(H) = \{i_1j_1, i_2j_2, \ldots, i_\ell j_\ell\}$.
	For a variable $X$, the predicate $\partition_X(X_1, \ldots, X_r)$ verifies that $(X_1, \ldots, X_r)$ is a partition of $X$ 
	and can be defined similarly to the $\partition$ predicate, see~\eqref{eq:exp:part}.
	
    We express the problems as follows.
    \begin{itemize}
    		\item[\ref{exp:fvs}.] $\exists X\, \card{X} \le k \land \acy(\overline{X})$
		\item[\ref{exp:lip}.] $\exists X\, \card{X} \ge k \land \conn(X) \land \acy(X) \land X \subseteq t_{1, \{1, 2\}}(X)$
		\item[\ref{exp:idcs}.] $\exists X_1\cdots \exists X_k~\idcs(X_1, \mathbf{P}_1, \ldots, X_k, \mathbf{P}_k)$, where
			\begin{align*}
				\idcs(X_1, T_1, \ldots, X_k, T_k) \equiv \bigwedge\nolimits_{1 \le i \le k} \conn(X_i)
						\land \bigwedge\nolimits_{1 \le i \le k} X_i \cap (T_1 \cup \cdots \cup T_k) = T_i \\
				~~~~~\land \bigwedge\nolimits_{1 \le i < j \le k} (X_i \setminus T_i) \cap (X_j \setminus T_j) = \emptyset 
						\land \bigwedge\nolimits_{i \neq j} X_j \cap N^1_1(X_i) \subseteq T_i \cap T_j
			\end{align*}
		\item[\ref{exp:idp}.] $\idp(\mathbf{P}_1, \ldots, \mathbf{P}_k)$, where
			\begin{align*}
				&\idp(T_1, \ldots, T_k) \equiv
				 \exists X_1 \cdots \exists X_k\, \idcs(X_1, T_1, \ldots, X_k, T_k) \land \bigwedge\nolimits_{1 \le i \le k} \acy(X_i) \\	
				&~~~~~\land T_i \subseteq t_{1, \{1\}}(X_i) 
					\land (X_1 \setminus T_1) \cup \cdots \cup (X_k \setminus T_k) \subseteq t_{1, \{2\}}(X_1 \cup \cdots \cup X_k)
			\end{align*}
		\item[\ref{exp:Hitm}.] $\exists x_1 \cdots \exists x_k\, \idp(\{x_{i_1}, x_{j_1}\}, \ldots, \{x_{i_\ell}, x_{j_\ell}\})$
		\item[\ref{exp:acycol}.] $\exists X_1 \cdots \exists X_k\, \propcol(X_1, \ldots, X_k) \land \bigwedge\nolimits_{1 \le i < j \le k} \acy(X_i \cup X_j)$, where
			\begin{align*}
				\independent(X) &\equiv N_1^1(X) \cap X = \emptyset \\
				\propcol(X_1, \ldots, X_r) &\equiv \partition(X_1, \ldots, X_r) \land \bigwedge\nolimits_{i \in \{1, \ldots, r\}} \independent(X_i)
			\end{align*}
		\item[\ref{exp:starcol}.] $\exists X_1 \cdots \exists X_k\, \propcol(X_1, \ldots, X_k) 
					\land \bigwedge\nolimits_{1 \le i < j \le k} \starfor(X_i \cup X_j)$, where
			\begin{align*}
				\starfor(X) \equiv \exists X_1 \exists X_2~\partition_X(X_1, X_2) \land \independent(X_1)  \land \independent(X_2) \land X_2 \subseteq t_{1, \{1\}}(X_1)
			\end{align*}
		\item[\ref{exp:bcol}.] $\exists X_1 \ldots \exists X_k\, \propcol(X_1, \ldots, X_k) \land 
				\exists x_1 \ldots \exists x_k\, 
				\bigwedge\nolimits_{i \in [k]} \left(x_i \in X_i \land \bigwedge\nolimits_{j \neq i} \neg (N_1^1(x_i) \cap X_j = \emptyset)\right)$
		\item[\ref{exp:ocfcol}.] $\exists X_1 \ldots \exists X_k\, \partition(X_1, \ldots, X_k) \land 
			\bigcup_{i \in [k]} \overline{N_2^1(X_i)} = \overline{\emptyset}$
		\item[\ref{exp:ccfcol}.] $\exists X_1 \ldots \exists X_k\, \partition(X_1, \ldots, X_k) \land
			\bigcup_{i \in [k]} (X_i \setminus N_1^1(X_i)) \cup \overline{N_2^1(X_i)} = \overline{\emptyset}$
    \end{itemize}
\end{proof}

\subsection{Solution diversity}\label{sec:exp:diversity}
In this section we show that diversity constraints in the sense of~\cite{BasteEtAl2020} can be encoded in \en logic,
which shows that the diverse variants of a large number of vertex subset problems in graphs can be solved in \XP time 
parameterized by the mim-width of a given rooted layout of the input graph plus the number of solutions.
For two sets $X$ and $Y$, their \emph{Hamming Distance} is the size of their symmetric difference, 
in other words $\hamd(X, Y) = \card{(X \setminus Y) \cup (Y \setminus X)}$.
Two types of diversity constraints on a collection of sets are commonly used: the \emph{minimum} pairwise Hamming distance
and the \emph{sum} over all pairwise Hamming distances; we denote them by $\mindiv$ and $\sumdiv$, respectively.
Formally, 
\begin{align*}
	\mindiv(X_1, \ldots, X_p) &= \min\nolimits_{1 \le i < j \le p} \hamd(X_i, X_j) \mbox{ and } \\
	\sumdiv(X_1, \ldots, X_p) &= \sum\nolimits_{1 \le i < j \le p} \hamd(X_i, X_j).
\end{align*}

For a subset problem $\Pi$, we denote by \textsc{Min-Diverse $\Pi$} the problem 
where we are given an instance of $\Pi$ and integers $p$ and $t$ 
and the question is whether the input graph has $p$ solutions $X_1, \ldots, X_p$ to $\Pi$ 
such that $\mindiv(X_1, \ldots, X_p) \ge t$.
Similarly, the \textsc{Sum-Diverse $\Pi$} problem asks, given the same kind of input, 
whether the input graph has $p$ solutions $X_1, \ldots, X_p$ to $\Pi$ such that $\sumdiv(X_1, \ldots, X_p) \ge t$.
We can encode ``$\mindiv(X_1, \ldots, X_p) \ge t$'' straightforwardly in \en by a formula
	\begin{align}
		\label{eq:div:min}
		\bigwedge_{1 \le i \le j \le p} \card{(X_i \setminus X_j) \cup (X_j \setminus X_i)} \ge t.
	\end{align}
Since we cannot quantify universally in \en logic, 
we have to repeat the formula encoding $\Pi$ for each of the requested $p$ solutions separately;
therefore the formula length increases by a multiplicative factor of $p$.
The encoding~\eqref{eq:div:min} of $\mindiv$ increases the formula length by another additive factor of $\calO(p^2)$
and the number of size measurements by at most $p^2$.
\begin{observation}\label{obs:div:min}
	Let $\Pi$ be a vertex subset problem that is expressible in \ac \en logic via formula~$\phi(X)$.
	Then, \textsc{Min-Diverse $\Pi$} is expressible in \ac \en logic via a formula $\phi_{div}(X_1, \ldots, X_p)$
	such that $\card{\phi_{div}} = \calO(p\card{\phi} + p^2)$,
	$d(\phi_{div}) = d(\phi)$, and $r(\phi_{div}) = r(\phi)$.
	Moreover, if $\phi$ has $s$ size measurements, 
	then $\phi_{div}$ has at most $s + p^2$ size measurements.
\end{observation}

For \textsc{Sum-Diverse $\Pi$}, we can do something slightly more efficient, 
by assigning an appropriate weight function to the input graph and 
using the objective value.
This way we do not increase the size measurements in the formula.
Suppose \textsc{Sum-Diverse $\Pi$} asks for $p$ solutions.
Then, we take the input graph $G$ to the vertex subset problem $\Pi$
and turn $G$ into a $p$-weighted graph,
by assigning, to each vertex $v \in V(G)$ and each $S \subseteq [p]$,
the weight $w(v, S) = \card{S}\cdot(p-\card{S})$.
Then, given a $p$-tuple of solutions $\tB \in \Pset(V(G))^p$ to $\Pi$,
we have that
\begin{align*}
	\obj(G,\tB) &= \sum_{v \in V(G)} w(v,\{ i \mid v \in B_i \})
				= \sum_{v \in V(G)} \card{\{i \mid v \in B_i\}} \cdot \card{\{i \mid v \notin B_i\}} \\
				&=\sumdiv(B_1, \ldots, B_p).
\end{align*}

It is not difficult to derive the last equality, we only have to count the contribution of each 
vertex to $\sumdiv(B_1, \ldots, B_p)$ separately;
this equality has also been derived in~\cite{BasteEtAl2020}.
We have the following observation.
\begin{observation}\label{obs:div:sum}
	Let $\Pi$ be a vertex subset problem that is expressible in \ac \en logic via formula~$\phi(X)$.
	Then, \textsc{Sum-Diverse $\Pi$} is expressible as an \ac\en optimization problem via a formula $\phi_{div}(X_1, \ldots, X_p)$
	such that $\phi$ and $\phi_{div}$ have the same number of size measurements,
	$\card{\phi_{div}} = \calO(p\card{\phi})$, 
	$d(\phi_{div}) = d(\phi)$, and $r(\phi_{div}) = r(\phi)$.
\end{observation}

\begin{corollary}\label{cor:div}
	Let $\Pi$ be a vertex subset problem expressible in \ac \en logic via formula $\phi$.
	Let $d \defeq d(\phi)$, $r \defeq r(\phi)$, and $s$ be the number of size measurements.
	If the input graph to $\Pi$ is given together with a rooted layout of $\mathsf{f}$-width $w$,
	we can solve \textsc{Min-Diverse $\Pi$} and \textsc{Sum-Diverse $\Pi$}, asking for $p$ solutions, 
	in time as shown in \cref{tab:div:runtime}.
\end{corollary}
\begin{table}
	\begin{tabular}{r|c|c|c}
		Variant & clique-w.\ $w$ & rank-w.\ $w$ & mim-w.\ $w$ \\
		\hline
		\textsc{Min-Diverse $\Pi$}, $r = 1$ &  $2^{\calO(dw(p\card{\phi} + p^2))}n^{\calO(s + p^2)}$ & $2^{\calO(dw^2(p\card{\phi} + p^2))}n^{\calO(s + p^2)}$ & 
			$n^{\calO(wdp(\card{\phi} + p))}$ \\
		\textsc{Min-Diverse $\Pi$}	&  $2^{\calO(d(wr(p\card{\phi} + p^2))^2)}n^{\calO(s + p^2)}$ & $2^{\calO(dw^4(r(p\card{\phi} + p^2))^2)}n^{\calO(s + p^2)}$ & $n^{\calO(dw(p\card{\phi} + p^2))}$ \\
		\hline
		\textsc{Sum-Diverse $\Pi$}, $r = 1$ &  $2^{\calO(dwp\card{\phi})}n^{\calO(s)}$ & $2^{\calO(dw^2p\card{\phi})}n^{\calO(s)}$ & $n^{\calO(dwp\card{\phi})}$ \\
		\textsc{Sum-Diverse $\Pi$}	&  $2^{\calO(d(wrp\card{\phi})^2)}n^{\calO(s)}$ & $2^{\calO(dw^4(rp\card{\phi})^2)}n^{\calO(s)}$ & $n^{\calO(dw(p\card{\phi})^2)}$
	\end{tabular}
	\caption{Running times of \cref{cor:div}. Note in particular that when $s$ is a fixed constant, 
		we achieve \FPT running times parameterized by the number of solutions plus clique-width/rank-width for \textsc{Sum-Diverse $\Pi$},
		while for \textsc{Min-Diverse $\Pi$} we have an $\XP$ dependence on the number of solutions.
		}
	\label{tab:div:runtime}
\end{table}

One interesting aspect of the previous observations is that they give applications of the parameterized solution diversity paradigm 
on several well-studied graph classes, as long as they have constant mim-width.
In particular, for all problems $\Pi$ as in Observation~\ref{obs:div:min}~and~\ref{obs:div:sum},
\textsc{Min-Diverse $\Pi$} and \textsc{Sum-Diverse $\Pi$} parameterized by $r$ alone are in \XP 
when the input graph is restricted to, for instance, 
interval graphs, permutation graphs, circular arc graphs, and $H$-graphs (given an $H$-representation);
see for instance~\cite{JaffkeThesis,VatshelleThesis} for an overview.
It would be interesting to see if these problems are actually in \FPT or \W[1]-hard,
or to study solution diversity in graph classes of unbounded mim-width, 
e.g.\ in planar or bipartite graphs.

Inspecting \cref{tab:div:runtime}, 
we observe that we lift the \FPT-results of Baste et al.~\cite{BasteEtAl2020}
for \textsc{Sum-Diverse $\Pi$} from the parameterization treewidth plus number of solutions to clique-width (or rank-width) plus number of solutions,
albeit only for vertex subset problems expressible in \ac \en logic.
Moreover, it shows that when the number of solutions is a fixed constant, 
we get \FPT-algorithms parameterized by clique-width (or rank-width) for \textsc{Min-Diverse $\Pi$} as well.
\section{Comparison to modal logic}\label{sec:modal}
In this section, we introduce the \emph{existential distance modal logic} \edml and show that it is equivalent to \en.
The logic \edml is obtained from the original existential counting modal logic \ecml, introduced by Pilipczuk \cite{Pilipczuk11},
by removing edge-set quantification and ultimately periodic counting
and adding the ability to query $d$-neighborhoods in the $r$-th power of the graph.
In \Cref{sec:hardness} we show that this restriction is necessary, i.e., that model checking for the extension of \edml with edge-set quantification
or ultimately periodic counting is NP-hard on graph classes with constant mim-width.

Modal logic differs from the other logics mentioned in this paper,
in the sense that evaluation of formulas in modal logic is tied to an active vertex that changes over time and is not explicitly represented by a variable of the formula.
The central operators are $\square$ and $\diamondsuit$ and variations thereof.
In the beginning, when there is no active vertex yet, the operators $\rangle\square$ and $\rangle\diamondsuit$ quantify the active vertex existentially and universally, respectively.
When the active vertex is $v$,
a formula $\square \phi$ shall be read as ``$\phi$ holds on some neighbor $w$ of $v$'',
while $\diamondsuit \phi$ shall be read as ``$\phi$ holds on all neighbors $w$ of $v$''.
When testing whether $\phi$ holds on a neighbor $w$ of $v$, the active vertex then becomes $w$ during the evaluation.
For additional expressiveness, we 
further allow $\square^r_d \phi$, meaning ``$\phi$ holds on at least $d$ vertices $w$ with distance between 1 and $r$ to $v$''.
The corresponding $\diamondsuit^r_d$-operator is defined by $\diamondsuit^{r}_d \phi = \neg \square^r_{d} \neg \phi$.
The logic \edml consists of so-called \emph{inner formulas}
that are surrounded by a block of existential set quantifiers.

\begin{itemize}
	\item Every set variable $X$ and every unary relational symbol $\mathbf{P}$ is an unfinished inner formula.
    \item $\square^r_{d} \phi$ is an unfinished inner formula for every $d,r \in \N^+$ and unfinished inner formula $\phi$.
    \item $\rangle \square_{d} \phi$ is a finished inner formula for every $d \in \N^+$ and unfinished inner formula $\phi$.
    \item $|X| \le m$ and $|X| \ge m$ are finished inner formulas for every variable $X$ and $m \in \N$.
    \item If $\phi$ and $\psi$ are unfinished/finished inner formulas then so are $\neg\phi$, $\phi \land \psi$.
\end{itemize}

\edml formulas may existentially quantify over set variables, as the following rules say.
\begin{itemize}
    \item Every finished inner formula is a formula.
    \item If $\phi$ is a formula and $X$ is a set variable then $\exists X\, \phi$ is a formula.
\end{itemize}

Besides the obvious rules for the Boolean combinations, existential quantification and size measurements, the semantics of \edml are as follows.
Note that only the semantics of unfinished inner formulas depend on the current active vertex.
\begin{itemize}
    \item $\ip{X}^{(G,\beta,v)} = 1$ iff $v \in \beta(X)$.
    \item $\ip{\mathbf{P}}^{(G,\beta,v)} = 1$ iff $v \in \mathbf{P}(G)$.
    \item $\ip{\square^r_{d} \phi}^{(G,\beta,v)} = 1$ iff the number of $w \in N^r(v)$ with $\ip{\phi}^{(G,\beta,w)}$ is at least $d$.
    \item $\ip{\rangle \square_{d} \phi}^{(G,\beta)} = 1$ iff the number of $v \in V(G)$ with $\ip{\phi}^{(G,\beta,v)}$ is at least $d$.
\end{itemize}

We treat $\diamondsuit_{d}^r \phi$ as a shorthand for $\neg \square^r_{d} \neg \phi$.
We also write $\square$ and $\diamondsuit$ as a shorthand for $\square_1^1$ and $\diamondsuit_1^1$.
Then for example $\rangle\diamondsuit \phi$ 
means ``all elements satisfy $\phi$''.

\begin{lemma}\label{lem:convert_en_to_modal}
For every formula $\phi \in \text{\en}$, there exists an equivalent formula $\phi' \in \text{\edml}$ with $|\phi'| = O(|\phi|)$.
\end{lemma}
\begin{proof}
    By \Cref{obs:core}, we can assume $\phi \in \text{\encore}$.
    \encore formulas consist of an existential block followed by a Boolean combination of primitive formulas.
    On the other hand, formulas of \edml consist of an existential block followed by a Boolean combination inner formulas.
    The following list shows that every primitive formula of \encore is equivalent to some inner formula of \edml.
    \begin{itemize}
        \item $\mathbf{P}=X$ is equivalent to $\rangle\diamondsuit (\mathbf{P} \leftrightarrow X)$.
        \item $X = \comp Y$ is equivalent to $\rangle\diamondsuit(X \leftrightarrow \neg Y)$.
        \item $X = Y$ is equivalent to $\rangle\diamondsuit(X \leftrightarrow Y)$.
        \item $X \cap Y = Z$ is equivalent to $\rangle\diamondsuit(X \land Y \leftrightarrow Z)$.
        \item $N^r_d(X) = Y$ is equivalent $\rangle\diamondsuit ((\square^r_{d} X) \leftrightarrow Y)$.
        \item $|X| \le m$  is equivalent to $|X| \le m$.
    \end{itemize}
    We can therefore convert every \encore formula into an equivalent formula of \edml
    by substituting according to the previous list.
\end{proof}

\begin{lemma}\label{lem:convert_modal_to_en}
For every formula $\phi \in \text{\edml}$, there exists an equivalent formula $\phi' \in \text{\en}$ with $|\phi'| = O(|\phi|)$.
\end{lemma}
\begin{proof}
    We consider the extension of \edml using these two rules.
    \begin{itemize}
        \item Every neighborhood term $t$ is a unfinished inner formula with the semantics $\ip{t}^{(G,\beta,v)} = 1$ iff $v \in \ip{t}^{(G,\beta)}$.
        \item If $t$ is a neighborhood term and $m \in \N$ then $t \ge m$ is a finished inner formula with semantics
                $\ip{t \ge m}^{(G,\beta)} = 1$ iff 
                $|\ip{t}^{(G,\beta)}| \ge m$.
    \end{itemize}
    Note that for all graphs $G$, assignments $\beta$, $v \in V(G)$
    and neighborhood terms $t$ and $t'$,
    \begin{itemize}
        \item $\ip{t \lor t'}^{(G,\beta,v)} = \ip{t \cup t'}^{(G,\beta,v)}$,
        \item $\ip{\neg t}^{(G,\beta,v)} ~~~\, = \ip{\comp t}^{(G,\beta,v)}$,
        \item $\ip{\square^r_d t}^{(G,\beta,v)} ~\hspace{0.1cm} = \ip{N^r_d(t)}^{(G,\beta,v)}$,
        \item $\ip{\rangle\square_d t}^{(G,\beta)} ~~\hspace{0.1cm} = \ip{t \ge d}^{(G,\beta)}$.
    \end{itemize}
    Let us fix an \edml formula.
    We repeatedly substitute, in accordance with the list above,
    inner unfinished subformulas of the form $t \lor t'$ with $t \cup t'$; $\neg t$ with $\comp t$; and $\square^r_d t$ with $N^r_d(t)$.
    We also substitute all inner finished subformulas of the form $\rangle\square_d t$ with $t \ge d$.
    By the previous observation, these operations preserve the semantics of the formula.
    After doing these substitutions exhaustively, the result is a \en formula.
\end{proof}

The combination of \Cref{lem:convert_en_to_modal} and \Cref{lem:convert_modal_to_en} yields
the following equivalence.
\begin{theorem}
\en and \edml are equally expressive.
\end{theorem}
\section{Hardness of model checking for extensions}\label{sec:hardness}

In this section, argue that \en is at the edge of what is tractable on bounded mim-width, 
in the sense that already slight extensions of this logic are \paraNPhard to evaluate parameterized by formula length and mim-width.
In particular, we will show that two features present in the 
original existential counting modal logic \ecml
--- edge set quantification and ultimately periodic counting \cite[Section 3.1]{Pilipczuk11} ---
are not tractable on mim-width.
\cref{thm:hardness:extensions} follows directly from \cref{lem:edgeSetHard,lem:parityhard2,lem:oneUniversalHard} presented in this section.

\myparagraph{One universal quantifier}
We start by showing that one cannot add universal quantification to our neighborhood logic.
In particular, even if we allow only a single universal quantifier per formula, we can already express \textsc{Clique} by the sentence
\[
    \exists X \, \card{X} \ge k \wedge \forall X' \, \bigl( X' \cap X \neq \emptyset) \to  \bigl( X \subseteq X' \cup N_1^1(X') \bigr).
\]
Let \textsf{DN+$\forall$} be the logic obtained from \en by allowing a single innermost universal quantifier
but also requiring $R(\phi)=\{1\}$ and $d(\phi)=2$ for all formulas $\phi$.
Since the above sentence lies in \textsf{DN+$\forall$}, complements of planar graphs have linear mim-width at most $6$~\cite{BelmonteVatshelle2013},
and \textsc{Independent Set} is \NP-hard on planar graphs, we have the following observation.
\begin{lemma}\label{lem:oneUniversalHard}
    Model checking for
    \textsc{\textsf{DN+$\forall$}} is \paraNPhard 
	parameterized by formula length plus (linear) mim-width.
\end{lemma}

Note that model checking becomes para-co\NP-hard if we allow a single \emph{outermost} (instead of innermost) universal quantifier,
as the complement of the \textsc{Clique}-sentence above falls within this logic.
Already with a single innermost universal \emph{vertex} quantifier we can express \textsc{Clique} via
$\exists X \, \card{X} \ge k \wedge \forall x \, \bigl( x \in X) \to  \bigl( X \subseteq N(x) \bigr)$,
while a single outermost universal vertex quantifier can always be evaluated with an additional run-time factor of $O(n)$.

\myparagraph{Edge Set quantification}
\ecml allows existential quantification over edge sets.
For every edge set $Y$, the logic allows the operator $Y$, which checks whether
the edge that was used to directly access the current active vertex belongs to $Y$.

A central problem in this section is the NP-hard~\cite{karp1972reducibility} \textsc{Max Cut} problem.
The input is a graph $G$ and an integer $k$ and the task is to decide whether there is a set $S \subseteq V(G)$
whose \emph{cut size} $|\{uv \in E(G) \colon u \in S \text{ xor } v \in S\}|$ is at least $k$.
If we add edge set quantification to \edml, we can construct the following formula
that holds if and only if a graph has a cut of size $k$
\[
    \exists_{\subseteq V} X \, \exists_{\subseteq E} Y \, |Y| \ge k \land \rangle\diamondsuit \Bigl( \bigl( 
    X \to \diamondsuit (Y \to \neg X)) \land \diamondsuit( \neg X \to \diamondsuit (Y \to X) \bigr) \Bigr).
\]
The subformula $\rangle\diamondsuit \bigl( ( X \to \diamondsuit (Y \to \neg X)) \land \diamondsuit( \neg X \to \diamondsuit (Y \to X)) \bigr)$
ensures that $Y$ contains exactly those edges with one endpoint in $X$ and one endpoint outside $X$.

We get the same expressive power if we extend \en with edge set quantification
and an operator $N_Y(t)$ that evaluates to all vertices that are adjacent to a term $t$ via an edge in an edge set $Y$.
With this extension, the following formula expresses that
a graph has a cut of size $k$
\[
    \exists_{\subseteq V} X \, \exists_{\subseteq E} Y \, |Y| \ge k \land N_Y(X)=\comp{X} \land N_Y(\comp{X})=X.
\]

Let \textsf{DN+EdgeSets} be the logic obtained from \en by adding edge set quantification and operators $N_Y(\cdot)$,
but also requiring $R(\phi)=\{1\}$ and $d(\phi)=2$ for all formulas $\phi$.
Since \textsc{Max~Cut} is NP-complete on interval graphs~\cite{adhikary},
the following observation follows immediately.
\begin{lemma}\label{lem:edgeSetHard}
    Model checking \textnormal{\textsf{DN+EdgeSets}} is
    \paraNPhard parameterized by formula length on interval graphs,
    and therefore \paraNPhard parameterized by formula length plus (linear) mim-width.
\end{lemma}

Since \textsc{Max Cut} is further W[1]-hard parameterized by the clique-width of a graph~\cite{fomin2010algorithmic},
we get the following hardness result for clique-width.
\begin{lemma}\label{lem:edgeSetHardCW}
    There exists a constant $c$ such that
    evaluating \textnormal{\textsf{DN+EdgeSets}} formulas of length at most $c$
    \textnormal{W[1]}-hard parameterized by the clique-width of the input graph.
\end{lemma}

\myparagraph{Ultimately Periodic Counting}
A set $S \subseteq \N$ is \emph{ultimately periodic}
if there exist positive integers $N$, $k$ such that 
for all $n \ge N$, $n \in S$ iff $n+k \in S$.
\ecml allows operators $\square_S \phi$ which should be read as
``the number of neighbors with property $\phi$ of the current active vertex is in $S$''.
We obtain hardness results already
if we only allow the additional operator $\square_\text{even}$,
which we use as a synonym for $\square_S$ with $S=\{2i \colon i \in \N\}$.

Following the arguments in \Cref{lem:convert_en_to_modal} and \Cref{lem:convert_modal_to_en},
it is easy to see that 
in the context of neighborhood logic, this is equivalent
to adding an operator $N_\text{even}(\cdot)$ to \en,
where $N_\text{even}(t)$ stands for the set of all vertices with an even number of neighbors in $t$.
For example the formula $\exists X \, \comp X = \emptyset \land |N_\text{even}(X)| = k$
expresses that there are $k$ vertices with even degree.
Let the logic \textsf{DN+Parity} be obtained from \en by adding the operator $N_\text{even}(\cdot)$,
but also requiring $R(\phi)=\{1\}$ and $d(\phi)=2$ for all formulas $\phi$.
\begin{lemma}\label{lem:parityhard2}
    Model checking for \textnormal{\textsf{DN+Parity}} is
    \paraNPhard parameterized by formula length on interval graphs,
    and therefore \paraNPhard parameterized by formula length plus (linear) mim-width.
\end{lemma}

\myparagraph{Hardness of Parity Extension}
The remainder of this section is concerned with proving \Cref{lem:parityhard2}.
In order to do so,
let us start by describing the NP-hard problem we reduce from.
We call this problem \textsc{Parity Interval Selection} and it is closely related to the NP-hard~\cite{karp1972reducibility} \textsc{Max Cut} problem.
We say a \emph{row} is a tuple of non-overlapping non-empty closed intervals of rational numbers,
such as for example $([0,3], [4,5.5], [6,6.1])$.
The input to the \textsc{Parity Interval Selection} problem is a number $k$, a set of rows (called \emph{constraint rows}) and an additional row (called \emph{selection row})
satisfying the following criteria.
For all the intervals in constraint rows, we require that the start- and endpoints are integers
and that no two intervals (of possibly different constraint rows) have the same start- or endpoints.
We further require every interval of the selection row to be of the form $[i+0.4,i+0.6]$ for some integer $i$.
For a constraint row $R$, the \emph{selection neighborhood} (denoted $N(R)$) is the set of intervals in the selection row that overlap with an interval in $R$.
The task of the \textsc{Parity Interval Selection} problem is to decide whether there exists a set of intervals $S$ from the selection row
such that there are exactly $k$ constraint rows $R$ for which $N(R) \cap S$ is odd.
See the top of \Cref{fig:intervalselection} for a visualization of a problem instance.

Let us consider the \textsc{Exact Cut} problem.
The input is a graph $G$ and an integer $k$ and the task is to decide whether there is a cut of size exactly $k$.
A simple Turing-reduction from the NP-hard~\cite{karp1972reducibility} \textsc{Max Cut} problem shows that \textsc{Exact Cut} is also NP-hard (by trying out all values for $k$).
We show that \textsc{Parity Interval Selection} is NP-hard by reducing from \textsc{Exact Cut}.
Let $(G,k)$ be an \textsc{Exact Cut} instance for which we will now construct an equivalent \textsc{Parity Interval Selection} instance.
The value $k$ in the \textsc{Parity Interval Selection} instance is the same as in the \textsc{Exact Cut} instance.
The selection row contains $|V(G)|$ many intervals that we identify with the vertices of the input graph.
For every edge $uv$, we create a constraint row $R_{uv}$ such that $N(R_{uv}) = \{u,v\}$ (this requires two intervals per constraint row).
We further shift all endpoints such that additionally no two intervals in any constraint row have the same start- or endpoints.
Now, for every set $S \subseteq V(G)$, 
the cut size of $S$ in $G$ (i.e., the number of edges with exactly one endpoint in $S$)
is equal to the number of constraint rows with exactly one selection neighbor in $S$.
Thus $(G,k)$ is a yes-instance of \textsc{Exact Cut} iff the corresponding \textsc{Parity Interval Selection} instance is a yes-instance.

\begin{proof}[Proof of \Cref{lem:parityhard2}]
We now show para-\NP-hardness of model checking for the parity extension of \en on interval graphs.
We are given an instance of \textsc{ParityIntervalSelection} consisting of a number $k$, a set of constraint rows $\cal C$ and one selection row $Q$.
The rows ${\cal C} \cup \{Q\}$ describe an interval graph $H$ with two types of vertices: selection intervals and constraint intervals.
We are going to construct a supergraph $G$ of $H$ by adding auxiliary selection and constraint intervals, as well as a third type of vertices, which call \emph{target intervals}.
Let $t$ be an integer such that no interval of any row from ${\cal C} \cup \{Q\}$ contains a number equal or larger to $t$.
We associate with each constraint row a so-called \emph{target interval} that is completely contained in $[t,\infty]$ (drawn red in \Cref{fig:intervalselection}).
We do it in such a way that the start- and endpoints of all target intervals are integers
and no two target or constraint intervals (also among different rows) have the same start- or endpoints.
Furthermore, for every constraint row $R \in \cal C$ with intervals $I_1,\dots,I_{l}$ in increasing order and corresponding target interval $I_{l+1}$,
and every $1 \le i \le l$, given $I_i=[a,b]$ and $I_{i+1} = [c,d]$, we add three additional intervals to $G$.
\begin{itemize}
    \item an auxiliary constraint interval $[b+0.1,b+0.2]$ (which we refer to as \emph{the tiny interval right of $I_i$ and left of $I_{i+1}$})
    \item auxiliary selection intervals $ [b+0.1,c-0.1]$ and $[b-0.1,c+0.1]$ (which we refer to as \emph{the twin tuple $(A,B) = ([b+0.1,c-0.1],[b-0.1,c+0.1])$ right of $I_i$ and left of $I_{i+1}$}).
\end{itemize}
\begin{figure}
\includegraphics[width=\textwidth]{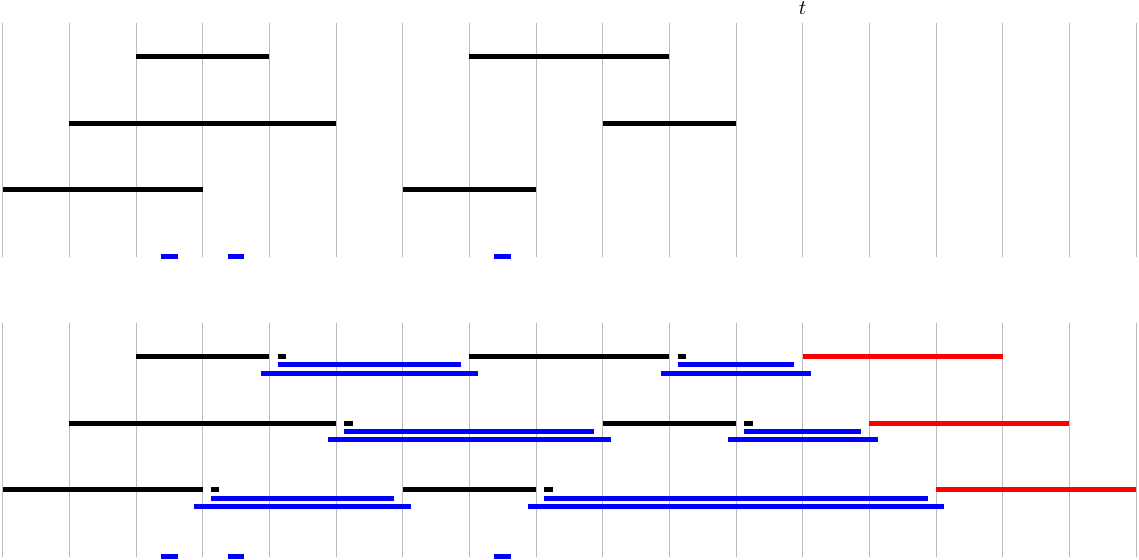}
\caption{Top: Visualization of an instance of the \textsc{ParityIntervalSelection} problem. Parity rows are black and the selection row is blue.
    It is a yes-instance iff $k \in \{0,1,2,3\}$.
Bottom: The corresponding interval graph constructed during the reduction to the model checking problem for the parity extension of \edml.
For every gap in a constraint row, two new selection intervals (blue) and one new constraint interval (black) are added. An additional target interval (red) is added at the end of each constraint row.}
\label{fig:intervalselection}
\end{figure}
The result is of this construction is shown at the bottom of \Cref{fig:intervalselection}.
The newly created interval graph $G$ has three types of vertices: 
selection intervals $\textbf{P}_\text{selection} \subseteq V(G)$, constraint intervals $\textbf{P}_\text{constraint} \subseteq V(G)$ and target intervals $\textbf{P}_\text{target} \subseteq V(G)$.
We say a set of selection intervals $S \subseteq \textbf{P}_\text{selection}$ is a \emph{solution candidate} in $G$
if every constraint interval has an even number of neighbors in $S$.

\begin{claim}\label{claim1}
    Let $S$ be a solution candidate in $G$.
    The target interval associated with a constraint row $R \in \cal C$ has an even number of neighbors from $S$ in $G$
    iff in the original \textsc{ParityIntervalSelection} instance $|N(R) \cap S'|$ is even, where $S'=S\cap V(H)$.
\end{claim}
\begin{proof*}[Proof of \Cref{claim1}]
    Assume the twin tuples of $G$ in ascending order by their starting point are $(X_1,Y_1), \dots, (X_m,Y_m)$.
    At first, we show by induction on $i$ that for all $j \le i$ either both or none of $X_j$, $Y_j$ are contained in $S$.
    For $i=0$, the statement is clear.
    Let $i>0$.
    We constructed $G$ such that for every tiny interval $T$ and every twin tuple $(A,B)$, $T$ is either adjacent to both or none of $A$, $B$.
    Let $T$ be the tiny interval with $X_i, Y_i \in N(T)$ and $N(T) \subseteq \{X_{1},Y_{1},\dots,X_{i},Y_{i}\}$.
    Such an interval exists, since no constraint intervals have the same start- and endpoints.
    By our induction hypothesis, all $j \le i$ either both or none of $X_j$, $Y_j$ are contained in $S$.
    Furthermore, $T$ is either adjacent to both or none of $X_j$, $Y_j$.
    Thus, $T$ is adjacent to an even number of intervals from
    $S \cap \{X_{1},Y_{1},\dots,X_{i-1},Y_{i-1}\}$.
    Since $S$ is a solution candidate, $T$ needs to have an even number of neighbors in $S$
    and therefore either both or none of $X_i, Y_i$ are contained in $S$.

    Let $S'$ be the subset of $S$ containing only intervals from the original \textsc{ParityIntervalSelection} instance (i.e., we exclude the twin intervals).
    Assume a constraint row $R$ consists of the intervals $I_1, \dots, I_l$ in ascending order, with corresponding target interval $I_{l+1}$.
    For $1 \le i \le l$ let $B_i$ be the longer of the two intervals in the twin tuple on the right of $I_i$.
    Since $S$ is a solution candidate, $|N(I_i) \cap S|$ is even for every $i$.
    This means $|N(I_i) \cap S'|$ is odd iff $|N(I_i) \cap S \setminus S'|$ is odd.
    $S \setminus S'$ contains only twin tuples.
    As discussed earlier, for all twin tuples $(A, B)$, either both or none of $A,B$ are contained in $S \setminus S'$.
    Furthermore, if $I_i$ is adjacent to exactly one of $A,B$ then $B=B_i$ or $B=B_{i-1}$.
    Therefore $|N(I_i) \cap S'|$ is odd iff $|N(I_i) \cap S \setminus S'|$ is odd iff $|S \cap \{B_i \cup B_{i-1}\}|=1$.
    Since $B_0$ does not exist, we have as a special case
    that $|N(I_1) \cap S'|$ is odd iff $B_1 \in S$.
    It follows by induction on $i$ that $\sum_{j=1}^i |N(I_j) \cap S'|$ odd iff $B_i \in S$.
    In particular, the constraint row $R \in \cal C$ has an even number of neighbors from $S$ in $G$,
    (i.e., $\sum_{j=1}^l |N(I_j) \cap S'|$ is even) iff $B_l \not\in S$.
    At last, observe that the target interval $I_{i+1}$ has an even number of neighbors from $S$
    iff $B_l \not\in S$.
\end{proof*}

By \Cref{claim1}, we can solve the original \textsc{ParityIntervalSelection} instance
by deciding whether there exists a solution candidate $S$ in $G$ such
that there are exactly $k$ target intervals with an odd number of neighbors in $S$.
This is expressible in the parity extension of \en by the sentence
\[
\phi_k \equiv \exists S\, \exists T\, \textit{candidate}(S) \land |T|=k \land T = N_\text{odd}(S) \cap \textbf{P}_\text{target}
\]
where 
\[
\textit{candidate}(S) \equiv S \subseteq \textbf{P}_\text{selection} \land \textbf{P}_\text{constraint} \subseteq N_\text{even}(S)
\]
and $\textbf{P}_\text{target}$, $\textbf{P}_\text{selection}$ and $\textbf{P}_\text{constraint}$
are unary properties describing the target, selection and constraint intervals.
The length of $\phi_k$ is constant and independent of $k$.
Thus model checking is para-\NP-hard on interval graphs for this logic.
\end{proof}

\bibliographystyle{plain}
\bibliography{ref}

\newpage
\appendix

\section{Proof of \Cref{lem:computerep}}

\computerep*
\begin{proof}
	Let $A\subseteq V(G)$, $d\in \N^+$ and $R\subseteq \bN^+$.
	For every $B\subseteq A$, we define the $(\comp{A},R)$-matrix $M_B$ such that $M_B[v,r]=\min(d,\abs{N^r(v)\cap B})$.
	Note that $B\equi{A}{d,R} C$ if and only if $M_B=M_C$.
	We claim that the lists $LRep,LMRep$ and the pointers computed by \Cref{algo:Rep} satisfied the conditions of the lemma.
	Observe that we add a set $C$ to $LRep$ only when $M_C$ is not contained in $LMRep$.
	Consequently, the sets in $LRep$ are pairwise incomparable for $\equi{A}{d,R}$.
	
	\begin{algorithm}[ht]
		\SetAlgoLined
		Initialize $LRep,LMRep,NextLevel$ to be empty and $LastLevel=\{\emptyset\}$\;
		\While{$LastLevel\neq \emptyset$}
		{
			\For{every $B\in LastLevel$ in lexicographic order}
			{
				\For{every vertex $v$ in $A\setminus B$ in lexicographic order}
				{
					$C=B\cup \{v\}$\;
					Computes $M_{C}$ from $M_{B}$\;
					\If{$M_{C}$ is not contained in $LMRep$}
					{
						Add $C$ to both $LRep$ and $NextLevel$\;
						Add $M_C$ to $LMRep$\;
						Add a pointer between $M_C$ and $C$\; 
					}
				}
			}
			Set $LastLevel=NextLevel$ and $NextLevel=\emptyset$\;
		}
		\Return $LRep$ and $LMRep$\;
		\caption{Computation of the data structures to represent $\Rep{A}{d,R}$ and compute $\rep{A}{d,R}$.}
		\label{algo:Rep}
	\end{algorithm}

	Assume towards a contradiction that $\Rep{A}{d,R}\not\subseteq LRep$.
	Let $U\in \Rep{A}{d,R}\setminus LRep$ such that $U$ is minimal for the inclusion.
	Let $v\in U$ and  $S_v= \rep{A}{d,R}(U\setminus \{v\})$.
	Since $S_v\equi{A}{d,R} U\setminus \{v\}$, 
	we have $S_v \cup \{v\}\equi{A}{d,R} U$.
	Since $S_v\in \Rep{A}{d,R}$, we deduce that $\abs{S_v}\leq \abs{U}-1$.
	Moreover, $S_v \cup \{v\}\equi{A}{d,R} U$ and $U\in \Rep{A}{d,R}$ implies that $\abs{U}\leq \abs{S_v} +1$.
	Hence, we have $\abs{S_v}=\abs{U\setminus \{v\}}$, this implies that $S_v$ is lexicographically smaller than $U\setminus \{v\}$ and $U$ is lexicographically smaller than $S_v\cup\{v\}$.
	We conclude that $S_v= U\setminus \{v\}$.
	As $S_v\in \Rep{A}{d,R}(A)$ and $U$ is chosen to be minimal in $\Rep{A}{d,R}\setminus LRep$, we have $S_v=U\setminus \{v\}\in LRep$.
	Consequently, $S_v\cup \{v\}=U$ is considered by the algorithm to be added to $LRep$.
	Thus, there is a set $U'$ added to $LRep$ such that $M_{U}=M_{U'}$.
	Since we consider the sets in $LastLevel$ and the vertices in $A$ in lexicographic order, we deduce that $U'=U$, yielding a contradiction.
	
	As $\Rep{A}{d,R}\subseteq LRep$ and the sets in $LRep$ are pairwise incomparable for $\equi{A}{d,R}$, we conclude that $LRep=\Rep{A}{d,R}$.
	
	It remains to prove the running times.
	The size of $\Rep{A}{d,R}=LRep$ is $\nec_{d,R}(A)$.
	The total number of sets considered by the algorithm is $\nec_{d,R}(A)\cdot n$.
	For every $B\in LRep$ and $v\in A\setminus B$, computing $C=B\cup \{v\}$ and $M_C$ can be done in time $O(n\cdot\abs{R})$ by updating the entries of the vertices $u\in \comp{A}$ such that $\dist(u,v)\leq r$.
	
	By implementing $LMRep$ with a self-balanced binary search tree, checking whether $M_C$ is in $LMRep$ and adding $M_C$ to $LMRep$ can be done in $O(\log(\nec_{d,R}(A))\cdot \abs{R} \cdot n)$.
	Hence, the construction of $LRep$ and $LMRep$ takes $O(\nec_{d,R}(A)\cdot \abs{R}\cdot n^2 \cdot \log(\nec_{d,R}(A)))$.
	Given $B\subseteq A$, we can compute $M_B$ in time $O(\abs{R}\cdot \abs{X}\cdot n)$ by consulting the distance matrix. Hence, finding a pointer to $\rep{A}{d,R}(B)$ can be done with a binary search in time $O(\log(\nec_{d,R}(A))\cdot \abs{R} \cdot n^2)$.
\end{proof}
\section{Omitted proofs from \Cref{ssec:A&C}}
In this section, we prove several lemmas for our model checking theorem on \ac \en.
We need the following lemma.

\begin{lemma}[\cite{BodlaenderCKN15}]\label{lem:basis:matrix} 
	Let $M$ be a binary $n \times m$-matrix with $m\leq n$ and let $w:\{1,\ldots,n\} \to \bN$ be a weight function on the rows of $M$. Then, one can find a basis of maximum weight of the row space of $M$ in time $O(nm^{O(1)})$.
\end{lemma}

Given a graph $G$ and $B \subseteq V(G)$, we denote by $\cc(B)$ the set of connected components of $G[B]$.
A \emph{consistent cut} of $ B$ is a ordered bipartition $(B_L,B_R)$ of $ B$ such that $N_1^1(B_R)\cap B_L=\emptyset$ (potentially $B_L$ or $B_R$ is empty). We denote by $\ccuts(B)$ the set of all consistent cuts of $B$. 
Whenever we use these notations, the graph $G$ is clear from the context.

\newcommand{\btB}{\tilde{\mathbb{B}}}

\lemRankbasedApproach*
\begin{proof}
	Let $\btE\in (\Rep{\comp{A}}{\phi})^{\abs{\var(\phi)}}$ and $\cB\subseteq \Pset(A)^{\abs{\var(\phi)}}$ such that the elements of $\cB$ are pairwise equivalent for $\equi{A}{\phi}$.
	We denote by $\btB$ the representative in $(\Rep{A}{\phi})^{\abs{\var(\phi)}}$ such that $\btB\equi{A}{\phi}  \tB$ for every $ \tB\in \cB$. Let $[\btE]$ be the set of all $ \tD\in \Pset(\comp{A})^{\abs{\var(\phi)}}$ such that $ \tD \equi{\comp{A}}{\phi} \btE$.
	
    If $\conn(\phi)=\emptyset$, then for every $\tD\in \Pset(\comp{A})^{\abs{\var(\phi)}}$, we have $\best_{\Pi_\phi}^{\conn}(\cB, \tD)= \max \{\obj_G(\tB)\mid \tB\in  \cB\}$ by definition of $\best_{\Pi_\phi}^{\conn}$. In this case, we return $\{\tB\}$ with $\tB\in \cB$ maximizing $\obj_G(\tB)$, this can be done in time $O(\abs{\cB}\cdot n \cdot \abs{\phi})$. 
	
	We suppose from now that $\conn(\phi)\neq\emptyset$.
	We assume w.l.o.g. that the variables $X_1,\dots,X_k$ in $\conn(\phi)$ are the first $k$ variables of $\var(\phi)$, in particular for every $i\in[k]$, $ \tB\in \cB$, and $ \tD\in [\btE]$, we have $\ip{\conn(X_i)}=1$ iff $G[B_i\cup D_i]$ is connected.
	
	As in \cite{BergougnouxKante2021}, we need to take care of the special cases when $D_i=\emptyset$ for some $i\in[k]$ and $\tD \in [\btE]$.
	This is due to the fact that in this case, for $G[B_i \cup D_i]$ to be connected, $G[B_i]$ already needs to be connected.
	In \cite{BergougnouxKante2021}, the authors are able to go over each special case since they prove how to handle at most two connectivity constraints.
	In this proof, we handle $2^{k}$ of these special cases, by considering for every $S\subseteq [k]$, a subset $[\btE]_S$ of $[\btE]$ satisfying some properties including that for every $\tD \in [\btE]_S$ we have $D_i=\emptyset$ for every $i\in S$.
	For every $S\subseteq [k]$, we compute a set $\cB_{S}^\conn$ such that 
	$\best_{\Pi_\phi}^{\conn}(\cB, \tD)=\best_{\Pi_\phi}^{\conn}(\cB_S^\conn, \tD)$ for every $ \tD\in [\btE]_S$.
	Then, we compute $\reduce^\conn_{\btE}(\cB)$ as the union of the sets $\cB_S^\conn$ over all $S\subseteq [k]$.
	
	We guarantee that $\best_{\Pi_\phi}^{\conn}(\cB, \tD)=\best_{\Pi_\phi}^{\conn}(\reduce^\conn_{\btE}(\cB), \tD)$ for every $ \tD\in [\btE]$ by proving that (1)~if $ \tD$ is not in any $[\btE]_S$, then $\best_{\Pi_\phi}^{\conn}(\cB, \tD)=-\infty$, i.e., there is not $ \tB\in \cB$ such that $G\models \phi_\conn( \tB\cup  \tD)$ or (2)~if $ \tD\in [\btE]_S$, then $\best_{\Pi_\phi}^{\conn}(\cB, \tD)=\best_{\Pi_\phi}^{\conn}(\cB^\conn_S, \tD)$.
	
	For every $S\subseteq [k]$, we compute $\cB_S^\conn$ from a subset $\cB_S$ of $\cB$ satisfying some properties including that for every $\tB \in \cB_S$ and $i\in S$, $G[B_i]$ is connected.
	This property guarantees that every $ \tB\in\cB_S$, $ \tD\in [\btE]_S$ and $i\in S$, we have $G[B_i\cup D_i]=G[B_i]$ is connected. Thus, computing $\cB_S^\conn$ from $\cB_S$ takes care of the constraints $\conn(X_i)$ for $i\in S$.
	
	To deal with the constraints $\conn(X_i)$ for $i\in [k]\setminus S$, we use the rank-based approach.
	The definitions of $[\btE]_S$ and $\cB_S$ guarantee the following property: for every $ \tB\in\cB_S$, $ \tD\in [\btE]_S$, and $i\in [k]\setminus S$, either $B_i= \emptyset$ or every connected component of $G[B_i]$ (resp. $G[D_i]$) has at least one neighbor in $D_i$ (resp. $B_i$).
	Not only is this property necessary for $G[B_i\cup D_i]$ to be connected but as we will see later in this proof, 
	it is also important for the rank-based approach to work correctly.

	For every $S\subseteq [k]$, we define $[\btE]_{S}$ as the set of all $ \tD\in[\btE]$ such that for every $i\in [k]$:
	\begin{itemize}
		\item If $i\in S$, then $D_i=\emptyset$.
		\item If $i\notin S$, then $D_i\neq \emptyset$ and $G[D_i]$ is connected or for each connected component $C$ of $G[D_i]$, we have $N_1^1(\bB_i)\cap V(C)\neq\emptyset$.
	\end{itemize}
	For every $S\subseteq [k]$,  we define by $\cB_{S}$ the set of all $ \tB\in \cB$ such that for every $i\in [k]$:
	\begin{itemize}
		\item If $i\in S$, then $G[B_i]$ is connected,
		\item If $i\notin S$, then either $B_i=\emptyset$ or for every connected component $C$ of $G[B_i]$, we have $N_{1}^{1}(\bE_i)\cap V(C)\neq\emptyset$.
	\end{itemize}
	The following fact proves the soundness of our strategy.
	\begin{claim}\label{fact:bE_S}
		For every $ \tD\in [\btE]$ such that $ \tD\notin \bigcup_{S\subseteq [k]} [\btE]_S$, we have $\best_{\Pi_\phi}^{\conn}(\cB, \tD)=-\infty$.	
		For every $S\subseteq [k]$ and $ \tD\in [\btE]_S$, we have $\best_{\Pi_\phi}^{\conn}(\cB, \tD)=\best_{\Pi_\phi}^{\conn}(\cB_{S}, \tD)$.%
	\end{claim}
	\begin{claimproof}
		Observe that by definition, we have $d(\phi)\geq 1$ and $1\in R(\phi)$.
		Since $\btE\equi{\comp{A}}{\phi} \tD$ for every $ \tD\in[\btE]$ and $\btB\equi{A}{\phi}  \tB$, for every $ \tB\in \cB$, we have $N^1_1(D_i)\cap A=N^1_1(\bE_i)\cap A$ and $N^1_1(B_i)\cap \comp{A}=N^1_1(\bB_i)\cap \comp{A}$ for every $i\in[k]$.
	
		Let $ \tD\in [\btE]$.
		First, suppose that $ \tD\notin \bigcup_{S\subseteq [k]} [\btE]_S$.
		This implies that there exists $i\in[k]$ such that $D_i\neq \emptyset$, $G[D_i]$ is not connected and there exists a connected component $C$ of $G[D_i]$ such that $N_1^1(\bB_i)\cap V(C)=\emptyset$.
		For every $ \tB\in \cB$, since $N^1_1(B_i)\cap \comp{A}=N^1_1(\bB_i)\cap \comp{A}$, we deduce that $C$ has no neighbors in $B_i$. Thus, for every $ \tB\in \cB$, the graph $G[B_i\cup D_i]$ is not connected and $G\not\models \phi_\conn( \tB\cup \tD)$.
		Consequently, we have $\best_{\Pi_\phi}^{\conn}(\cB, \tD)=-\infty$.
		
		Now suppose that there exists $S \subseteq [k]$ such that $ \tD\in [\btE]_S$.
		We prove that $\best_{\Pi_\phi}^{\conn}(\cB, \tD)=\best_{\Pi_\phi}^{\conn}(\cB_{S}, \tD)$ by showing that $G\not\models \phi_\conn( \tB\cup \tD)$ for every $ \tB\notin \cB_S$.
		Let $ \tB\in \cB \setminus \cB_S$.
		By definition of $\cB_S$, this means that one of the following cases is true:
		\begin{itemize}
			\item There exists $i\in S$ such that $G[B_i]$ is not connected.
			Since $ \tD\in [\btE]_S$, we have $D_i=\emptyset$ and thus $G[B_i\cup D_i]=G[B_i]$ is not connected.
			
			\item There exists $i\in [k]\setminus S$ such that $B_i\neq \emptyset$ and there exists a connected component $C$ of $G[ B_i]$ such that $N_1^1(\bE_i)\cap V(C)=\emptyset$.
			As $N^1_1(D_i)\cap A=N^1_1(\bE_i)\cap A$, we deduce that $C$ has not neighbors in $B_i$ and thus $G[B_i\cup D_i]$ is not connected.
		\end{itemize}
		In both cases, $G[B_i\cup D_i]$ is not connected and $G\not\models \phi_\conn( \tB\cup \tD)$.
	\end{claimproof}
	Before we show how to compute the sets $\cB_S^{\conn}$, we need to prove some properties on the following Boolean variables.
	For every $i\in [k]\setminus S$, $ \tB\in \cB_S$, $ \tD\in [\btE]_S$ and $L,R\in \Rep{\comp{A}}{1}$, we define the Boolean variables $\cX[B_i,L,R]$ and $\cY[L,R,D_i]$ such that:
	\begin{align*}
		\cX[B_i,L,R] &= 1  \text{ iff } \exists  (B_L,B_R)\in \ccuts(B_i)  \text{ such that } N_1^1(R)\cap B_L=\emptyset \wedge N_1^1(L)\cap B_R=\emptyset.\\
		\cY[L,R,D_i] &= 1  \text{ iff } \exists  (D_L,D_R)\in \ccuts(D_i)  \text{ such that } v_{D_i}\in D_L \wedge D_L\equi{\comp{A}}{1,1} L \wedge D_R \equi{\comp{A}}{1,1} R.
	\end{align*}
	\begin{claim}\label{fact:Xi}
		Let $S\subseteq [k]$.
		For every $i\in [k]\setminus S$, $ \tB\in \cB_S$ and $ \tD\in [\btE]_S$, we have
		\[ \sum_{L,R\in \Rep{\comp{A}}{1,1}} \cX[B_i,L,R]\cdot \cY[L,R,D_i] = 2^{\abs{\cc(B_i\cup D_i)} -1}. \]
	\end{claim}
	\begin{claimproof}
		Let $i\in [k]\setminus S$, $ \tB\in \cB_S$ and $ \tD\in [\btE]_S$.
		This claim follows from these observations:
		\begin{itemize}
			\item If $\cX[B_i,L,R]\cdot \cY[L,R,D_i] = 1$, then there exists $(C_L,C_R)\in \ccuts(B_i\cup D_i)$ with  $v_{D_i}\in C_L$ and $C_I\cap \comp{A} \equi{\comp{A}}{1,1} I$ for each $I\in \{L,R\}$.
			Indeed, suppose that there exist $(B_L,B_R)\in \ccuts(B_i)$ and $(D_L,D_R)\in \ccuts(D_i)$ satisfying the required properties so that $\cX[B_i,L,R]=1$ and $\cY[L,R,D_i]=1$ respectively.
			From the definition of $\equi{\comp{A}}{1,1}$, we deduce that $N_1^1(D_L)\cap B_R =\emptyset$ and $N_1^1(D_R)\cap B_L$. We deduce that $(C_L,C_R)=(B_L\cup D_L, B_R\cup D_R)\in \ccuts(B_i\cup D_i)$. 
			The definition of $\cY[L,R,E_i]$ implies that $v_{D_i}\in C_L$ and $C_I\cap A \equi{\comp{A}}{1,1} I$ for each $I\in \{L,R\}$.

			\item For every $L,R\in \Rep{\comp{A}}{1,1}$, there is at most one $(C_L,C_R)\in\ccuts(B_i\cup D_i)$ such that $v_{D_i}\in C_L$ and $C_I\cap \comp{A} \equi{\comp{A}}{1,1} I$ for each $I\in \{L,R\}$ (we say that $(C_L,C_R)$ \emph{is associated with} $(L,R)$). 
			By definition of $[\btE]_S$ and $\cB_S$ and since $ \tB\equi{A}{1,1} \btB$, we know that either (A)~$B_i=\emptyset$ and $G[D_i]$ is connected or (B)~every connected component $C$ of $G[B_i]$ (resp. $D_i$) has a neighbor in $D_i$ (resp. $B_i$).
			If (A) is true, then $(D_i,\emptyset)$ is the only consistent cut of $B_i\cup D_i$ with $v_{D_i}$ on the left side.
			
			Suppose that (B) is true.
			Let $L,R\in \Rep{\comp{A}}{1,1}$ and $(C_L,C_R)\in\ccuts(B_i\cup D_i)$ associated with $(L,R)$.
			Every connected component $C$ of $G[B_i]$ is either included in $C_L$ or $C_R$.
			Suppose that $C\subseteq C_L$.
			Since $C$ has a neighbor in $D_i$ and $N_1^1(C_L)\cap C_R=\emptyset$, we deduce that $N_{1}^{1}(C_L)\cap C =\emptyset$.
			As  $C_L \cap \comp{A} \equi{\comp{A}}{1,1} L$, we have  $N_{1}^{1}(L)\cap C \neq \emptyset$.
			Thus for every $(C_L',C_R')\in\ccuts(B_i\cup D_i)$ associated with $(L,R)$, we have $C\subseteq C_L'$. 
			Because this holds for every $C\in \cc(B_i)$, we  deduce that $C_L\cap A = C_L' \cap A$ and $C_R\cap A = C_R'\cap A$.
			We conclude that $(C_L,C_R)$ is the only consistent cut of $B_i\cup D_i$ associated with $(L,R)$ because every connected component of $G[D_i]$ has a neighbor in $B_i$.

			\item The number of consistent cuts $(D_L,D_R)\in \ccuts(B_i\cup D_i)$ with $v_{D_i}\in D_L$ is $2^{\abs{\cc(B_i\cup D_i)} -1}$
			because every connected component not containing $v_{D_i}$ can be on both side of these cuts.
		\end{itemize}
	\end{claimproof}
	\begin{claim}\label{claim:connectivity}
		For every $S\subseteq [k]$, we can compute in time $O(\abs{\phi}\cdot \abs{\cB}\cdot \nec_{1,1}(\comp{A})^{O(k)} \cdot n^2)$, %
		a subset $\cB^\conn_S$ of $\cB$ of size at most $\nec_{1,1}(\comp{A})^{2k}$ such that for every $ \tD\in [\btE]_S$, we have $\best_{\Pi^\conn_\phi}(\cB, \tD)=\best_{\Pi^\conn_\phi}(\cB^\conn_S, \tD)$.
	\end{claim}
	\begin{claimproof}
		Let $S\subseteq [k]$.
		For every $ \tD\in [\btE]_S$ and $i\in [k]\setminus S$, we let $v_{D_i}$ be a fixed vertex of $D_i$ (such vertex exists because $D_i\neq \emptyset$ by definition of $[\btE]_S$).
		Let $\cM$ be the $(\cB_S,[\btE]_S)$-matrix such that $\cM[ \tB, \tD]=1$ iff $G\models \phi_\conn( \tB\cup \tD)$.
		The set $\cB^\conn_S$ we want to compute is a basis of maximum weight
		of the row space of $\cM$ over the binary field where the weight of a row $ \tB\in \cB_S$ is $\obj( \tB)$.
		However, we cannot compute this basis from $\cM$ as this latter matrix is too big.
		We overcome this by (1)~defining two matrices $\cC_A$ and $\cC_{\comp{A}}$, (2)~proving that $\cM=_2\cC_A\cdot \cC_{\comp{A}}$ where $=_2$ denotes the equality modulo 2 and (3)~computing $\cC_A$ and $\cB^\conn_S$ as a row basis of maximum weight of the row space of $\cC_A$. The claimed running time and the upper bound on the size of $\cB_S^\conn$ follows from the size of $\cC_A$.
		Let $\cL\cR$ be the set of all pairs $( L, R)$ with $ L, R\in (\Rep{\comp{A}}{1})^k$.
		Let $\cC_A$ and $\cC_{\comp{A}}$ be, respectively the $(\cB, \cL\cR)$-matrix and the $(\cL\cR,[\btE]_S)$-matrix such that:
		\begin{equation*}
			\cC_A[ B,( L, R)]=\prod_{i\in [k]\setminus S} \cX[B_i,L_i,R_i]   \text{ and }	\cC_{\comp{A}}[( L, R), \tD]=\prod_{i\in [k]\setminus S} \cY[L_i,R_i,D_i].	
		\end{equation*}
		For every $ \tB\in \cB_S$ and $ \tD\in [\btE]_S$, we have 
		$(\cC_A\cdot \cC_{\comp{A}})[ \tB, \tD]= \prod_{i\in [k]\setminus S} 2^{\abs{\cc(B_i\cup D_i)} -1}$,
		which can be derived using \Cref{fact:Xi} as follows.
		Let $[k] \setminus S = \{i_1, \ldots, i_t\}$. For $i \in [k] \setminus S$, 
		we use the shorthands $x_i = \cX[B_i, L_i, R_i]$ and $y_i = \cY[L_i,R_i,D_i]$. By definition of $\cL\cR$, we deduce the following:	
		\begin{align*}
			(\cC_A\cdot \cC_{\comp{A}})[ B, \tD] 
				&= \sum\nolimits_{( L, R) \in \cL\cR} \prod\nolimits_{i \in [k] \setminus S} x_iy_i 
					= \sum\nolimits_{L_{i_j}, R_{i_j} \in \Rep{\comp{A}}{1,1}, j \in [t]} \prod\nolimits_{h \in [t]} x_{i_h}y_{i_h} \\
				&= \sum\nolimits_{L_{i_1}, R_{i_1} \in \Rep{\comp{A}}{1,1}}\sum\nolimits_{L_{i_j}, R_{i_j} \in \Rep{\comp{A}}{1,1}, j \in \{2,\ldots,t\}}x_{i_1}y_{i_1}\prod\nolimits_{h \in \{2,\ldots,t\}}x_{i_h}y_{i_h} \\
				&= \sum\nolimits_{L_{i_1}, R_{i_1} \in \Rep{\comp{A}}{1,1}} x_{i_1}y_{i_1}\cdot \left(\sum\nolimits_{L_{i_j}, R_{i_j} \in \Rep{\comp{A}}{1,1}, j \in \{2,\ldots,t\}}\prod\nolimits_{h \in \{2,\ldots,t\}}x_{i_h}y_{i_h}\right) \\
				&= 2^{\card{\cc(B_{i_1}\cup D_{i_1})} -1} \cdot \left(\sum\nolimits_{L_{i_j}, R_{i_j} \in \Rep{\comp{A}}{1,1}, j \in \{2,\ldots,t\}}\prod\nolimits_{h \in \{2,\ldots,t\}}x_{i_h}y_{i_h}\right) = \ldots \\
				&= \prod\nolimits_{j \in [t]} 2^{\card{\cc(B_{i_j}\cup D_{i_j})} -1} = \prod\nolimits_{i\in [k]\setminus S} 2^{\abs{\cc(B_i\cup D_i)} -1}.
		\end{align*}
		Observe that $(\cC_A\cdot \cC_{\comp{A}})[ \tB, \tD]$ is odd iff  $G[B_i\cup D_i]$ is connected for every $i\in [k]\setminus S$. 
		Since, for every $i\in S$, $D_i=\emptyset$ and $G[B_i]$ is connected, we deduce that $(\cC_A\cdot \cC_{\comp{A}})[ \tB, \tD]=_2 \cM[ \tB, \tD]$.
		
		\medskip
		
		Let $\cB^\conn_S \subseteq \cB_S$ be a basis of maximum weight of the row space of $\cC_A$ over the binary field (the weight of a row $ \tB$ being $\obj( \tB)$).
		We claim that for every $ \tD\in [\btE]_S$, we have $\best_{\Pi^\conn_\phi}(\cB, \tD)=\best_{\Pi^\conn_\phi}(\cB^\conn_S, \tD)$.
		Let $ \tD \in [\btE_S]$.
		If $\best_{\Pi^\conn_\phi}(\cB, \tD)=-\infty$, then as $\cB^\conn_S\subseteq \cB_S \subseteq \cB$, we have $\best_{\Pi^\conn_\phi}(\cB^\conn_S, \tD)= -\infty$.
		
		Assume now that $\best_{\Pi^\conn_\phi}(\cB, \tD)\neq -\infty$.
		By \Cref{fact:bE_S}, there exists $ \tB\in \cB_S$ such that $G\models \phi_\conn( \tB\cup  \tD)$ and $\obj( \tB)=\best_{\Pi^\conn_\phi}(\cB, \tD)$.
		As $\cB^\conn_S$ is a row basis of $\cC_A$, there exists a subset $\widehat{\cB}\subseteq \cB^\conn_S$ that generates the row of $ \tB$ in $\cC_A$, i.e. for every $( L, R)\in \cL\cR$, we have $\cC_A[ \tB, ( L, R)]=_2 \sum_{ \tB'\in \widehat{\cB}} \cC_A[ \tB', ( L, R)]$.
		Thus, we have the following equality
		\begin{align*}
			\cM[ B, \tD]=_2  (\cC_A\cdot \cC_{\comp{A}})[ B, \tD] =_2 &\sum_{( L, R)\in \cL\cR} \cC_A[ B,( L, R)]\cdot \cC_{\comp{A}}[( L, R), \tD]\\
			=_2& \sum_{( L, R)\in \cL\cR} \left( \sum_{ B'\in \widehat{\cB}} \cC_A[ B', ( L, R)]  \right) \cdot \cC_{\comp{A}}[( L, R), \tD]\\
			=_2 &\sum_{ B'\in \widehat{\cB}} \left( \sum_{( L, R)\in \cL\cR} \cC_A[ B', ( L, R)]  \right) \cdot \cC_{\comp{A}}[( L, R), \tD]\\
			=_2 &  \sum_{ B'\in \widehat{\cB}}(\cC_A\cdot \cC_{\comp{A}})[ B', \tD] =_2 \sum_{ B'\in \widehat{\cB}}\cM[ B', \tD].
		\end{align*}
		As $G\models \phi_\conn( \tB\cup  \tD)$, we have $\cM[ \tB, \tD]=1$ and from the above equality, we deduce that $\sum_{ \tB'\in \widehat{\cB}}\cM[ \tB', \tD]$ is odd.
		Consequently, there is at least one $ \tB^\star \in \widehat{\cB}$ such that $\cM[ \tB^\star, \tD]=1$.
		Let $ \tB^\star \in \cB^\conn_S$ such that $\cM[ \tB^\star, \tD]=1$ and $\obj( \tB^\star)$ is maximum.
		Since $\widehat{\cB}$ generates the row of $ \tB$ in $\cC_A$ and $ \tB^\star\in \widehat{\cB}$, we deduce that $(\widehat{\cB}\setminus \{ \tB^\star\})\cup \{ \tB\}$ generates the row of $ \tB^\star$ in $\cC_A$.
		Thus, $(\cB^\conn_S \setminus \{ \tB^\star\}) \cup \{  \tB\}$ is also a basis of the row space of $\cC_A$.
		Since $\cB^\conn_S$ is a basis of maximum weight, we deduce that $\obj( \tB)\leq \obj( \tB^\star)$.
		Because $\cM[ \tB^\star, \tD]=1$, we know that $G\models \phi_\conn( \tB^\star\cup  \tD)$ and thus $\obj( \tB^\star)\leq \best_{\Pi^\conn_\phi}(\cB^\conn_S, \tD)$.
		As $\obj( \tB)=\best_{\Pi^\conn_\phi}(\cB, \tD)$ and $\cB^\conn_S \subseteq \cB$, we conclude that $\best_{\Pi^\conn_\phi}(\cB, \tD)=\best_{\Pi^\conn_\phi}(\cB^\conn_S, \tD)$
		
		The size of $\cB^\conn_S$ is at most the size of $\abs{\cL\cR}=\nec_{1,1}(A)^{2k}$ because $\cB^\conn_S$ is a basis of the row space and thus its size is at most the number of columns of $\cC_A$.
		
		It remains to prove the running time. 
		First, from its definition, computing $\cB_S$ can be done in time $O(\abs{\cB}\cdot n^2)$. Computing $\Rep{\comp{A}}{1,1}$ can be done in time $O(\nec_{1,1}(\comp{A}) \cdot n^2 \cdot \log(\nec_{1,1}(\comp{A})))$ by \Cref{lem:computerep}.
		
		As proved in \Cref{lem:reducelogic}, we can compute $\obj(\tB)$ for all $\tB \in \cB$ in time $O(|\cB| \cdot n \cdot |\phi|)$.
		After computing these values, we can decide in constant time whether $\obj(\tC)<\obj(\tB)$ for arbitrary $\tB,\tC \in \cB$.
		
		Then, for every $i\in[k]\setminus S$, $ B\in \cB_S$ and $L,R\in \Rep{\comp{A}}{1,1}$, we can compute the boolean variable $\cX[B_i,L,R]$ in time $O(n^2)$ since $\cX[B_i,L,R]=1$ iff for every connected component $C$ of $G[B_i]$, either $N_1^1(L)\cap C= \emptyset$ or $N_1^1(R)\cap C=\emptyset$.
		As $k\leq \abs{\phi}$, we deduce that $\cC_A$ can be computed in time $O(\abs{\phi}\cdot \abs{\cB}\cdot \nec_{1,1}(\comp{A})^{2k} \cdot n^{2})$.
		By \Cref{lem:basis:matrix}, we can compute $\cB^\conn_S$ from $\cC_A$ in time $O(\abs{\cB_S}\cdot \abs{\cL\cR}^{O(1)})$.
		We deduce the claimed running time from the fact that $\abs{\cL\cR} = \nec_{1,1}(\comp{A})^{2k}$.
	\end{claimproof}
	For every $S\subseteq [k]$, we compute $\cB_{S}^\conn$ with \Cref{claim:connectivity}.
	We set $\reduce^\conn_{\btE}(\cB)=\bigcup_{S\subseteq [k]}\cB_S^\conn$.
	By \Cref{claim:connectivity}, we have $\abs{\reduce^\conn_{\btE}(\cB)} \leq \nec_{1,1}(\comp{A})^{2k} \cdot 2^{k}$ and $\reduce^\conn_{\btE}(\cB)$ can be computed in time 
	$O(\abs{\phi}\cdot \abs{\cB}\cdot \nec_{1,1}(\comp{A})^{O(k)} \cdot  2^{k} \cdot n^2)$
	Since $1\in d(\phi)$ and $1\in R(\phi)$, we have $\nec_{1,1}(\comp{A}) \leq \snec_{1,1}(A)$.
	As $k\leq \var(\phi) \leq \abs{\phi}$, we deduce that $\abs{\reduce^\conn_{\btE}(\cB)}\leq \nec_{1,1}(\comp{A})^{2\abs{\var(\phi)}} \cdot 2^{\abs{\phi}}$ and that $\reduce^\conn_{\btE}(\cB)$ can be computed in time 
	\[ 	O(\abs{\cB}\cdot \snec_{1,1}(A)^{O(\abs{\var(\phi)})} \cdot  2^{\abs{\phi}} \cdot n^2). \]
	
	It remains to prove that for every $ \tD\in [\btE]$, we have $\best_{\Pi_\phi}^{\conn}(\cB, \tD)=\best_{\Pi_\phi}^{\conn}(\reduce^\conn_{\btE}(\cB), \tD)$.
	Let $ \tD \in [\btE]$.
	If for every $S\subseteq [k]$, we have $ \tD\notin [\btE]_S$, then $\best_{\Pi_\phi}^{\conn}(\cB, \tD)=-\infty$ by \Cref{fact:bE_S} and since $\reduce^\conn_{\btE}(\cB) \subseteq \cB$ by construction, we deduce that $\best_{\Pi_\phi}^{\conn}(\reduce^\conn_{\btE}(\cB), \tD)$ is also equal to $-\infty$.
	
	If there exists $S\subseteq [k]$ such that $ \tD\in [\bE]_S$, then by \Cref{claim:connectivity} we have $\best_{\Pi_\phi}^{\conn}(\cB, \tD)=\best_{\Pi_\phi}^{\conn}(\reduce^\conn_{\btE}(\cB), \tD)$.
\end{proof}

To deal with the acyclicity constraints, we need to following definition and lemmata.
The following properties of $\equi{A}{n,1}$ are proved in \cite[Lemmata 3.5 and 3.6]{BergougnouxKante2021}.
For every graph $G$ and pair $(\tB,D)$ of disjoint subsets of $V(G)$, we denote by $E(B,D)$ the set of edges with one endpoint in $B$ and the other in $D$.

\begin{lemma}[\cite{BergougnouxKante2021}]\label{lem:n:neighbor:eq}
	Let $G$ be an $n$-vertex graph, $A\subseteq V(G)$.
	The number of $(n,1)$-neighbor equivalence class over the subsets of $A$ of size at most $t$ is upper bounded by 
	$2^{t\cdot \rw(A)}$. For every $B,C\subseteq A$ and $D\subseteq \comp{A}$, if $B\equi{A}{n,1} C$, then $\abs{E(B,D)}=\abs{E(C,D)}$.%
\end{lemma}

In the next proof, we use the following lemma due to Bergougnoux et al.~\cite{BergougnouxPapadopoulosTelle2020}.
\begin{lemma}[Lemma 4 in~\cite{BergougnouxPapadopoulosTelle2020}]\label{lem4:bpt20}
	Let $G$ be a graph and $X$ and $Y$ be two disjoint subsets of $V(G)$. If $G[X \cup Y]$ is a forest, then the number of vertices in $X$ that have two neighbors in $Y$ is at most $2w$, where $w$ denotes the size of a maximum induced matching in $G[X \cup Y]$.
\end{lemma}

\ACproperties*
\begin{proof}
	Let $ B\subseteq A$ and $ D\subseteq \comp{A}$ such that $ D\equi{\comp{A}}{2,1} \bE$ and $G[ B\cup D]$ is a tree. 	Obviously, $G[ B]$ is a forest.
	Since $ D\equi{\comp{A}}{2,1} \bE$, the vertices in $ B$ with exactly one (resp. at least two) neighbors in $ D$ are the same as those with exactly one (resp. at least two) neighbors in $\bE$.
	
	Thanks to \cref{lem4:bpt20}, we know that $\abs{ B^{2+}_{\bE}}$ is at most $2w$ where $w$ is the size of a maximum induced matching in the bipartite graph $G[ B, D]$. 
	As $G[ B, D]$ is an induced subgraph of $G[A,\comp{A}]$, we deduce that  $\abs{ B^{2+}_{\bE}}\leq 2\mim(A)$.
	
	Suppose that two vertices $u,v\in  B^{1}_{\bE}\cup B^{2+}_{\bE}$ have the same neighborhood in $\comp{A}$.
	Observe that $u$ and $v$ have at least one common neighbor in $ D$. Since $G[ B\cup  D]$ is a tree, we deduce that $u$ and $v$ are not connected in $G[ B]$ and $u,v\notin  B^{2+}_{\bE}$.
	
	Assume towards a contradiction that there exists a cycle $(C_1,D_1,C_2,\dots,C_\ell,D_\ell)$ in $H_{\bE, B}$ with $C_1,\dots,C_\ell\in \cc( B)$ and $D_1,\dots,D_\ell\in \{N_1^1(v)\cap \comp{A}\mid v\in  B^{1}_{\bE}\}$.
	For every $i\in [\ell]$, $C_i$ contains two vertices $s_{C_i}, t_{C_i}\in  B^{1}_{\bE}$ such that $N_1^1(s_{C_i})\cap\comp{A}=D_{i-1}$ and $N_1^1(t_{C_i})\cap \comp{A}= D_{i}$ (we consider that $D_0=D_\ell$).
	By definition of $ B^{1}_{\bE}$, for every $i\in [\ell]$, $s_{C_i}$ and $t_{C_{i+1}}$ have a common neighbor in $ D$ (we consider that $C_{\ell+1}=C_1$).
	Since $s_{C_i}$ is connected to $t_{C_i}$ in $G[ B]$ for every $i\in [\ell]$, we deduce that $G[ B\cup D]$ contains a cycle, yielding a contradiction.
	
	If $\bE\neq \emptyset$, then $ D\neq \emptyset$ and every connected component of $G[ B]$ must have a neighbor in $ D$ (if $ B=\emptyset$, this is true by vacuity), otherwise $G[ B\cup D]$ would not be connected.
	We conclude that$  B$ satisfies all the properties to be in $\forest(A,\bE)$.	
\end{proof}

\acyeq*
\begin{proof}
	We define $\sim^\acy_\bE$ such that $ B\sim^\acy_\bE  C$ if 
	\begin{itemize}
		\item $ B^{2+}_{\bE} \equi{A}{n,1} C^{2+}_{\bE}$, and
		
		\item $\abs{ B^1_\bE} - \abs{\cc( B)} = \abs{ C^1_\bE} - \abs{\cc( C)}$.
	\end{itemize}
	\myparagraph{Proof of Property~\ref{lem:acy:eq:1}} Let $ B, C\in \forest(A,\bE)$ and $ D\subseteq \comp{A}$ such that  $ B\sim^\acy_\bE  C$ and $ D\equi{\comp{A}}{2,1} \bE$. 
	Assume that $G[ B\cup D]$ is a tree. This implies that $\abs{E( B, D)}= \abs{\cc( B)} + \abs{\cc( D)} -1$.
	Since $ D\equi{\comp{A}}{2,1} \bE$, for each $ D\in\{ B, C\}$, the vertices in $ D$ with exactly one neighbor (resp. at least two neighbors) in $ D$ are the vertices in $ D$ with exactly one neighbor (resp. at least two neighbors) in $ D$.
	Consequently, for each $ D\in\{ B, C\}$ we have 
	\begin{equation}\label{eq:acy}
		\abs{E( D, D)}= \abs{E( D^{2+}_{\bE}, D)} + \abs{E( D^1_\bE, D)} = \abs{E( D^{2+}_{\bE}, D)} + \abs{ D^1_\bE}.
	\end{equation}
	As $ B^{2+}_{\bE} \equi{A}{n,1} C^{2+}_{\bE}$, by \Cref{lem:n:neighbor:eq},  we have $\abs{E( B^{2+}_{\bE}, D)}=\abs{E( C^{2+}_{\bE}, D)}$.
	Since $\abs{B^1_\bE} - \abs{\cc( B)} = \abs{ C^1_\bE} - \abs{\cc( C)}$ and $\abs{E( B, D)}= \abs{\cc( B)} + \abs{\cc( D)} -1$, we deduce from \Cref{eq:acy} that $\abs{E( C, D)}= \abs{\cc( C)} + \abs{\cc( D)} -1$.
	Thus, if $G[ C\cup D]$ is connected, then $G[ C\cup D]$ is necessarily a tree.
	We conclude that Property~\ref{lem:acy:eq:1} is true.

	\myparagraph{Proof of Property~\ref{lem:acy:eq:2}} 
	For every $ B\in \forest(A,\bE)$, we can compute $ B^{1}_{\bE}$ and $ B^{2+}_{\bE}$ in time $O(n^2)$ by computing for each $v\in  B$ the number of neighbors of $v$ in $\bE$.
	Hence, $\abs{ B^1_\bE} - \abs{\cc( B)}$ can also be computed in time $O(n^2)$.
	Given $ B, C\in \forest(A,\bE)$, we can check whether $ B^{2+}_{\bE} \equi{A}{n,1} C^{2+}_{\bE}$ in time $O(n^2)$ by checking whether $\abs{N_1^1(v)\cap  B^{2+}_{\bE}}=\abs{N_1^1(v)\cap  C^{2+}_\bE}$ for every $v\in \comp{A}$.
	We conclude that Property~\ref{lem:acy:eq:2} is true.

	\myparagraph{Proof of Property~\ref{lem:acy:eq:3}} 
	Let $S$ be the set containing the integers $\abs{ B^1_\bE} - \abs{\cc( B)}$ for every $ B \in \forest(A,\bE)$ and let $N_2$ be the number of equivalence classes of $\equi{A}{n,1}$ over $\{ B^{2+}_{\bE} \mid  B\in \forest(A,\bE)\}$.
	Observe that the number of equivalence classes of $\sim^\acy_\bE$ is at most $(\max S-\min S)\cdot N_2$.
	
	First, we prove some upper bounds on $\max S-\min S$.
	Let $ B\in \forest(A,\bE)$. 

	Condition~\ref{item:no:twin} implies that for every pair of distinct vertices $(u,v)\in  B^{1}_{\bE}$, $u$ and $v$ are not connected in $G[ B]$.
	Thus, for every edge in $H_{\bE,B}$ (the bipartite graph defined in \Cref{def:forest:acy}) between $C\in \cc(B)$ and $F\in \{N(v)\cap \comp{A} \mid v\in B^{1}_{\bE}\}$, there is exactly one vertex $v \in C\cap B^1_{\bE}$ such that $N(v)\cap \comp{A}=F$.
	Since every vertex in $B^1_{\bE}$ is in exactly one connected component of $G[B]$, we deduce that $\abs{ B^1_\bE}= \abs{E(H_{\bE, B})}$.

	Let $\cC^2$ be the set of all components in $\cc( B)$ that intersect $ B^{2+}_{\bE}$ and $\cC^1$ be the set of those not intersecting $ B^{2+}_{\bE}$ and containing at least two vertices in $ B^{1}_{\bE}$.
	As every connected component of $G[B]$ intersects $ B^{1}_{\bE}\cup B^{2+}_{\bE}$, we know that
	the connected components $C\in \cc(B)$ that are not in $\cC^1\cup\cC^2$ contain exactly one vertex $v_C \in   B^{1}_{\bE}$.
	Thus, these connected components $C$ and $v_C$ ``cancel each other out'' in  $\abs{ B^1_\bE} - \abs{\cc( B)}$.
	We deduce that $\abs{ B^1_\bE} - \abs{\cc( B)}$ equals the number of edges incident in $H_ B$ to $\cC^1\cup\cC^2$ minus $\abs{\cC^1\cup\cC^2}$.
	
	By definition of module-width, we have $\abs{\{N_1^1(v)\cap \comp{A}\mid v\in  A\}} \leq \mw(A)$. 
	Since the vertices in $ B^{2+}_{\bE}$ have pairwise different neighborhoods in $\comp{A}$, we have that $\abs{B^{2+}_{\bE}}\leq \mw(A)$. Since every component in $\cC^2$ at least one vertex in $B_{\bE}^{2+}$, we deduce that $\abs{\cC^2}\leq \mw(A)$.
	As the degree of the components
	in $\cC^1$ in $H_B$ is at least 2 and $H_ B$ is a forest, it follows that $\abs{\cC^1}\leq \frac{1}{2}\mw(A)$. 
	Because $H_{\bE, B}$ is a forest, the number of edges in $H_{\bE, B}$ incident to $\cC^1\cup\cC^2$ is at most 
	$\abs{\{N_1^1(v)\cap \comp{A}\mid v\in  B^{1}_{\bE}\}} + \abs{ \cC_1 \cup \cC_2} -1 \leq \frac{5}{2}\mw(A)-1$. 
	Thus, we have 
	\[ \abs{ B^1_\bE} - \abs{\cc( B)} \leq \max S \leq \frac{5}{2}\mw(A)-1. \]
	Moreover, since the number of edges in $H_{\bE, B}$ incident to $\cC^1$ is at least $\abs{\cC^1}$, we deduce that 
	\[ 	\abs{ B^1_\bE} - \abs{\cc( B)} \geq \min S \geq - \abs{\cC^2} \geq - \mw(A). \]	
	Hence, we have $\max S - \min S \leq  \frac{7}{2}\mw(A)$.
	By definition of module-width, $\mw(A)$ is upper bounded by $n$ and $\nec_{1,1}(A)$.
	By \Cref{lem:necd1}, we know that $\nec_{1,1}(A)$ is upper bounded by $2^{2\mmw(A)}$ and $2^{\rw(A)^{2}}$.
	We conclude that $(\max S - \min S)$ is upper bounded by $4\mw(A)$, $2^{2\mmw(A)+1}$, $2^{\rw(A)^{2}+1}$ and $2n$.
	
	\newcommand{\signature}{\mathsf{signature}}
	
	\noindent
	Property~\ref{lem:acy:eq:3} is deduced from the following upper bounds on $N_2$.
	\begin{itemize}
		\item For every $ B\in \forest(A,\bE)$, the size of $ B^{2+}_{\bE}$ is bounded $2\mim(A)$. Thus, we have $N_2\leq n^{2\mim(A)}$. As $\mim(A)\leq \rw(A)$ \cite{VatshelleThesis}, by \Cref{lem:n:neighbor:eq}, we deduce that $N_2 \leq 2^{2\rw(A)^{2}}$.
		
		\item For every $ B\in \forest(A,\bE)$, we define $\signature( B)=\{N_1^1(v)\cap \comp{A}\mid v\in  B^{2+}_{\bE}\}$.
		Since for every $ B\in\forest(A,\bE)$, the neighborhoods in $\comp{A}$ of the vertices in $ B^{2+}_{\bE}$ are pairwise distinct, we deduce that for every $ B, C\in \forest(A,\bE)$, if $\signature( B)=\signature( C)$, then $ B^{2+}_{\bE} \equi{A}{n,1}  C^{2+}_{\bE}$.
		It follows that $N_2\leq \abs{\{\signature( B)\mid  B\in \forest(A,\bE)\}}$.
		Since we have $\mw(A)=\abs{\{ N_1^1(v)\cap \comp{A} \mid v\in A\}}$, we conclude that $N_2\leq 2^{\mw(A)}$.
	\end{itemize}	
\end{proof}

\lemACclause*
\begin{proof}
	Let $\phi(X_1,\dots,X_k)$ be a  \ac-clause, $G$ be a $k$-weighted graph and $\cL$ be rooted layout of $G$.
	We assume w.l.o.g. that $\acy(\phi) \setminus\conn(\phi)$ contains the first $\ell$ variables $X_1,\dots,X_\ell$ of $\var(\phi)$.
	To deal with the constraints $\acy(X_1),\dots,\acy(X_\ell)$, we reduce solving our problem on $G$ with $\cL$ to a solving a $k+\ell$-problem $\Pi$ on a supergraph $H$ of $G$ with a rooted layout $\cL_H$ based on $\cL$.
	Then, we show that we can solve $\Pi$ on $H$ by modifying the reduce routine of \Cref{lem:reduce:A&C} and applying \Cref{thm:genericAlgo} on $\cL_H$.
	
	 We define $\phi'$ the \ac-clause obtained from $\phi$ by removing the subformula $\acy(X_i)$ for every $i\in [\ell]$ and we define the \ac-clause $\omega$ as
	 \[ \omega \equiv \bigwedge_{1\in [\ell]} \left(\acy(X_i^\star) \wedge \conn(X_i^\star) \right)  \]
	 where $(X_1^\star,\dots,X_\ell^\star)$ are new variables. 
	 
	 To simplify the proof, we assume w.l.o.g. that $\var(\omega)=(X_1,\dots,X_k,X_1^\star,\dots,X_\ell^\star)$ and $\var(\phi')=(X_1,\dots,X_k)$ (we can always add $X=X$ to $\phi'$ or $\omega$ if $X\in \var(\phi)$ does not appear in these formulas).
	 Since, we consider that $(X_1^\star,\dots,X_\ell^\star)$ are the last $\ell$ variables in $\var(\omega)$, for every graph $G'$ and $\tB\in \Pset(V(G'))^{k+\ell}$, $\tB$ assigned $B_i$ to $X_i$ for every $i\in[k]$ and $\tB$ assigned $B_{i+k}$ to $X_i^\star$ for every $i\in [\ell]$.
	 
	 We construct a $k+\ell$-weighted graph $H$ from $G$ by adding, for every vertex $v\in V(G)$, a new vertex $v^\star$ adjacent to $v$  and a new vertex $u_\conn$ whose neighborhood in $H$ is $\{v^\star \mid v\in V(G)\}$.
	We construct the weight function $w_H$ of $H$ from the weight function of $w_G$ such that variables $X^\star_{1},\dots,X^\star_{\ell}$ and the vertices in $V(H)\setminus V(G)$ do not contribute to the weight of a tuple.
	For every $S\subseteq [k+\ell]$, we set $w(v,S) = w(v,S \cap [k])$ if $v\in V(G)$ and $w(v,S)=0$ otherwise (if $v\in V(H)\setminus V(G)$).

	 We obtain $\cL_H=(T_H,\delta_H)$ from $\cL=(T,\delta)$ as follows: for every leaf $t$ with $\delta(t)=v \in V(G)$, we create two new node $a_t$ and $a_t^\star$ adjacent to $t$ with $\delta_H(a_t)=v$ and $\delta_H(a_t^\star)=v^\star$.
	 The root of $T_H$ is a new node $r^\star$ adjacent to the root of $T$ and a new node $a_r^\star$ with $\delta_H(a_r^\star)=u_\conn$.

	 The construction of $H$ and $\cL_H$ was used in \cite[Theorem 6.11]{BergougnouxKante2021} where the following upper bounds are proved.
	 
	 \begin{fact}[\cite{BergougnouxKante2021}]\label{fact:upper:cL:H}
	 	We have $\snec_{2, 1}(\cL_H)\leq 3 \snec_{2, 1}(\cL)$ and for every $\f\in \{\mw,\rw,\mim\}$, we have $\f(\cL_H)\leq \f(\cL)+1$.
	 \end{fact} 
	
	We define the $k+\ell$-problem $\Pi$ such that $\Pi(H)$ contains all the tuples $\tB\in \Pset(V(H))^{k+\ell}$ such that:
	\begin{itemize}
		\item $B_1,\dots,B_k\subseteq V(G)$.
		\item For every $i\in [\ell]$, we have $B_i\subseteq B_{i+k}$,
		\item $G\models \phi'((B_1,\dots,B_k))$ and $H\models \omega(\tB)$.
	\end{itemize}
	Observe that for every $\tB\in\Pi(H)$ and $i\in [\ell]$, we have $B_i\subseteq B_{i+k}$ and $G[B_{i+k}]$ is a tree (because $H\models \omega(\tB)$). 
	Moreover, for every $(B_1,\dots,B_{k})\in \Pi_\phi(G)$ and every $i\in [\ell]$, the forest $G[B_i]$ can be extended into a tree in $H$ by adding $u_\conn$ and a vertex $v^\star$ with $v\in C$ for each connected component of $G[B_i]$. 
	From this observation, we deduce the following fact.
	
	\begin{fact}\label{fact:Pi:H}
		For all $\tB\in \Pi(H)$ holds $(B_1,\dots,B_{k})\in \Pi_\phi(G)$ and $\obj_H(\tB)=\obj_G((B_1,\dots,B_k))$.
		For all $\tB=(B_1,\dots,B_k)\in \Pi_\phi(G)$, there exist $B_{k+1},\dots,B_{k+\ell}\subseteq V(H)$ such that $(B_1,\dots,B_{k+\ell})\in \Pi(H)$.
	\end{fact}
 	
 	Let $x$ be a node of $T_H$, $A=V_x^{\cL_H}$, $A_G= A \cap V(G)$ and $\cB\subseteq \Pset(A)^{k+\ell}$.
 	With some modification of the reduce routine of \Cref{lem:reduce:A&C}, we prove that we can compute a set $\reduce(\cB)$ 	in time
	\begin{equation}\label{eq:runtime:acy}
		O(\abs{\cB}\cdot 2^{O(\abs{\phi})}  \cdot \snec^{G}_{\phi}(A_G)^{O(k)}\cdot \sfN_\acy(G,A_G)^{O(k)} \cdot n^2)
	\end{equation} 
 	 such that $\reduce(\cB)$ $\Pi$-represents $\cB$  and the size of $\reduce(\cB)$ is at most
 	\begin{equation}\label{eq:sizebound:acy}
 		\snec^G_{\phi}(A_G)^{6k} \cdot \sfN_\acy(G,A_G)^{2k} \cdot 2^{2\abs{\phi}} \cdot  \prod_{1\leq i\leq p} (m_i+2).
 	\end{equation}
 	By \Cref{thm:genericAlgo} and the upper bounds of \Cref{fact:upper:cL:H}, this implies that we can solve $\Pi$ on $H$ and thus $\Pi_\phi$ on $G$ in the claimed run time.
 	
	Intuitively, we main modification on the reduce routine of \Cref{lem:reduce:A&C} is that we use the relations $\equi{G,A_G}{\phi}$ and $\equi{G, \comp{A_G}}{\phi}$ for the variables $X_1,\dots,X_k$ and the relations $\equi{H,A}{2, 1}$ and $\equi{H,\comp A}{2, 1}$ for the variables $X_1^\star,\dots,X_\ell^\star$.
 	Formally, we modify the reduce routine of \Cref{lem:reduce:A&C} as follows:
 	\begin{itemize}
 		\item We start by removing from $\cB$ all the tuples $\tB$ such that $B_i\not\subseteq V(G)$ for some $i\in [k]$ or $B_i\not\subseteq B_{i+k}$ for some $i\in [\ell]$.
 		\item Instead of considering the tuples $\btE\in(\Rep{H,\comp A}{\phi})^{k+\ell}$, we consider the tuples $\btE=(\bE_1,\dots,\bE_{k+\ell})$ such that $\bE_1,\dots,\bE_k\in\Rep{G,\comp{A_G}}{\phi}$ and $\bE_{k+1},\dots,\bE_{k+\ell}\in \Rep{H,\comp{A}}{2, 1}$.
 		Let $\cR$ be the set of all the tuples $\btE$ satisfying these conditions.
 		
 		\item Given $\btE\in \cR$ and  $\tB,\tC\in \cB$, we modify the Condition~\ref{item:equiphi} of the definition of $\bowtie_\btE$ from \Cref{lem:equivalenceformula}. Instead of requiring that $\tB\equi{H,A}{\phi} \tC$, we require that
		$(B_1,\dots,B_k)\equi{G,A_G}{\phi} (C_1,\dots,C_k)$ and 
 			$(B_{k+1},\dots,B_{k+\ell}) \equi{H,A}{2, 1} (C_{k+1},\dots,C_{k+\ell})$.
 			
		\item We use the equivalence relation $\sim^\acy_{\bE}$ given by \Cref{lem:acy:eq} on $H$. Note that the number of equivalence classes of $\sim^\acy_{\bE}$ is upper bounded by $\sfN_\acy(H,A)$.
 	\end{itemize}
 
	We deduce that $\reduce(\cB)$ $\Pi$-represents $\cB$ from the fact that each set variable $X_i^\star$ appears only the subformulas $\conn(X_i^\star)$ and $\acy(X_i^\star)$ of $\omega$ and the tools we used to handle these contraints only use the properties of $\equi{H,A}{1,1}$ (\Cref{lem:rankbased:approach}) and $\equi{H,A}{2, 1}$ (\Cref{def:forest:acy}, \Cref{lem:A&C:properties,lem:acy:eq}).
 	From these modifications, we deduce that we can compute $\reduce(\cB)$ in time
	\begin{align*}
		O(2^{\abs{\phi}+\abs{\omega}}  \cdot \snec^G_{\phi}(A_G)^{O(k)} \cdot \snec^H_{2, 1}(A)^{O(\ell)} \cdot  \sfN_\acy(G,A_G)^{k} 
		 \cdot & \sfN_\acy(H,A)^{\ell} \\ & \cdot n^2)
	\end{align*}
 	and  the size of $\reduce(\cB)$ is at most
	\[
	\snec^G_{\phi}(A_G)^{3k}\cdot \snec^H_{2, 1}(A)^{3 \ell }\cdot \sfN_\acy(G,A_G)^{k} \cdot \sfN_\acy(H,A)^{\ell } \cdot 2^{2\abs{\phi}+\abs{\omega}} \cdot  \prod_{1\leq i\leq p} (m_i+2)
 	 \]
 	By construction, we have $\ell \leq k$, $\abs{\omega}\leq 2\abs{\phi}$.
 	Since by definition, $d(\phi)\geq 1$ and $1\in R(\phi)$, we have $\snec_{2, 1}^G(A_G)\leq \snec_{\phi}^G(A_G)$.
 	By \Cref{fact:upper:cL:H}, we deduce that $\snec_{2, 1}^H(A)\leq 3 \snec_{\phi}^G(A_G)$.
 	From the definition of $\sfN_\acy$ in \Cref{lem:acy:eq} and the upper bounds of \Cref{fact:upper:cL:H}, we deduce that $\sfN_{\acy}(H,A)\leq \sfN_\acy(G,A)^{O(1)}$.
	From these observations, we conclude that \Cref{eq:sizebound:acy} is an upper bound on the size of $\reduce(\cB)$ and \Cref{eq:runtime:acy} is an upper bound on the run time to compute $\reduce(\cB)$.
\end{proof}

\end{document}